\documentclass[aps,prx,superscriptaddress,twocolumn,floatfix,longbibliography]{revtex4-2} 

\usepackage[final]{graphicx}    
\usepackage{soul}
\usepackage[dvipsnames]{xcolor}
\usepackage[utf8]{inputenc}
\usepackage[T1]{fontenc}
\usepackage{natbib}
\usepackage{framed}
\usepackage{amsfonts}
\usepackage{amsmath}
\usepackage{amsthm}
\usepackage{mathtools}
\usepackage{amssymb}
\usepackage{hyperref}
\usepackage{MnSymbol}
\hypersetup{colorlinks=true}

\newtheorem{remark}{Remark}
\newtheorem{property}{Property}
\newtheorem{theorem}{Theorem}	
\newtheorem{definition}{Definition}	

\newcommand{\eg}{\emph{e.g.,~}}
\newcommand{\ie}{\emph{i.e.,~}}
\newcommand{\etal}{\emph{et al.}~}

\newcommand{\nn}{\nonumber}

\begin{document}

\title{Modeling integrated frequency shifters and beam splitters}
\author{Manuel H. Muñoz-Arias}
\email{mhmunoz@sandia.gov}
\affiliation{Quantum Algorithms and Applications Collaboratory, Sandia National Laboratories, Livermore, CA 94550, USA}
\author{Kevin J. Randles}
\affiliation{Center for Computing Research, Quantum Algorithms and Applications Collaboratory, Sandia National Laboratories, Albuquerque, NM 87185, USA}
\author{Nils T. Otterstrom}
\affiliation{Photonic and Phononic Microsystems, Sandia National Laboratories, Albuquerque, New Mexico 87185, USA}
\author{Paul S. Davids}
\affiliation{Photonic and Phononic Microsystems, Sandia National Laboratories, Albuquerque, New Mexico 87185, USA}
\author{Michael Gehl}
\affiliation{Photonic and Phononic Microsystems, Sandia National Laboratories, Albuquerque, New Mexico 87185, USA}
\author{Mohan Sarovar}
\affiliation{Quantum Algorithms and Applications Collaboratory, Sandia National Laboratories, Livermore, CA 94550, USA}
\date{\today}

\begin{abstract}
Photonic quantum computing is a strong contender in the race to fault-tolerance. Recent proposals using qubits encoded in frequency modes promise a large reduction in hardware footprint, and have garnered much attention. In this encoding, linear optics, \ie beam splitters and phase shifters, is necessarily not energy-conserving, and is costly to implement. In this work, we present designs of frequency-mode beam splitters based on modulated arrays of coupled resonators. We develop a methodology to construct their effective transfer matrices based on the SLH formalism for quantum input-output networks. Our methodology is flexible and highly composable, allowing us to define $N$-mode beam splitters either natively based on arrays of $N$-resonators of arbitrary connectivity or as networks of interconnected $l$-mode beam splitters, with $l<N$. We apply our methodology to analyze a two-resonator device, a frequency-domain phase shifter and a Mach--Zehnder interferometer obtained from composing these devices, a four-resonator device, and present a formal no-go theorem on the possibility of natively generating certain $N$-mode frequency-domain beam splitters with arrays of $N$-resonators.
\end{abstract}

\maketitle


%
%

\section{Introduction}
\label{sec:intro}
Integrated photonic platforms for quantum information processing hold great promise due to their industrial scalability. Such photonic platforms are leading contenders for implementing fault-tolerant quantum computing, with several competing approaches that show promise \cite{Bourassa2021blueprintscalable,Alexander2025,AghaeeRad2025}. This new application area for integrated photonics has motivated impressive progress in the field, including the fabrication of ultra low loss waveguides and interconnects~\cite{Chauhan2022ultra,Heck2014ultra,Bauters2011ultra}, novel electro-optic components~\cite{Hu2025review,Hu2025integrated}, and new fabrication processes and tools~\cite{Bose2024anneal}.

The electromagnetic field has many degrees of freedom, and photonic processors require different components depending on which degrees of freedom encode the quantum information to be processed. While spatial and polarization mode encodings have been the most extensively explored, recently, several authors have envisioned photonic quantum computing (PQC) platforms based on encoding qubits into frequency modes~\cite{Lukens2017,kues2017,Karnieli18,Lu2019controlledNOT,Joshi2020,Lu2020arbitrary,Lu2023frequency}. Such encodings have a number of benefits when compared to spatial and polarization encodings, including: (i) potential for increased density of quantum information due to the fact that one waveguide can carry many frequency modes, and (ii) inherent robustness to noise coming from the need for active operations (energy exchange) to change frequencies of optical modes. At the same time, using frequency encoding requires reimagining the implementation of many of the common transformations used in PQC, including the common linear optics transformations of beam splitting and phase shifting. Lukens \etal have developed the framework of a \emph{quantum frequency processor} (QFP) for implementing general linear transformations over frequency modes~\cite{Lukens2017,Lu2018quantum}. While implementations of a QFP using bulk optics equipment have been demonstrated, scalable integrated optics implementations have not. 

Recently, Hu \etal reported on an integrated optics implementation of a core linear optics component, the beam splitter, for frequency modes~\cite{Hu2021,Yeh22,Yeh25}. This device is composed of two coupled ring resonators that are electro-optically modulated, which provides a parametric shift of their resonance frequencies. The ring resonators isolate two frequency modes (within their free spectral range) and the parametric modulation serves to couple these modes. By tuning the modulation strength Hu \etal demonstrated arbitrary splitting of amplitude between the two selected frequency modes. 
Their device has the important benefit of being high efficiency and low loss, when compared to previous demonstrations based on bulk optics, electro-optical devices~\cite{Izutsu1981,Johnson1988,Preble2007,Savchenkov2009,kobayashi2016,Wright2017,Lu2018,Zhu2022spectral}. Other proposals for integrated low-loss and high efficiency frequency beam splitters have appeared recently, they are based on Bragg-scattering four-wave mixing~\citep{Oliver2025}, and acousto-optics~\cite{Lukens2026frodo}. Although promising, in this work we will focus solely on the aforementioned electro-optics devices.

In this work we develop a complete quantum theory for such ring resonator-based frequency-domain beam splitters (RBSs), building off modern input-output theory and the SLH formalism. This theory enables a number of advances in understanding over the treatment in Hu \etal ~\cite{Hu2021}. First, we show how to write transfer matrices describing such ring resonator beamsplitters. Transfer matrices are a standard description of passive linear optics components (\eg directional couplers, bulk optics beamsplitters), which have the critical property of composability that allows for modeling of large linear optics networks by simple multiplication of transfer matrices for individual components in the network. 

Being a resonant device with time-dependent, active modulation it is not immediately clear how the RBS can be modeled using a transfer matrix. We show here that indeed a transfer matrix treatment is possible in certain relevant operating regimes, which greatly alleviates the burden of modeling networks of RBS devices. Secondly, we apply our theory and derived transfer matrices to model devices composed of multiple frequency-domain transformations, and to understand frequency-domain transformations enabled by devices composed of $N>2$ coupled, modulated ring resonators. For these multimode devices, we show some of the interesting multimode transformations possible, the design tradeoffs involved, and also derive some impossibility results. Finally, we analyze the sensitivity of the performance of RBS devices as a function of parameters defining ring and modulation properties.  

The remainder of the manuscript is organized as follows. In Sec.~\ref{sec:background_linear_optics} we comment on the relationship between PQC and quantum input-output networks (QIONs). In Sec.~\ref{sec:slh} we review both the SLH formalism for general QIONs, and the $ABCD$ representation of the equations of motion for linear and passive QIONs, and in Sec.~\ref{sec:background_frequency_modes} we frame the challenge of frequency-domain linear optics in the context of the SLH description of QIONs. In Sec.~\ref{sec:2_reso} we detail the properties of frequency beam splitters generated with arrays of $2$-resonators, and in Sec.~\ref{sec:mach_zender}, as an application of these $2$-resonator devices, we show how to construct a frequency-domain phase shifter and a frequency-domain Mach--Zehnder interferometer. In Sec.~\ref{sec:4_reso} we explore in detail the properties of the frequency beam splitters generated with a $4$-resonator device. In Sec.~\ref{sec:no_go} we discuss the possibility of generating frequency-domain beam splitters with arrays of $N>4$ resonators, and prove a no-go theorem forbidding certain $N$-mode frequency beam splitters in this platform. Finally, in Sec.~\ref{sec:discussion}, we discuss our findings and comment on future directions.
 
%
%

\section{Background: linear optical devices and their modeling}
\label{sec:background_linear_optics}
In photonic quantum computation, propagating optical modes are the carriers of quantum information. A computation is implemented by sending these modes as inputs into an in-place network of linear optical elements, and the result is collected through measurements on the output modes~\cite{Knill2001,Kok2007}. An optical quantum computer is then a Quantum Input-Output Network (QION), with the quantum computation described by the unitary matrix of the network, which relates input modes to output modes. Since the network is composed of linear optical elements, it is standard to view the action of this network in terms of its \emph{transfer matrix}, which is an $N\times N$ matrix, mapping the canonical annihilation operators of the $N$ output modes to linear combinations of annihilation operators of the $N$ input modes \footnote{This is strictly only true for \emph{passive} linear networks that do not contain any squeezing components. For active linear optics networks the transfer matrix is $2N\times 2N$ and prescribes a mapping for annihilation and creation operators.}.

For linear optical networks acting on frequency modes, as mentioned in the introduction, basic components such as beam splitters have more complex implementations, involving time-dependent modulation and resonant structures, which makes transfer matrix descriptions of such components questionable. However, the QION description of such networks is still valid and, as we show below, presents a path towards the construction of effective transfer matrix descriptions in certain parameter regimes. We will briefly review the necessary minimum of QIONs and the associated SLH formalism in the following. An extensive presentation of this formalism can be found in~\cite{combes2017slh} and references therein. 

\subsection{Quantum input-output networks: SLH, ABCD, and transfer matrices}
\label{sec:slh}
A QION is a collection of localized quantum systems, which we will take to be bosonic, interacting with propagating fields via input and output ports. The ultimate goal for the theoretical modeling of these systems is to construct equations of motion for the degrees of freedom (dofs) of the localized systems evolving under their own dynamics and being driven by the input fields, and to write down the input-output relations for the propagating fields. This is done under three approximations: (1) negligible propagation time of the fields between distinct localized systems in the network, (2) weak, linear coupling of localized system dofs and the propagating fields, and (3) the Markov approximation \cite{combes2017slh}. 

The SLH formalism represents a QION by a triple of operators~\footnote{Or arrays of operators, depending on the number of input-output ports}  $(\mathbf{S}, \mathbf{L}, H)$, where $\mathbf{S}$ dictates the scattering of photons between mode $i$ and mode $j$ which is not a direct consequence of energy exchange with the localized system dofs. The set of operators $\mathbf{L}$ governs the energy exchange between the localized system and the propagating modes, and $H$ governs the internal dynamics of the localized system. 

One of the strengths of this formalism is its composability. Given a large QION, one can follow simple rules~\citep{combes2017slh} to arrive at a triple describing the global network from the triples of the individual network components. With such an effective description in hand one can obtain the equations of motion for the internal degrees of freedom of the network, and the network's input-output relation. For a QION composed of passive linear optical components, these equations of motion can be written in a compact form, called the $ABCD$ representation. Consider for simplicity a (passive) linear $N$-port component with $M$ internal dofs described by the triple $(\mathbf{S},\mathbf{L},H)$, then its $ABCD$ representation is given by
\begin{subequations}
\label{eqn:abcd_linear_auto}
\begin{align}
\dot{\mathbf{a}}(t) &= A \mathbf{a}(t) + B \mathbf{b}_{\rm in}(t), \label{eqn:abcd_linear_autoA} \\
\mathbf{b}_{\rm out}(t) &= C \mathbf{a}(t) + D \mathbf{b}_{\rm in}(t), \label{eqn:abcd_linear_autoB}
\end{align}
\end{subequations}
where $\mathbf{a} = (a_1,...,a_M)^T$ is an $(M\times1)$ vector of annihilation operators for the internal degrees of freedom, $\mathbf{b}_{\rm in},\mathbf{b}_{\rm out}$ are $(N \times 1)$ vectors of annihilation operators for the fields of the input and output modes. $A,B,C,D$ are matrices with scalar entries of sizes $(M\times M)$, $(M\times N)$, $(N\times M)$, $(N\times N)$, respectively.

Let us now introduce auxiliary matrices $\Omega$ of size $(M\times M)$ and $\Phi$ of size $(N\times M)$, whose elements $\omega_{lm}$ and $\phi_{lm}$ are defined via $H = \sum_{l,m = 1}^M a^\dagger_l \omega_{lm} a_m$ and $\mathbf{L}_l = \sum_{m = 1}^N \phi_{lm} a_m$. They allow us to establish a direct relation between the elements of the SLH triple and the $A,B,C,D$ matrices, 
\begin{subequations}
\label{eqn:ABCD_def}
\begin{align}
A &= -\frac{1}{2}\Phi^\dagger\Phi - i\Omega, \enspace &B& = -\Phi^\dagger \mathbf{S}, \\
C &= \Phi, \enspace &D& = \mathbf{S}.
\end{align}
\end{subequations}
The linear and time-invariant nature of this mixed system of differential and linear equations ($ABCD$ system), immediately provides us a way of arriving at the transfer matrix of this $N$-port component. By direct Fourier transform of Eq.~(\ref{eqn:abcd_linear_autoA}), and then substituting $\mathbf{a}(\omega)$ in Eq.~(\ref{eqn:abcd_linear_autoB}), we obtain
\begin{equation}
\label{eqn:transfer_func_linear}
\mathbf{b}_{\rm out}(\omega) = \Xi(\omega) \mathbf{b}_{\rm in}(\omega),
\end{equation}
where the frequency-domain transfer function is given by
\begin{equation}
\label{eqn:freq_domain_transfer}
\Xi(\omega) = \left[ C( i\omega \mathbf{I}_{M} - A)^{-1} B + D \right],
\end{equation}
and $\mathbf{I}_{M}$ is the identity matrix of size $(M\times M)$ in the Hilbert space of the localized system internal dofs. Therefore, to write down transfer matrices for specific optical components, it is enough to specify its SLH triple and subsequently the corresponding $ABCD$ matrices.

%
%

\subsection{Linear optics with frequency-modes}
\label{sec:background_frequency_modes}
Given a single input frequency mode, a frequency beam splitter is a linear optical element that distributes the input intensity into two (or more) frequencies. Since the frequency of an optical mode determines its energy, this type of transformation breaks energy conservation, and cannot be achieved with simple energy-conserving optical elements alone. Consequently, we must add or remove energy to/from the propagating mode. The frequency transformation on the propagating mode can be generated using a time-dependent modulation on the localized system and exploiting the scattering consequence of energy exchange between localized system and propagating mode. This modulated system lies outside of the formalism considered in Sec.~\ref{sec:slh}, preventing us from directly applying those techniques to obtain the desired transfer matrices.

To explicitly specify the issue, the addition of the time-dependent modulation results in an $A$ matrix in the $ABCD$ system in Eq.~\eqref{eqn:abcd_linear_auto} to be time-dependent. As a result, the system is no longer time-invariant and one cannot directly derive a transfer function through the Fourier transform. However, as we shall show below, through a careful consideration of relevant parameter regimes and application of the rotating-wave approximation, we can derive \emph{effective} transfer matrices for these devices. We will do this by considering specific example systems in the following two sections, followed by an analysis of the general device consisting of $N$ ring resonators in Sec.~\ref{sec:no_go} and App.~\ref{app:general_transfer_mat}.

%
%
\begin{figure}
    \centering
    \includegraphics[width=\linewidth]{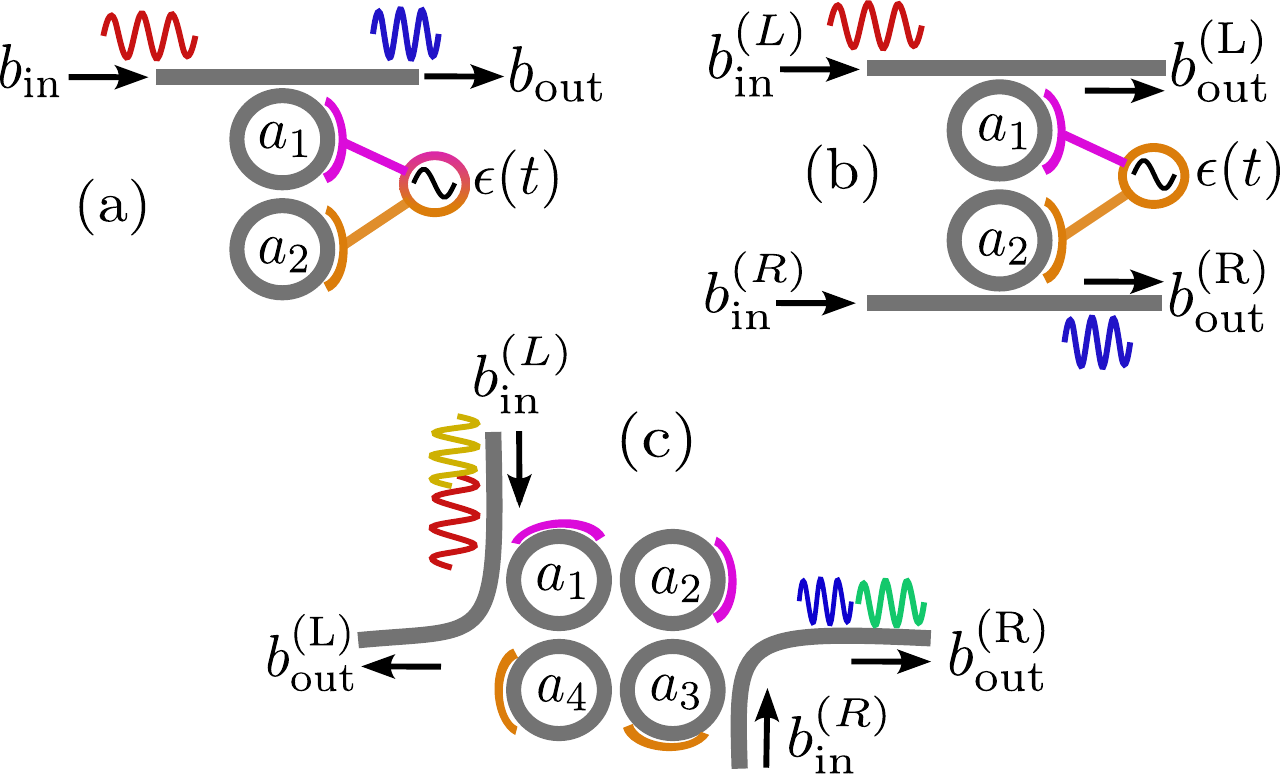}
    \caption{Schematics of the frequency-domain beam splitters based on modulated coupled ring resonators. (a) Two-resonator device which implements $2$-mode frequency beam splitters. (b) Two-resonator device in the two waveguide configuration. The output of interest is the one on the second waveguide, $b_{\rm out}^{\rm (R)}$, which we wish to split by some predetermined ratio. (c) Four-resonator device used to natively implement $4$-mode frequency beam splitters. The second waveguide might or might not be included. In all sketches the resonators have the same central frequency $\omega_0$. In (b) and (c), the superscript (L) and (R) indicate first and second waveguide, respectively. We always assume $b_{\rm in}^{(R)} = 0$.}
    \label{fig:fig_1}
\end{figure}

\section{Frequency beam splitters generated with two resonators}
\label{sec:2_reso}
We begin by considering $2$-mode frequency-domain beam splitters based on ring resonators, as demonstrated by Hu~\etal~\cite{Hu2021}. The setup, shown in Fig.~\ref{fig:fig_1}a,b, is composed of two resonators of central frequencies $\omega_{a_1}$ and $\omega_{a_2}$ which are strongly coupled with strength $u$, both resonators are modulated by a monochromatic tone of amplitude $\epsilon$, frequency $\omega_{\rm d}$, and phase $\phi$. The system is described by the Hamiltonian
\begin{equation}
\label{eqn:hamil_2mode_as}
\begin{split}
H(t) = \sum_{j=1}^2\omega_{a_j} a_j^\dagger a_j + u(a_1^\dagger a_2 + a_1a_2^\dagger) \\ 
+ \epsilon\cos(\omega_{\rm d}t + \phi)(a_1^\dagger a_1 - a_2^\dagger a_2),
\end{split}
\end{equation}
where $a_{1,2}$ ($a_{1,2}^\dagger$) are annihilation (creation) operators for the first and second resonator, respectively. Both resonators are modulated, with the net effect of this modulation being to increase/decrease the frequency of the first/second resonator. Additionally, the system is coupled to either one or two waveguides. The coupling to the first waveguide is via $a_1$ at a rate $\Gamma_{\rm L}$, Fig.~\ref{fig:fig_1}a, and that to the second waveguide, whenever present, is via $a_2$ at a rate $\Gamma_{\rm R}$, Fig.~\ref{fig:fig_1}b. 

\subsection{Single waveguide}
The SLH description of the 2-ring, single waveguide case is $(\mathbf{S}=\mathbf{I}_{1}, \mathbf{L}=\sqrt{\Gamma_L}a_1, H=H(t))$, where $H(t)$ is the Hamiltonian in Eq.~\eqref{eqn:hamil_2mode_as}. As discussed above, since $H$ is time-dependent, the system matrix $A$ in the corresponding $ABCD$ description will be time-dependent, and a transfer matrix is not derivable in a straightforward manner. However, in the following we show how a sequence of rotating wave approximations (RWAs) can be employed to recover a transfer matrix description of this system in the experimentally relevant parameter regime.

We focus on the case of identical resonators, i.e, $\omega_{a_1} = \omega_{a_2} = \omega_0$. With this assumption the normal modes of the unmodulated system correspond to the asymmetric, $c_1 = \frac{a_1 - a_2}{\sqrt{2}}$, and symmetric, $c_2 = \frac{a_1 + a_2}{\sqrt{2}}$, combinations of the resonators annihilation operators, with frequencies $\omega_1 = \omega_0 - u$ and $\omega_2 = \omega_0 + u$, respectively. In terms of these normal mode operators the Hamiltonian takes the form
\begin{equation}
\label{eqn:hamil_2mode_cs}
H(t) = \sum_{j=1}^2 \omega_j c_j^\dagger c_j + \epsilon\cos(\omega_{\rm d} t + \phi)(c_1^\dagger c_2 + c_1 c_2^\dagger).
\end{equation}
The aim of the modulation is to drive transitions between the two normal modes in the 2-ring system, which are separated by $\omega_2-\omega_1 \equiv \Delta_{12}=2u$. Therefore, $\omega_{\rm d}$ should be close to resonant with the splitting. Take $\omega_{\rm d} = \Delta_{12} + \delta$, where $\delta$ is a detuning. Then, rewriting the Hamiltonian in a rotating frame with $c_{1,2} \rightarrow \tilde{c}_{1,2}=c_{1,2}e^{-i(\omega_{1,2} \pm \frac{\delta}{2}) t}$, we get
\begin{equation}
\label{eqn:rwa_hamil_2ring}
\tilde{H} = -\frac{\delta}{2} \tilde{c}_1^\dagger \tilde{c}_1 + \frac{\delta}{2} \tilde{c}_2^\dagger \tilde{c}_2 + \frac{\epsilon}{2} (e^{i\phi} \tilde{c}_1^\dagger \tilde{c}_2 + e^{-i\phi} \tilde{c}_1 \tilde{c}_2^\dagger ),
\end{equation}
where we have dropped fast rotating terms oscillating at frequencies $\geq 2\omega_d$. This frame change and RWA removes the time-dependence in the system's Hamiltonian. In the following, we will assume $\delta=0$ for simplicity.

Moreover, there is another RWA we can make to simplify the system further. This relates to the coupling of the normal modes to the waveguide. The $\mathbf{L}$ operator for this system in terms of the normal mode operators is $\mathbf{L} = \sqrt{\Gamma_{\rm L}}a_1 = \sqrt{\gamma_{\rm L}}c_1 + \sqrt{\gamma_{\rm L}}c_2$, with $\gamma_{\rm L} = \frac{\Gamma_{\rm L}}{2}$ the ``effective'' coupling rate of each normal mode to the waveguide. 
We will assume $\Gamma_{\rm L} \ll \Delta_{12}$, which implies well-resolved normal modes with narrow enough linewidth to be well-separated. This regime has an implication for the input-output theory: while there is a waveguide with a single spatial mode carrying multiple frequencies, the only way photons can convert from $\omega_1$ to $\omega_2$ and \emph{vice versa} is through the modulated coherent coupling. Outside this regime (\emph{i.e.,} if $\Gamma_{\rm L} \sim \Delta_{12}$) spectral overlap of the cavity modes can cause incoherent scattering between these frequencies, and this is not the regime we would want to operate in for clean beam-splitter behavior. As a result, the input-output modeling of the 2-ring system should consider two independent input modes and modify $\mathbf{L}$ and $\mathbf{S}$ to
\begin{align}
    \mathbf{L} &= \begin{pmatrix} 
            \sqrt{\gamma_{\rm L}}~ \tilde{c}_1 \\
            \sqrt{\gamma_{\rm L}}~ \tilde{c}_2 
            \end{pmatrix},
         \label{eq:L_up} \\
    \mathbf{S} &= \mathbf{I}_{2}. \label{eq:S_up}   
\end{align}
Another justification for this modification can be seen from writing the Heisenberg equations of motion for the normal modes in the rotating frame (assuming $\delta=0$),
\begin{align}
\dot{\tilde{c}}_1 = -i\epsilon e^{i\phi}\tilde{c}_2 - \frac{\gamma_{\rm L}}{2}\tilde{c}_2 + \sqrt{\gamma_{\rm L}} b_{\rm in}(t),
\end{align}
and similarly for $\dot{\tilde{c}}_2$. It is clear that when this equation is integrated, only frequency content around $\omega=0$ in $b_{\rm in}(t)$ will appreciably influence $\tilde{c}_1$. Or in the original lab frame, only frequency content around $\omega=\omega_1$. Similarly, $\tilde{c}_2$ is only driven by frequency content of $b_{\rm in}(t)$ around $\omega_2$. As a result, we can treat the input field as consisting of independent frequency modes at $\omega_1$ and $\omega_2$ and expand the input (and output) field into a vector $\mathbf{b}_{\rm in} = (b_{\rm in}(\omega_1), b_{\rm in}(\omega_2))^{T}$, where the arguments indicate the frequency mode carried on the spatial waveguide. As a result, the $\mathbf{L}$ and $\mathbf{S}$ operators in the SLH description are updated to the ones in Eqs.~\eqref{eq:L_up} and \eqref{eq:S_up}.

Finally, we assume that both ring resonators suffer internal losses at rate $\kappa_{\rm int}$. The effects of this can be captured by introducing a fictitious input-output mode with vacuum input and corresponding expansion of the $\mathbf{L}$ operator, but the net effect of doing this at the transfer matrix level is just to modify the linewidth of each normal mode, and therefore, we shall phenomenologically capture this effect below without going through the trouble of expanding $\mathbf{L}$.

At this point we have a time-independent SLH and equivalent $ABCD$ system that can be solved via Fourier transform, to obtain the effective transfer matrix
\begin{equation}
\label{eqn:transfer_mat_2reso}
\Xi(\epsilon, \gamma_{\rm L}, \kappa_{\rm int}) = \begin{pmatrix}
1 - \frac{2\kappa\gamma_{\rm L}}{\epsilon^2 + \kappa^2} && i \frac{2\epsilon\gamma_{\rm L}}{\epsilon^2 + \kappa^2} e^{i\phi} \\ 
i \frac{2\epsilon\gamma_{\rm L}}{\epsilon^2 + \kappa^2} e^{-i\phi}  && 1 - \frac{2\kappa\gamma_{\rm L}}{\epsilon^2 + \kappa^2}
\end{pmatrix},
\end{equation}
where $\kappa = \gamma_{\rm L} + \kappa_{\rm int}$ is the total linewidth of each normal mode, and the intrinsic losses were added into the $A$ matrix phenomenologically.

As we increase $\epsilon$, Eq.~(\ref{eqn:transfer_mat_2reso}) densely covers all the splitting rations from $100$-$0$ to $0$-$100$, and back to $100$-$0$. In particular, if $0\le 1 - R \le1$ is the probability of staying in the same frequency mode, we can cast $\Xi$ in the more interpretable form
\begin{equation}
\label{eqn:transfer_mat_R}
\Xi(R, \alpha_{\rm L}) = \begin{pmatrix}
    \sqrt{1-R} & i\sqrt{R}e^{i\phi} \\ 
    i\sqrt{R}e^{-i\phi} &  \sqrt{1-R}
\end{pmatrix} \mathcal{K}_\pm(\alpha_{\rm L}, R),
\end{equation}
where $\mathcal{K}_\pm(\alpha_{\rm L}, R)$ is the transmission amplitude parameter given by
\begin{equation}
\mathcal{K}_\pm(\alpha_{\rm L}, R) = \frac{2\sqrt{1-R} \pm \sqrt{\alpha_{\rm L}^2  - 4R}}{\alpha_{\rm L} + 2},
\end{equation}
which dictates the total loss of the beam splitter via $\mathcal{L}_\pm(\alpha_{\rm L}, R) = 1-|\mathcal{K}_\pm(\alpha_{\rm L}, R)|^2$. In analogy with cavity QED systems, we have introduced the ``cooperativity'' parameter 
\begin{align}
\alpha_{\rm L} = \frac{\Gamma_{\rm L}}{\kappa_{\rm int}},
\end{align}
which will facilitate the comparison between devices obtained from resonator arrays of different sizes.

For a given value of $R$, there are two values of modulation amplitude, $\epsilon_{\rm R}^{\pm}$, such that $\Xi(\epsilon_{\rm R}^\pm, \gamma_{\rm L}, \kappa_{\rm int}) = \Xi(R, \alpha_{\rm L})$. They can be obtained by solving for $\epsilon$ in $\frac{|\Xi_{11}|^2}{|\Xi_{11}|^2 + |\Xi_{21}|^2} = 1-R$, yielding
\begin{equation}
\label{eqn:epsilon_R}
\epsilon_{\rm R}^{\pm} = \left|\sqrt{\frac{1-R}{R}}\gamma_{\rm L} \pm \sqrt{\frac{\gamma_{\rm L}^2}{R} - \kappa_{\rm int}^2} \right|.
\end{equation}
Thus, as we increase $\epsilon$ from $0$, our device implements all the beam splitters with ratios from $100$-$0$ to $0$-$100$, with corresponding transmission amplitude parameter $\mathcal{K}_-$. At the $0$-$100$ point ($R=1$), if we keep increasing $\epsilon$, the device implements beam splitters covering all possible splitting ratios from $0$-$100$ to $100$-$0$, with corresponding transmission amplitude parameter $\mathcal{K}_+$. 

For the remainder of this section we focus on the particular ratios $0$-$100$ ($R=1$) and $50$-$50$ ($R=1/2$), and characterize the resulting frequency beam splitters.

\subsubsection{\texorpdfstring{$0$}{\textit{0}}-\texorpdfstring{$100$}{\textit{100}} frequency beam splitter}
Ideal frequency shifting is achieved whenever the output carries zero component of the input mode, \ie if the input is either at frequency $\omega_1$ or $\omega_2$, the output \emph{must} only be at frequency $\omega_2$ or $\omega_1$, respectively. The modulation amplitude satisfying this condition is
\begin{equation}
\label{eqn:gcc_2reso_1wv}
\epsilon_{\rm GCC} = \sqrt{\gamma_{\rm L}^2 - \kappa_{\rm int}^2},
\end{equation}
which is the ``generalized critical coupling'' (GCC) condition of Ref.~\cite{Hu2021}. At $\epsilon_{\rm GCC}$ the transfer matrix can be written as
\begin{equation}
\label{eqn:transfer_mat_gcc}
\Xi(\epsilon_{\rm GCC}) = \begin{pmatrix}
0 & ie^{i\phi} \\ 
ie^{-i\phi} & 0 \\
\end{pmatrix} \mathcal{K}_{\rm GCC},
\end{equation} 
a $0$-$100$ frequency beamspliter with nonunity transmission amplitude parameter, $\mathcal{K}_{\rm GCC} = \sqrt{\frac{\gamma_{\rm L} - \kappa_{\rm int}}{\gamma_{\rm L} + \kappa_{\rm int}}} = \sqrt{\frac{\alpha_{\rm L} - 2}{\alpha_{\rm L} + 2}}$. The total intensity loss, $\mathcal{L}(\epsilon_{\rm GCC}) = 1 - |\mathcal{K}_{\rm GCC}|^2$, is 
\begin{equation}
\label{eqn:loss_gcc_2mode}
\mathcal{L}(\epsilon_{\rm GCC}) = \frac{2\kappa_{\rm int}}{\gamma_{\rm L} + \kappa_{\rm int}} = \frac{4}{\alpha_L + 2}.
\end{equation}
Two observations are in order here: (1) Lossless frequency conversion is achieved at GCC only when $\kappa_{\rm int} = 0$. (2) At GCC, we always have $\Xi_{11}(\epsilon_{\rm GCC}) = \Xi_{22}(\epsilon_{\rm GCC}) = 0$, which is a global minimum. However, the global maximum of $|\Xi_{12}(\epsilon_{\rm GCC})|^2$ is located at $\epsilon = \gamma_{\rm L} + \kappa_{\rm int}$, which coincides with the GCC only when $\kappa_{\rm int}=0$ (see the peak of the green line in Fig.~\ref{fig:fig_2}).

\subsubsection{\texorpdfstring{$50$}{\textit{50}}-\texorpdfstring{$50$}{\textit{50}} frequency beam splitter}
\begin{figure}
    \centering
    \includegraphics[width=\linewidth]{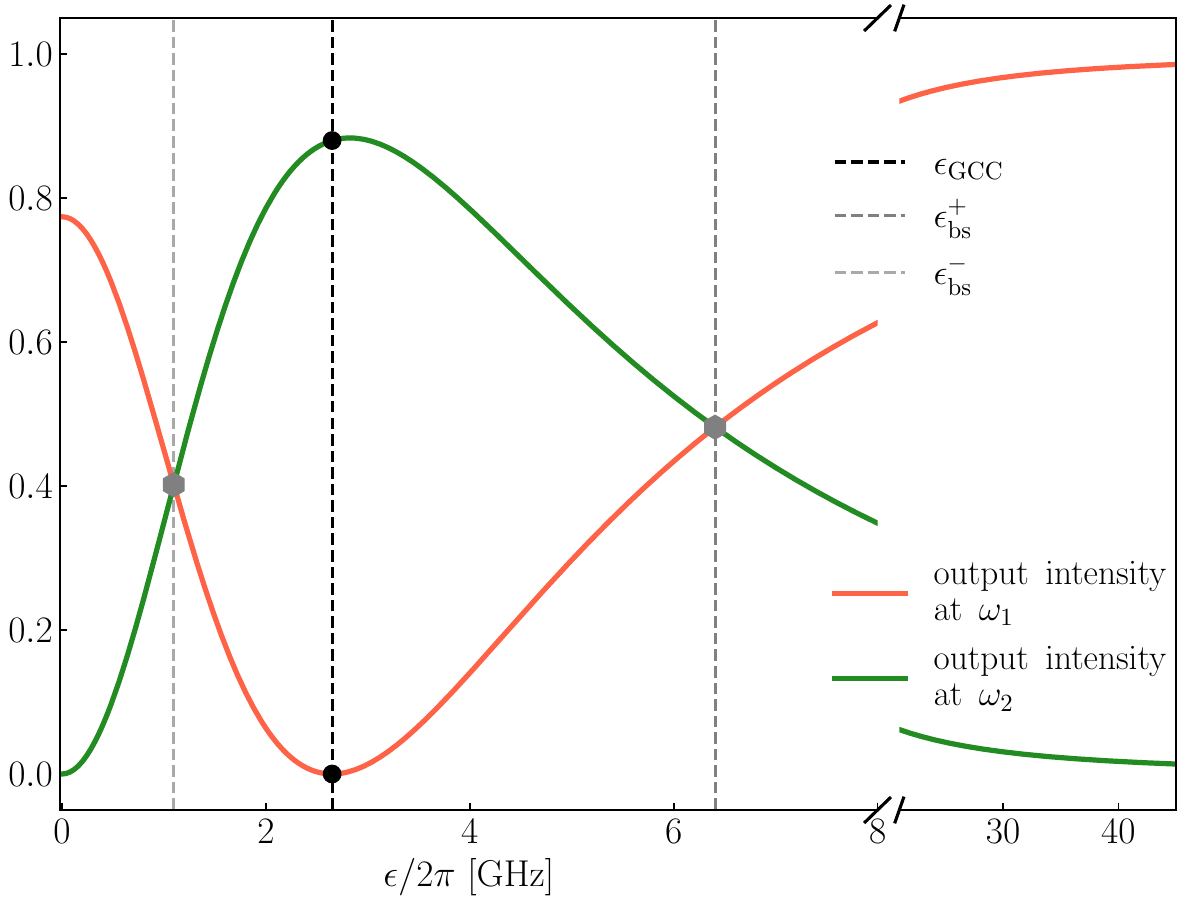}
    \caption{Norm squared entries of the transfer matrix for the two resonator device coupled to a single waveguide as a function of the modulation amplitude $\epsilon$, the input is at $\omega_1$. The dashed black line indicates the $0$-$100$ frequency beam splitter point, $\epsilon_{\rm GCC}$, while the two gray dashed lines indicate the $50$-$50$ frequency beam splitter points, $\epsilon_{\rm bs}^{\pm}$. The extended axis shows the asymptotic behvior of the transfer matrix entries as $\epsilon\to\infty$. Parameters: $\Delta_{12}/2\pi = 28.2~{\rm GHz}$, $\Gamma_{\rm L}/2\pi = 5.31~{\rm GHz}$, $\kappa_{\rm int}/2\pi = 0.17~{\rm GHz}$, leading to $\alpha_{\rm L}\sim 30$, which correspond to those of device (c) in Ref.~\cite{Hu2021}.}
    \label{fig:fig_2}
\end{figure}
One obtains a $50$-$50$ frequency-mode beam splitter when the output contains both normal modes, $c_1$ and $c_2,$ with equal intensity. Interestingly, the device features two distinct modulation amplitudes satisfying this condition
\begin{equation}
\label{eqn:bs_2reso_1wv}
\epsilon_{\rm bs}^{\pm} = |\gamma_{\rm L} \pm \sqrt{2\gamma_{\rm L}^2 - \kappa_{\rm int}^2}|.
\end{equation}
 See grey dashed lines in Fig.~\ref{fig:fig_2}. At $\epsilon_{\rm bs}^{\pm}$, the transfer matrix is
\begin{equation}
\label{eqn:transfer_mat_bs_2reso}
\Xi(\epsilon_{\rm bs}^{\pm}) = \begin{pmatrix}
    \frac{1}{\sqrt{2}} & \frac{i}{\sqrt{2}}e^{i\phi} \\ 
    \frac{i}{\sqrt{2}}e^{-i\phi} & \frac{1}{\sqrt{2}}
\end{pmatrix} \mathcal{K}_{\rm bs}^\pm,
\end{equation}
a $50$-$50$ frequency beam splitter with nonunity transmision amplitude parameter, $\mathcal{K}_{\rm bs}^\pm = \frac{\sqrt{2} \pm \sqrt{\alpha_{\rm L}^2  -2}}{\alpha_{\rm L} + 2}$.

Since at both operating points the system exhibits the target behavior, and, $\epsilon_{\rm bs}^{-}<\epsilon_{\rm bs}^+$, operating at $\epsilon_{\rm bs}^{-}$ is preferred, as it demands a weaker modulation amplitude. Specifically, we find the ratio between the modulation amplitudes defining these two operating points to be 
\begin{equation}
\frac{\epsilon_{\rm bs}^{+}}{\epsilon_{\rm bs}^{-}} = \frac{(\gamma_{\rm L} + \sqrt{2\gamma_{\rm L}^2 - \kappa_{\rm int}^2})^2}{|\kappa_{\rm int}^2 - \gamma_{\rm L}^2|} = \frac{(\alpha_{\rm L} + \sqrt{2}\sqrt{\alpha_{\rm L}^2 - 2})^2}{|4 - \alpha_{\rm L}^2|},
\end{equation}
which in the low loss limit, $\alpha_{\rm L}\gg1$, becomes $\frac{\epsilon_{\rm bs}^{+}}{\epsilon_{\rm bs}^{-}} \approx 3+2\sqrt{2} \approx 5.8$. Thus, operating at $\epsilon_{\rm bs}^{+}$ requires a modulation almost six times stronger.

The total intensity loss of the device at these operating points is
\begin{equation}
\label{eqn:loss_bs_2reso_1wv}
\mathcal{L}(\epsilon_{\rm bs}^{\pm}) = 1-|\mathcal{K}_{\rm bs}^\pm|^2 = \frac{4\sqrt{2}(\alpha_{\rm L} + 1) \mp 4 \sqrt{\alpha_{\rm L}^2 - 2}}{\sqrt{2}(\alpha_{\rm L} + 2)^2},
\end{equation}
thus, for all values of $\alpha_{\rm L}$, the beam splitter at $\epsilon_{\rm bs}^{+}$ suffers less loss than the beam splitter at $\epsilon_{\rm bs}^{-}$. Specifically, when $\alpha_{\rm L}\gg1$, we find $\frac{\mathcal{L}(\epsilon_{\rm bs}^+)}{\mathcal{L}(\epsilon_{\rm bs}^-)} \to \frac{1}{3 + 2\sqrt{2}}\approx 0.17$, meaning that at $\epsilon_{\rm bs}^{+}$ the device experiences about a sixth of the losses it experiences at $\epsilon_{\rm bs}^{-}$, in this low loss limit. Importantly, at $\epsilon_{\rm bs}^-$, to have a $50$-$50$ beam splitter which ``breaks even'', \ie we loose less than half of the input, one must have $\alpha_{\rm L}\ge 10$.  

In Fig.~\ref{fig:fig_2} we show the behavior of both $|\Xi_{11}(\epsilon)|^2$ and $|\Xi_{21}(\epsilon)|^2$ for the parameters of device (c) in Ref.~\cite{Hu2021}. That is, $\Delta_{12}/2\pi = 28.2~{\rm GHz}$, $\Gamma_{\rm L}/2\pi = 5.31~{\rm GHz}$, and $\kappa_{\rm int}/2\pi = 0.17~{\rm GHz}$, leading to $\alpha_{\rm L}\sim30$. The experiment of Ref.~\cite{Hu2021} explores amplitudes up to $\epsilon_{\rm GCC}$ (black dashed line in Fig.~\ref{fig:fig_2}). Hence, the second $50$-$50$ beam splitting point at $\epsilon_{\rm bs}^+$ constitutes an unexplored regime of the device. However, it likely requires modulation amplitudes beyond what can be generated in current state-of-the-art setups.

In the limit $\epsilon\gg1$ the transfer matrix behaves as an identity element, illustrated in the extended axis of Fig.~\ref{fig:fig_2} (see App.~\ref{app:2_reso} for further details). In presence of the modulation, the frequency levels of the coupled ring system undergo a second-order splitting~\cite{zhang2019electronically}, $\omega_{1,2}\to\omega_{1,2, \pm}$. Each pair of new peaks moves away from the parent frequency as $\epsilon$ increases, as such, for large enough $\epsilon$ the system will no longer feature a resonance at the chosen input frequency, with the input light no longer coupling to the system. Importantly, for the effective transfer matrix model to remain valid we require $\omega_{1,-}$ to not couple to $\omega_{2,+}$, we find this to be guaranteed as long as $\epsilon<\Delta_{1,2}=2u$ (see App.~\ref{app:2_reso} for details). 

\subsection{Two waveguides}
Similar to the single waveguide case, the effective transfer matrix for the system is obtained starting from the SLH triple. Clearly $H=H(t)$ in Eq.~(\ref{eqn:hamil_2mode_cs}), and $\mathbf{L}=(\sqrt{\Gamma_{\rm L}}a_1, \sqrt{\Gamma_{\rm R}}a_2)^T$. Since each waveguide feeds both normal modes, we assign individual ``virtual'' input-output ports to $c_1$ and $c_2$, and promote $\mathbf{L}$ to $\mathbf{L} = (\sqrt{\gamma_{\rm L}}\tilde{c}_1, -\sqrt{\gamma_{\rm R}}\tilde{c}_1, \sqrt{\gamma_{\rm L}}\tilde{c}_2, \sqrt{\gamma_{\rm R}}\tilde{c}_2)^T$, where $\gamma_{\rm L,R} = \frac{\Gamma_{\rm L,R}}{2}$ are the ``effective'' coupling rates of the normal modes to each waveguide. The $\mathbf{S}$ operator is $\mathbf{S}=\mathbf{I}_{4}$. The inputs on each of these virtual ports are $b_{\rm in}^{(\rm L)}(t)$ and $b_{\rm in}^{(\rm R)}(t)$ for the first and second waveguides, respectively.

From the SLH triple we construct the time-dependent $ABCD$ system, then, by going into the rotating frames of the slow components of the normal modes, invoking the RWA to discard oscillatory terms at $\pm2\omega_{\rm d}$, and neglecting the inputs on the slow components which are not resonant, we recover a time-independent system which can be solved via Fourier transform. When the input is resonant with either $\omega_1$ or $\omega_2$, and the modulation is resonant with the level splitting, $\Delta_{12}$, the transfer matrix for this system is
\begin{widetext}
\begin{equation}
\label{eqn:transfer_mat_2reso_2wv}
\Xi  (\epsilon, \gamma_{\rm L}, \gamma_{\rm R}, \kappa_{\rm int}) = \\
 \begin{pmatrix}
1 - \frac{2\kappa\gamma_{\rm L}}{\kappa^2 + \epsilon^2}  &  \frac{2\kappa\sqrt{\gamma_{\rm L}\gamma_{\rm R}}}{\kappa^2 + \epsilon^2} & i\frac{2\epsilon\gamma_{\rm L}}{\kappa^2 + \epsilon^2}e^{i\phi} & i\frac{2\epsilon\sqrt{\gamma_{\rm L}\gamma_{\rm R}}}{\kappa^2 + \epsilon^2}e^{i\phi} \\
 \frac{2\kappa\sqrt{\gamma_{\rm L}\gamma_{\rm R}}}{\kappa^2 + \epsilon^2} & 1 - \frac{2\kappa\gamma_{\rm R}}{\kappa^2 + \epsilon^2} & -i\frac{2\epsilon\sqrt{\gamma_{\rm L}\gamma_{\rm R}}}{\kappa^2 + \epsilon^2}e^{i\phi} & -i\frac{2\epsilon\gamma_{\rm R}}{\kappa^2 + \epsilon^2}e^{i\phi} \\
i\frac{2\epsilon\gamma_{\rm L}}{\kappa^2 + \epsilon^2}e^{-i\phi} & -i\frac{2\epsilon\sqrt{\gamma_{\rm L}\gamma_{\rm R}}}{\kappa^2 + \epsilon^2}e^{-i\phi} & 1 - \frac{2\kappa\gamma_{\rm L}}{\kappa^2 + \epsilon^2}  & - \frac{2\kappa\sqrt{\gamma_{\rm L}\gamma_{\rm R}}}{\kappa^2 + \epsilon^2} \\
i\frac{2\epsilon\sqrt{\gamma_{\rm L}\gamma_{\rm R}}}{\kappa^2 + \epsilon^2}e^{-i\phi} & -i\frac{2\epsilon\gamma_{\rm R}}{\kappa^2 + \epsilon^2}e^{-i\phi} & - \frac{2\kappa\sqrt{\gamma_{\rm L}\gamma_{\rm R}}}{\kappa^2 + \epsilon^2} & 1 - \frac{2\kappa\gamma_{\rm R}}{\kappa^2 + \epsilon^2}
\end{pmatrix},
\end{equation}
\end{widetext}
where now the total mode linewith, $\kappa$, is given by $\kappa = \gamma_{\rm L} + \gamma_{\rm R} + \kappa_{\rm int}$, and the intrinsic losses where included phenomenologically via the $A$ matrix.

\subsubsection{Absence of \texorpdfstring{$0$}{\textit{0}}-\texorpdfstring{$100$}{\textit{100}} frequency beam splitter}
Consider, without loss of generality, an input on the first waveguide resonant with $\omega_1$. Then, we only need to look at the first column of $\Xi(\epsilon)$. In particular, ideal frequency shifting is observed whenever $\Xi_{21}(\epsilon) = 0$ and  $\Xi_{41} (\epsilon) \ne 0$, and we do not impose any condition on the values of $\Xi_{11}(\epsilon)$ or $\Xi_{31}(\epsilon)$, since, in principle, we are not interested in the contents of the output of the first waveguide. A simple inspection of the transfer matrix tells us that no value of $\epsilon$ satisfies: $\Xi_{21}(\epsilon) = 0$ and  $\Xi_{41} (\epsilon) \ne 0$. Thus, the second waveguide cannot be use to define a frequency shifter.

However, there is a value of modulation amplitude at which the output on the first waveguide is frequency-shifted. Solving $\Xi_{11}(\epsilon) = 0$, we find $\epsilon_{\rm GCC} = \sqrt{\gamma_{\rm L}^2 - (\gamma_{\rm R} + \kappa_{\rm int})^2}$. This operating point does not lead to the desired output on the second waveguide, since $\Xi_{21}(\epsilon_{\rm GCC}) = \sqrt{\frac{\gamma_{\rm R}}{\gamma_{\rm L}}}\ne 0$, and $\Xi_{41}(\epsilon_{\rm GCC})\ne 0$. Thus, even when the output on the first waveguide is successfully frequency-shifted, the output on the second waveguide carries both frequency modes. 

\begin{figure}
    \centering
    \includegraphics[width=\linewidth]{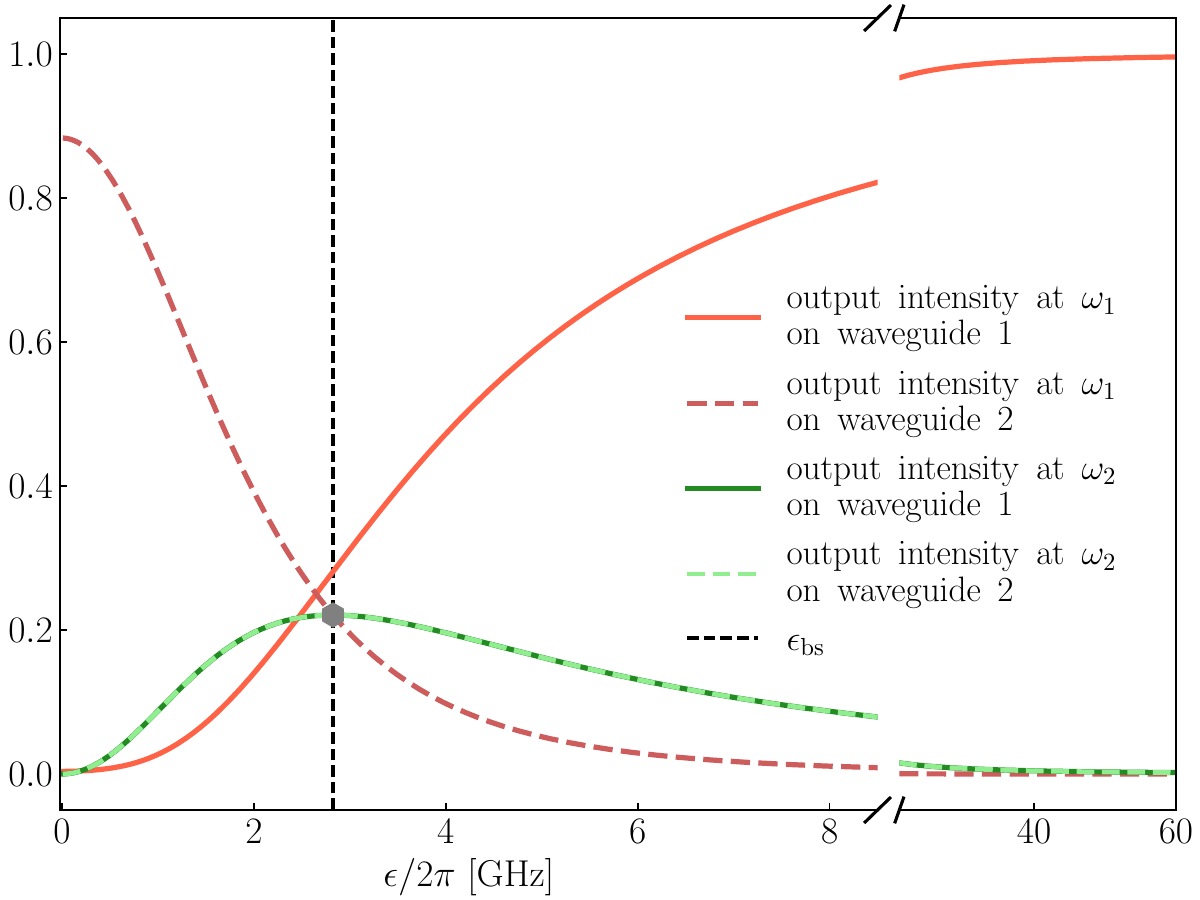}
    \caption{Norm squared entries of the transfer matrix for the two resonator device in the two waveguide configuration as function of the modulation amplitude $\epsilon$. The input is taken at $\omega_1$. The dashed black line indicates the $50$-$50$ frequency beam splitter working point, $\epsilon_{\rm bs}^{(\rm R)}$, for the output of the second waveguide. The extended axis shows the asymptotic behavior of the transfer matrix entries as $\epsilon\to\infty$. Parameters are as in in Fig.~\ref{fig:fig_2} and $\Gamma_{\rm R} = \Gamma_{\rm L}$.}
    \label{fig:fig_3}
\end{figure}

\subsubsection{\texorpdfstring{$50$}{\textit{50}}-\texorpdfstring{$50$}{\textit{50}} frequency beam splitter}
A $50$-$50$ frequency-domain beam splitter on the output of the right-waveguide is obtained whenever $|\Xi_{21}(\epsilon)|^2 = |\Xi_{41}(\epsilon)|^2$. The modulation amplitude satisfying this condition is 
\begin{equation}
\label{eqn:epsilon_bs_2vw}
\epsilon_{\rm bs}^{(\rm R)} = \kappa = \gamma_{\rm L} + \gamma_{\rm R} + \kappa_{\rm int},
\end{equation}
that is, we must modulate with an amplitude equal to the total mode linewidth. At this operating point, the total output intensity on the second waveguide is
\begin{equation}
\label{eqn:intensity_output_right_2wv}
\mathcal{I}_{\rm R}(\epsilon_{\rm bs}^{(\rm R)}) = \frac{2\gamma_{\rm L}\gamma_{\rm R}}{(\gamma_{\rm L} + \gamma_{\rm R} + \kappa_{\rm int})^2} = \frac{2\alpha_{\rm L} \alpha_{\rm R}}{(\alpha_{\rm L} + \alpha_{\rm R} + 2)^2},
\end{equation}
where $\alpha_{\rm L,R} = \Gamma_{\rm L,R}/\kappa_{\rm int}$ are the cooperativities for the first and second waveguides, respectively. In absence of internal loss we can write Eq.~(\ref{eqn:intensity_output_right_2wv}) as $\mathcal{I}_R(\epsilon_{\rm bs}^{(\rm R)}) = \frac{2\beta}{(\beta+2)^2}$ where $\beta = \frac{\gamma_{\rm L}}{\gamma_{\rm R}}$. This quantity has a global maximum of $\mathcal{I}_R(\epsilon_{\rm bs}^{(\rm R)}) = 1/2$ at $\beta=1$. Hence, under ideal conditions, the output intensity on the second waveguide is at best half of the input. This is achieved with symmetric waveguide couplings, $\Gamma_{\rm L} = \Gamma_{\rm R} = \Gamma$. 

The total output intensity on the first waveguide at $\epsilon_{\rm bs}^{(\rm R)}$ is 
\begin{equation}
\label{eqn:intensity_output_left_2wv}
\mathcal{I}_{\rm L}(\epsilon_{\rm bs}^{(\rm R)}) = \frac{(\gamma_{\rm R} + \kappa_{\rm int})^2 + \gamma_{\rm L}^2}{(\gamma_{\rm L} + \gamma_{\rm R} + \kappa_{\rm int})^2} = \frac{(\alpha_{\rm R} + 2)^2 + \alpha_{\rm L}^2}{(\alpha_{\rm R} + \alpha_{\rm L} + 2)^2},
\end{equation}
which for vanishing internal loss, $\alpha_{\rm L,R}\gg1$, and with symmetric waveguide couplings becomes $\mathcal{I}_{\rm L}(\epsilon_{\rm bs}^{(\rm R)}) \to 1/2$. Therefore, a lossless device with symmetric waveguide couplings at $\epsilon_{\rm bs}^{\rm (R)}$ implements both a spatial $50$-$50$ beam splitter between the two waveguides, and a frequency $50$-$50$ beam splitter on each waveguide. Nonzero losses change this scenario in two ways: (1) the first waveguide carries a larger proportion of the total output intensity, as can be inferred by comparing Eq.~(\ref{eqn:intensity_output_left_2wv}) with Eq.~(\ref{eqn:intensity_output_right_2wv}), and (2) only the output of the second waveguide is equally split between the two normal modes being mixed. This is illustrated in Fig.~\ref{fig:fig_3}, where the dashed lines show the behavior of the entries of $\Xi(\epsilon)$ controlling the output on the second waveguide. 

The total loss of the device at the $50$-$50$ beam splitting point of the second waveguide is $L(\epsilon_{\rm bs}^{(\rm R)}) = 1 - \mathcal{I}_R(\epsilon_{\rm bs}^{(\rm R)}) - \mathcal{I}_L(\epsilon_{\rm bs}^{(\rm R)})$, which we find to be
\begin{equation}
\label{eqn:loss_2resos_2wv}
L(\epsilon_{\rm bs}^{(\rm R)}) = \frac{2\gamma_L \kappa_{\rm int}}{(\gamma_R + \gamma_L+\kappa_{\rm int})^2} = \frac{4\alpha_L}{(\alpha_R + \alpha_L + 2)^2},
\end{equation}
and with equal couplings to the waveguides we get $L(\epsilon_{\rm bs}^{(\rm R)}) = \frac{\alpha}{(\alpha + 1)^2}$, where $\alpha = \alpha_{\rm L} = \alpha_{\rm R}$, which in the limit of $\alpha\gg1$ goes as $1/\alpha$. In particular when $\alpha = 30$ we get a total intensity loss of about $3\%$.

Finally we look at the behavior of the norm squared entries of $\Xi$ in the limits $\epsilon\to0$ and $\epsilon\gg1$. In the limit of vanishing modulation amplitude we can write $\Xi_{11}(\epsilon\to0)\approx 1 - \frac{2\gamma_L}{\kappa}$ and $\Xi_{31}(\epsilon\to0) = \Xi_{41}(\epsilon\to0)\approx0$, most output goes out the left-waveguide and a small fraction goes out the right-waveguide. On both of these, the output light carries only the $c_1$ mode (which in our case was the mode resonant with the input light). In general, the output light will carry only the normal mode resonant with the input light. In the limit of a large modulation amplitude, all off-diagonal entries of the transfer matrix $\Xi_{ij, i\ne j}(\epsilon\gg1)\to 0$. Thus, the device behaves as an identity element, which we illustrate in the extended axis of Fig.~\ref{fig:fig_3}. 

\begin{figure}
    \centering
    \includegraphics[width=\linewidth]{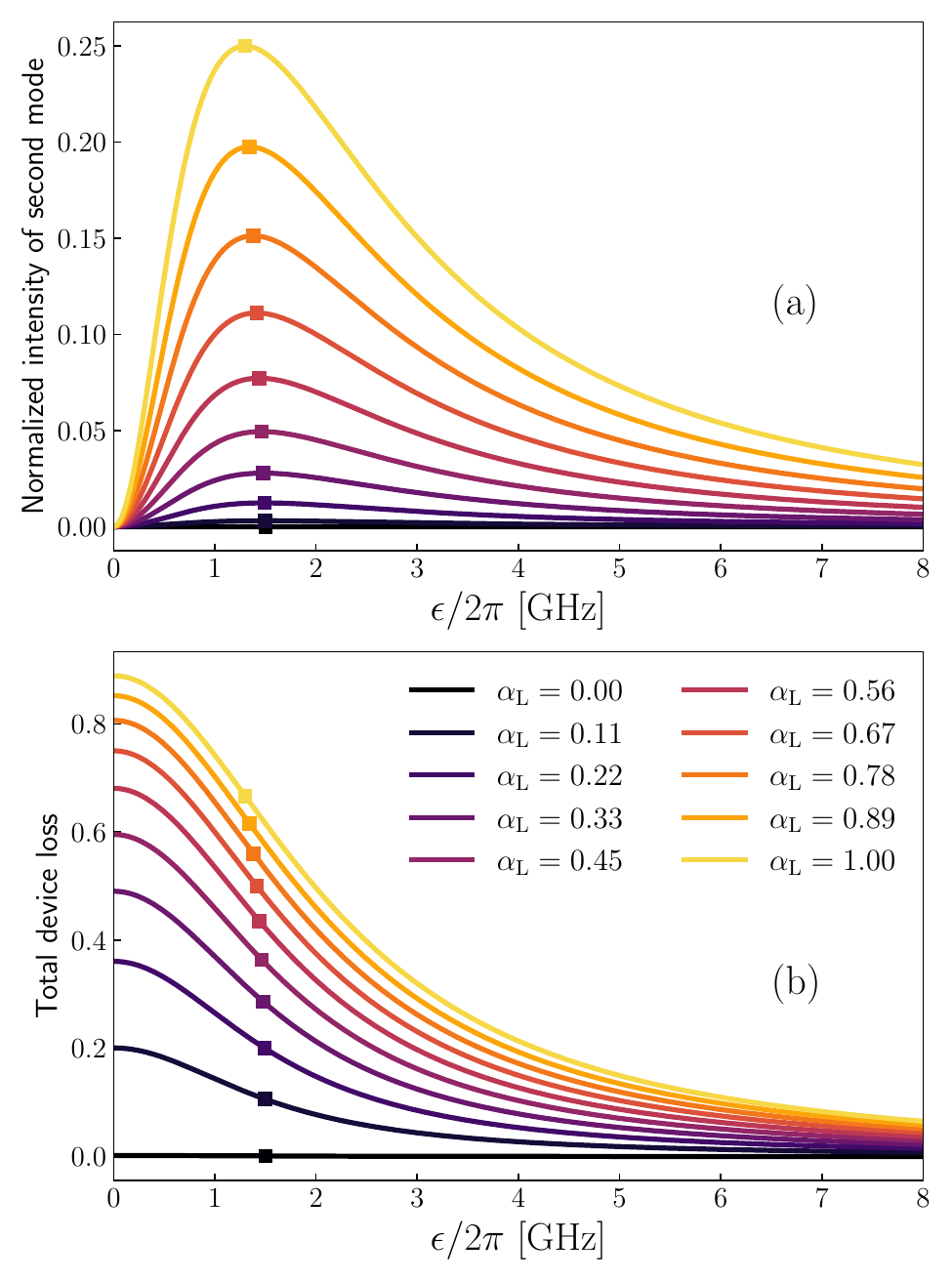}
    \caption{Frequency beam splitters for an under-coupled, $\Gamma_{\rm L}<\kappa_{\rm int}$, two-resonators-waveguide system. (a) Normalized output intensity on the second mode, $c_2$, assuming the input was resonant with $c_1$, as a function of modulation amplitude $\epsilon$. The squares indicate the best splitting ratio achievable at a given $\alpha_{\rm L}$, which occurs at the modulation amplitude $\epsilon_{\rm uc}$ in Eq.~(\ref{eqn:epsilon_uc}). (b) Total loss of the under-coupled system as function of the modulation amplitude. The colored square corresponds to the total loss at $\epsilon_{\rm uc}$. In both panels the color map, dark-to-light, indicates values of $\alpha_{\rm L}\in[0,1]$.}
    \label{fig:fig_undercoupled}
\end{figure}

In summary, in presence of the second waveguide we can frequency shift the output of the first waveguide but not that of the second waveguide. At the same time, the output of the second waveguide can be split $50$-$50$, however, at most half of the total device output can be routed to the second waveguide.

\subsection{Beam splitter ratios achievable with an undercoupled resonators-waveguide system}
\label{subsec:undercoupled}
All frequency beam splitters analyzed in the previous two subsections required the system of coupled resonators to be strongly coupled to the waveguide(s), that is, their cooperativities must satisfy $\alpha_{\rm L} = \Gamma_{\rm L}/\kappa_{\rm int}>1$ (similarly for $\alpha_{\rm R}$ in the case of two waveguides). Although this over-coupled regime has been already demonstrated (see Refs.~\cite{zhang2019electronically,Hu2021}), engineering devices in this regime is by no means an easy task. Here we investigate the types of frequency beam splitters implementable in the under-coupled regime for the two-resonator device. 

We begin with the case of a single waveguide. Clearly the implementable beam splitters in this regime cannot have the form in Eq.~(\ref{eqn:transfer_mat_R}), as they all require $\alpha_{\rm L}\ge 2\sqrt{1-R}$. This includes both the $0$-$100$ and the $50$-$50$ beam splitters in Eq.~(\ref{eqn:transfer_mat_gcc}) and Eq.~(\ref{eqn:transfer_mat_bs_2reso}), respectively. 
Thus, to gain some intuition on the behavior of the device in this regime we look at $\frac{|\Xi_{21}(\epsilon)|^2}{|\Xi_{11}(\epsilon)|^2 + |\Xi_{21}(\epsilon)|^2}$, the normalized output intensity at $\omega_2$ given $\omega_1$ as input, where $\Xi(\epsilon)$ is the transfer matrix in Eq.~(\ref{eqn:transfer_mat_2reso}). The behavior of this quantity as function of modulation amplitude is shown in Fig.~\ref{fig:fig_undercoupled}a for several values of $\alpha_{\rm L}\le 1$. Notably, a peak of this quantity is clearly seen. We find its position to be at
\begin{equation}
\label{eqn:epsilon_uc}
\epsilon_{\rm uc} = \sqrt{\kappa_{\rm int}^2 - \gamma_{\rm L}^2}.
\end{equation}
At this modulation amplitude the device implements a frequency beam splitter described by the transfer matrix
\begin{equation}
\label{eqn:transfer_mat_uc}
\Xi(\epsilon_{\rm uc};\alpha_{\rm L}) = \begin{pmatrix}
    \sqrt{1 - \frac{\alpha_{\rm L}^2}{4}} & i\frac{\alpha_{\rm L}}{2}e^{i\phi} \\ 
    i\frac{\alpha_{\rm L}}{2}e^{-i\phi} & \sqrt{1 - \frac{\alpha_{\rm L}^2}{4}}
\end{pmatrix} \mathcal{K}_{\rm uc}(\alpha_{\rm L}),
\end{equation}
where $\mathcal{K}_{\rm uc}(\alpha_{\rm L}) = \sqrt{\frac{2-\alpha_{\rm L}}{2 + \alpha_{\rm L}}}$ is the transmission amplitude parameter, leading to the total loss of $\mathcal{L}(\epsilon_{\rm uc};\alpha_{\rm L}) = 1 - |\mathcal{K}_{\rm uc}(\alpha_{\rm L})|^2 = \frac{2\alpha_{\rm L}}{2 + \alpha_{\rm L}}$. The output always contains a larger proportion of the input normal mode, with splitting ratio $(1-\frac{\alpha_{\rm L}^2}{4})$-$\frac{\alpha_{\rm L}^2}{4}$.

The most favorable splitting ratio is found at $\epsilon_{\rm uc}$ and $\alpha_{\rm L}=1$. It is a $75$-$25$ beam splitter (yellow square in Fig.~\ref{fig:fig_undercoupled}a). It has total losses of $\mathcal{L}(\epsilon_{\rm uc};1) = 2/3$ (yellow square in Fig.~\ref{fig:fig_undercoupled}b). Thus, the beam splitter achieving the most favorable ratio in this regime is also the most lossy. At $\epsilon_{\rm uc}$, with decreasing $\alpha_{\rm L}$, the total losses decrease and the splitting ratio becomes less favorable, until $\alpha_{\rm L}=0$, where no coupling between waveguide and resonators implies the input and output are identical, thus the device behaves as an identity element. Further, in the asymptotic limits $\epsilon\to0$ and $\epsilon\to\infty$, the under-coupled beam splitters behave as lossy and lossless identity elements, respectively.

Last but not least, in the case of two waveguides, the $50$-$50$ frequency beam splitter on the output of the second waveguide can \emph{always} be achieved. Recall that $\epsilon_{\rm bs}^{\rm (R)}$ equals the normal mode linewidth. However, the under-coupled regime is loss-dominated and we find the output intensity to be too depleted to have a meaningful output in any practical setting.

%
%
\section{Applications: Frequency-mode phase shifter and frequency-mode Mach--Zehnder interferometer}
\label{sec:mach_zender}
In this section we exploit the flexibility and composability of our frequency beam splitter models to give explicit constructions of a frequency-domain phase shifter and a frequency domain Mach--Zehnder interferometer.

\begin{figure}
    \centering
    \includegraphics[width=0.75\linewidth]{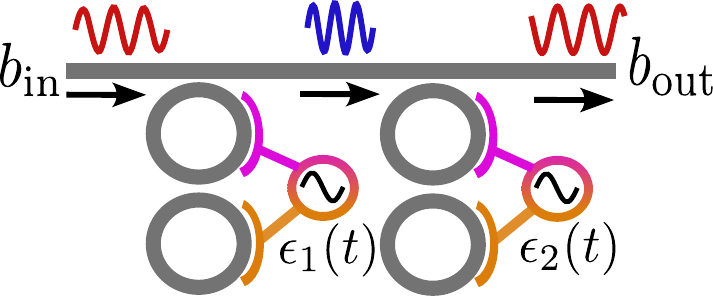}
    \caption{Schematic of two cascaded frequency beam splitters, each implemented with a pair of modulated coupled ring resonators. If the individual beam splitters are operated with $0$-$100$ splitting ratio, the cascaded device implements a frequency-domain phase shifter (shown). If the individual beam splitters are operated with $50$-$50$ splitting ratio, the cascaded device behaves as a frequency-domain Mach--Zehnder interferometer (not shown). In both cases, the desired behavior is obtained by adjusting the phases, $\phi_{1,2}$, of the individual modulations, $\epsilon_{1,2}(t)$, see text for details.}\label{fig:mach_zender}
\end{figure}  

A frequency-domain phase shifter is obtained via two consecutive frequency swaps, with the final output mode having the same frequency as the initial input, but picking up a phase-shift in the process. This operation can be achieved by cascading two $0$-$100$ frequency beam splitters mixing the same two frequencies $\omega_{1}$ and $\omega_2$, as shown in Fig.~\ref{fig:mach_zender}. Thus, the central frequency and coupling on each two-ring device must be picked as to guarantee pairs of normal modes with identical frequencies. Under these conditions an arbitrary phase-shift of $\mu$ can be generated adjusting the phases of the modulations $\epsilon_1(t)$ and $\epsilon_2(t)$. The input-output relation for the phase shifter is given by 
\begin{equation}
\mathbf{b}_{\rm out} = \Xi_{\rm PS}(\phi_1, \phi_2; \gamma_1, \gamma_2, \kappa_{\rm int}) \mathbf{b}_{\rm in},
\end{equation}
where $\phi_{1,2}$ are the phases of the modulations, $\gamma_{1,2} = \Gamma_{1,2}/2$ with $\Gamma_{1,2}$ the coupling rates to the waveguide of each pair of coupled resonators. The transfer matrix is obtained from those of the individual $0$-$100$ beam splitters by direct multiplication, we find
\begin{equation}
\label{eqn:transfer_mat_ps}
\begin{split}
&\Xi_{\rm PS}(\phi_1, \phi_2) = \Xi(\epsilon_{\rm GCC}^{(1)}) \Xi(\epsilon_{\rm GCC}^{(2)}) \\
&= \begin{pmatrix}
e^{i(\phi_1 - \phi_2 + \pi)} & 0 \\ 
0 & e^{-i(\phi_1 - \phi_2 + \pi)} \\
\end{pmatrix} \mathcal{K}_{\rm PS},
\end{split}
\end{equation}
where $\epsilon_{\rm GCC}^{(i)}$ is the modulation amplitude satisfying the GCC condition for each pair of coupled resonators, and 
\begin{equation}
\label{eqn:k_ps}
\mathcal{K}_{\rm PS}(\gamma_1,\gamma_2,\kappa_{\rm int}) = \sqrt{\frac{(\gamma_1 - \kappa_{\rm int})(\gamma_2 - \kappa_{\rm int})}{(\gamma_1 + \kappa_{\rm int})(\gamma_2 + \kappa_{\rm int})}}
\end{equation}
is the phase-shifter transmission amplitude parameter, 
leading to a total loss for the phase shifter of $\mathcal{L}_{\rm PS} = 1 - |\mathcal{K}_{\rm PS}(\gamma_1,\gamma_2,\kappa_{\rm int}) |^2$. From Eq.~(\ref{eqn:transfer_mat_ps}) we immediately see that the choice $\phi_1 = \mu/2$ and $\phi_2=\pi$, generates the desired relative phase-shift of $\mu$ between the two frequency modes. Interestingly, the design in Fig.~\ref{fig:mach_zender} suggests that, for frequency modes, a phase shifter is more complicated than a beam splitter. A flipped picture compared to that of spatial modes. Whether this design  offers any advantage over the usual approach involving a secondary spatial mode~\cite{Lukens2017} is an open question. 

The expression in Eq.~(\ref{eqn:transfer_mat_ps}) can equivalently be obtained from the SLH triples of the two coupled resonator systems applying the SLH composition rules, and invoking the same set of approximations used to obtain the effective transfer matrices of the individual beam splitters. As such, effective transfer matrices of composite frequency-domain elements can be obtain as product of the individual transfer matrices as long as the normal modes of the ring systems are well defined and the individual resonances do not overlap.

Similarly to the frequency-domain phase shifter, a frequency-domain Mach--Zehnder interferometer can be constructed by cascading two $50$-$50$ frequency beam splitters mixing the same two frequency modes. We obtain the transfer matrix for the interferometer by direct multiplication of the transfer matrices of the individual beam splitters.

Let $\phi_{1,2}$ be the phases of the modulations acting on the first and second pair of rings, respectively. After fixing $\phi_2=0$, we find
\begin{equation}
\begin{split}
\Xi_{\rm MZ}(\phi_1) &= \Xi(\epsilon_{\rm bs}^{\pm(1)})\Xi(\epsilon_{\rm bs}^{\pm(2)}) \\
&= i\begin{pmatrix}
-e^{i\frac{\phi_1}{2}}\sin(\phi_1/2) & e^{i\frac{\phi_1}{2}}\cos(\phi_1/2) \\ 
e^{-i\frac{\phi_1}{2}}\cos(\phi_1/2) & e^{-i\frac{\phi_1}{2}}\sin(\phi_1/2) \\
\end{pmatrix} \mathcal{K}_{\rm MZ},
\end{split}
\end{equation}
where $\epsilon_{\rm bs}^{\pm(i)}$ are the modulation amplitudes defining the individual $50$-$50$ beam splitters, and $\mathcal{K}_{\rm MZ}$ is the transmission amplitude parameter of the interferometer given by
\begin{equation}
\mathcal{K}_{\rm MZ} = \mathcal{K}_{\rm bs}^{\pm (1)}\mathcal{K}_{\rm bs}^{\pm (2)},
\end{equation}
with each term in the product being the transmission amplitude parameter of a $50$-$50$ beam splitter (see Eq.~(\ref{eqn:transfer_mat_bs_2reso})). Thus, given an input of frequency $\omega_i$ with $i=1,2$, we will see an output of frequency $\omega_j$, with $j=1,2$, with probability
\begin{equation}
p(\omega_j|\omega_i) = |\mathcal{K}_{\rm MZ}|^2 \times \begin{cases}
\frac{1 - \cos(\phi_1)}{2} \enspace \text{if} \enspace i=j, \\
\frac{1 + \cos(\phi_1)}{2} \enspace \text{if} \enspace i\ne j,
\end{cases}
\end{equation}
in agreement with the standard Mach--Zehnder interferometer output probabilities.

%
%

\section{Frequency beam splitters generated with four resonators}
\label{sec:4_reso}
We consider now a system of four resonators as in Fig.~\ref{fig:fig_1}c. We assume identical frequencies $\omega_0$, and coupling strengths $u$ (resonators $1$ and $2$, and $3$ and $4$), and $v$ (resonators $2$ and $3$, and $1$ and $4$). The system is described by the Hamiltonian 
\begin{equation}
\label{eqn:hamil_a4_usigned}
\begin{split}
H_0 &= \sum_{j=1}^4\omega_{0} a_j^\dagger a_j + u(a_1^\dagger a_{2} + a_3a^\dagger_{4} + {\rm h.c.}) \\ 
&+ v(a_1^\dagger a_{4} + a_2a^\dagger_{3} + {\rm h.c.}), 
\end{split}
\end{equation}
where $a_j$ ($a^\dagger_j$) are annihilation (creation) operators for the $j$-th resonator.  As illustrated in Fig.~\ref{fig:fig_1}c, we take the first resonator to be strongly coupled to a waveguide placed to the left of the array at rate $\Gamma_{\rm L}$, and the third resonator to be strongly coupled to a waveguide placed to the right of the array at rate $\Gamma_{\rm R}$. The $\mathbf{S}$ and $\mathbf{L}$ operators for the system are $\mathbf{S}=\mathbf{I}_2$ and $\mathbf{L}=(\sqrt{\Gamma_{\rm L}}a_1, \sqrt{\Gamma_{\rm R}}a_3)^T$ (or $\mathbf{S}=\mathbf{I}_1$ and $\mathbf{L}=\sqrt{\Gamma_{\rm L}}a_1$ in the case where only the left waveguide is present). We assume each of the four resonators has internal losses at rate $\kappa_{\rm int}$ and, as before, include the effects of these phenomenologically in the effective transfer matrices derived below.

This device has four normal modes and the above choice of resonator couplings yields normal modes with some desirable properties. By diagonalizing $H_0$ we can explicitly write the normal modes as
\begin{subequations}
\label{eqn:normal_modes_4reso}
\begin{align}
c_1 &= \frac{1}{2}(1,-1,1,-1).\mathbf{a}, \\
c_2 &= \frac{1}{2}(1,1,-1,-1).\mathbf{a}, \\ 
c_3 &= \frac{1}{2}(1,-1,-1,1).\mathbf{a}, \\ 
c_4 &= \frac{1}{2}(1,1,1,1).\mathbf{a},
\end{align}
\end{subequations}
where $\mathbf{a} = (a_1, a_2,a_3,a_4)^T$. Note that all normal modes have equal support on all of the resonators. The corresponding normal mode frequencies, listed in increasing order (for $v>u$), are $\omega_1 = \omega_0 - (u + v)$, $\omega_2 = \omega_0 + u - v$, $\omega_3 = \omega_0 - u + v$, and $\omega_4 = \omega_0 + u + v$, respectively. Importantly, note that choosing $v=2u$ yields equally spaced frequency levels. These two properties arising from the engineered symmetry---equal support of normal modes on all resonators and equally spaced normal mode frequencies---are useful for quantum computing applications, where each qubit is encoded in two frequency modes (dual rail encoding). In such applications a regular spacing of frequency modes makes the accounting simpler, and the equal support of normal modes makes the modulation protocols for implementing multimode beamsplitters simpler.

We introduce a time dependent modulation of the resonators frequencies, given by
\begin{align}
H_d(t) = \epsilon(t)\sum_{j=1}^4 f_j a_j^\dagger a_j
\label{eq:H_d}
\end{align}
with coefficients $f_{j} \in \{0,\pm1\}$. The common time-dependent portion is taken to be a single- or two-tone modulation of amplitude $\epsilon$: $\epsilon(t) = \epsilon(\cos(\omega_{\rm d}^{(1)} t + \phi_1) + \cos(\omega_{\rm d}^{(2)} + \phi_2))$. The total Hamiltonian $H(t) = H_0 + H_{\rm d}(t)$, in the normal mode basis, reads
\begin{equation}
\label{eqn:hamil_total_cs_4reso}
H(t) = \sum_{j=1}^4 \omega_{j} c_j^\dagger c_j + \epsilon(t)\sum_{(i,j)\in\mathcal{P}} \left( c_i^\dagger c_j + {\rm h.c.} \right),
\end{equation}
where $\mathcal{P}$, which we will refer to as the ``coupling pattern'', specifies the pairs of coupled normal modes, it is a direct consequence of our choice of signs $f_j$ in Eq.~(\ref{eq:H_d}).
This modulation activates couplings between the normal modes, and the exact type of coupling and $4$-mode beam-splitter that is implemented is determined by (i) the coupling pattern $\mathcal{P}$ and (ii) the choice of modulation frequencies, $\omega_{\rm d}^{(1)}$ and $\omega_{\rm d}^{(2)}$.

As in Sec.~\ref{sec:2_reso}, we proceed by making some rotating wave approximations in order to move into a frame without time dependence. In a rotating frame with respect to the normal mode frequencies, the Hamiltonian takes the form
\begin{equation}
\label{eqn:hamil_total_cs_4reso_RF}
\tilde{H}(t) =  \epsilon(t)\sum_{(i,j)\in\mathcal{P}} \tilde{c}_i^\dagger \tilde{c}_j e^{-i\Delta_{ij}t} + {\rm h.c.},
\end{equation}
where $\Delta_{ij} \equiv \omega_j - \omega_i$. Note that due to the parameter choices made above, this frequency difference can only take six values: $\pm2u, \pm4u, \pm6u$. The modulation frequencies will be chosen to be one of these values, and we will make the RWA to drop all oscillating terms in this Hamiltonian (and set the modulation phases $\phi_i=0$), to yield
\begin{equation}
\label{eqn:hamil_total_cs_4reso_RF}
\tilde{H} =  \epsilon\sum_{(i,j)\in\bar{\mathcal{P}}} \tilde{c}_i^\dagger \tilde{c}_j + {\rm h.c.},
\end{equation}
where $\bar{\mathcal{P}} \subset \mathcal{P}$ is the subset of the coupling patterns selected out by resonance conditions with the modulation frequencies. 

As with the two resonator case, we can also apply a RWA to rewrite the input-output fields, which are carried on one or two waveguides, in terms of frequency modes (since under the RWA only frequency content around the normal mode frequencies couples into the RBS system). Therefore, the input modes are described by the vectors $\mathbf{b}^L_{\rm in} = (b^L_{\rm in}(\omega_1), b^L_{\rm in}(\omega_2), b^L_{\rm in}(\omega_3), b^L_{\rm in}(\omega_4))^T$, and similarly for $\mathbf{b}^R_{\rm in}$. The $\mathbf{L}$ and $\mathbf{S}$ also get modified to
\begin{align}
\mathbf{L} &= (\sqrt{\gamma_{\rm L}}\tilde{c}_1, \sqrt{\gamma_{\rm R}}\tilde{c}_1, \sqrt{\gamma_{\rm L}}\tilde{c}_2, -\sqrt{\gamma_{\rm R}}\tilde{c}_2,  \nn \\
	& \quad\quad\sqrt{\gamma_{\rm L}}\tilde{c}_3, -\sqrt{\gamma_{\rm R}}\tilde{c}_3, \sqrt{\gamma_{\rm L}}\tilde{c}_4, \sqrt{\gamma_{\rm R}}\tilde{c}_4)^T \nn, \\
	\mathbf{S} &= \mathbf{I}_8,
\end{align}
where here $\gamma_{{\rm L/R}} = \Gamma_{{\rm L/R}}/4$.

At this point, we have a time-independent SLH description of the system which can be translated into a autonomous $ABCD$ description, and corresponding effective transfer matrix. 
In the next two sections we analyze the properties of the effective transfer matrices describing different configurations of this device. The different configurations are determined by choices of $f_j$ and modulation frequencies, $\omega_{\rm d}^{(i)}$. We focus on two configurations in the following, one implementing pairs of independent $2$-mode beam splitters and the other implementing a $4$-way symmetric splitter, because these are the most useful for resource state generation and fusion operations in PQC.

\subsection{Networks of non-overlapping \texorpdfstring{$2$}{\textit{2}}-mode frequency beam splitters}
By choice of $f_j$ and modulation frequencies we can implement networks of non-overlapping $2$-mode beam splitters between pairs of the four normal modes in this RBS device. We modulate all four resonators in this configuration, and use modulation coefficients of the form: $(f_1, f_2, f_3, f_4) = (1,-1,1,-1)$, $(f_1, f_2, f_3, f_4) = (1,-1,-1,1)$, and, $(f_1, f_2, f_3, f_4) = (1,1,-1,-1)$, leading to normal modes coupled via the patterns
\begin{equation}
\label{eqn:patterns_4reso}
\begin{split}
\mathcal{P}_1 = \{(2,3),(1,4)\}, \\
\mathcal{P}_2 = \{(1,2),(3,4)\}, \\
\mathcal{P}_3 = \{(1,3),(2,4)\}, \\
\end{split}
\end{equation}
which we schematically show in Fig.~\ref{fig:fig_loss}a. These patterns cover all three possible different ways of coupling two of four modes without repetitions and overlaps. Notice that when $v=2u$, the splittings between coupled levels in patterns $\mathcal{P}_{2}$ are equal (similarly for those in $\mathcal{P}_3$), the corresponding pair of beam splitters are activated by a single-tone modulation. However, activating the beam splitters specified by pattern $\mathcal{P}_1$ requires a two-tone modulation.

Each coupling pattern in Eq.~(\ref{eqn:patterns_4reso}) specifies two independent $2$-mode beam splitters, we find the corresponding transfer matrices to be block-diagonal, with each block being either $2\times2$ (one waveguide) or $4\times4$ (two waveguides) and having the respective form of Eq.~(\ref{eqn:transfer_mat_2reso}) or Eq.~(\ref{eqn:transfer_mat_2reso_2wv}). Consequently, the modulation amplitudes defining the $0$-$100$ and $50$-$50$ beam splitters have the same form as those in Eq.~(\ref{eqn:gcc_2reso_1wv}) and Eq.~(\ref{eqn:bs_2reso_1wv}) (one waveguide), and Eq.~(\ref{eqn:epsilon_bs_2vw}) (two waveguides), where now $\gamma_{\rm L,R} = \frac{\Gamma_{\rm L,R}}{4}$. Hence, we require weaker modulation amplitudes to activate the desired beam splitters when implemented with the array of four resonators. However, the implemented beam splitters experience larger losses. We explore this trade-off for the cases of one and two waveguides next.

\subsubsection{Single waveguide}
\label{subsec:4_reso_1wv}
We begin by analyzing the behavior of the $0$-$100$ frequency beam splitter. Both non-overlapping beam splitters in the pattern are tuned to this ratio at the GCC condition, $\epsilon_{\rm GCC} = \sqrt{\frac{\Gamma_{\rm L}^2}{16} - \kappa_{\rm int}^2}$. Clearly, activating this behavior demands an over-coupled resonator-waveguide system, with a ratio $\frac{\Gamma_{\rm L}}{\kappa_{\rm int}}>4$, twice the ratio required to activate the same behavior in the two-resonator device. 

Each of the two $0$-$100$ frequency beam splitters experiences total losses equal to
\begin{equation}
\label{eqn:loss_gcc_4reso_1wv}
\mathcal{L}(\epsilon_{\rm GCC}) = \frac{8}{\alpha_{\rm L} + 4},
\end{equation} 
where $\alpha_{\rm L} = \Gamma_{\rm L}/\kappa_{\rm int}$. This expression is to be compared with Eq.~(\ref{eqn:loss_gcc_2mode}), which we illustrate in Fig.~\ref{fig:fig_loss}b. Their ratio is given by $\frac{\mathcal{L}(\epsilon_{\rm GCC};4)}{\mathcal{L}(\epsilon_{\rm GCC};2)} = \frac{2(\alpha_{\rm L} + 2)}{\alpha_{\rm L} + 4}$. That is, one frequency shifter implemented with the array of four resonators experiences at most twice the amount of loss experienced by the frequency shifter implemented with the two-resonator device.

Consider now the $50$-$50$ beam splitter. Similarly to the two-resonator device, one activates this behavior at the modulation amplitudes $\epsilon_{\rm bs}^\pm = |\frac{\Gamma_{\rm L}}{4} \pm \sqrt{\frac{\Gamma_{\rm L}^2}{8} - \kappa_{\rm int}^2}|$, hence activating this behavior requires operation in a regime of overcoupling between the resonators and the waveguide, with a minimum ratio $\frac{\Gamma_{\rm L}}{\kappa_{\rm int}}\ge\sqrt{8}$. Twice the minimum ratio required to activate the same behavior in the two-resonator device.

Each of the two $50$-$50$ beam splitter in one of the patterns experiences a total amount of loss equal to
\begin{equation}
\mathcal{L}(\epsilon_{\rm bs}^\pm) = \frac{8\left(\sqrt{2}(\alpha_{\rm L} + 2) \mp \sqrt{\alpha_{\rm L}^2 - 8}\right)}{\sqrt{2}(\alpha_{\rm L} + 4)^2},
\end{equation}
which is to be compared with Eq.~(\ref{eqn:loss_bs_2reso_1wv}), which we illustrate in Fig.~\ref{fig:fig_loss}c. We find, in the limit $\alpha_{\rm L}\gg1$, the loss on each $50$-$50$ beam splitter generated with the four-resonator device to be twice the loss of the $50$-$50$ beam splitter generated with the two-resonator device.

Importantly, the limitation we just discussed can be addressed at the design and fabrication stages. Our analysis assumes that $\Gamma_{\rm L}$ (equivalently $\alpha_{\rm L}$) remains fixed as we scale up the array size. As such, if one can engineer stronger coupling to the waveguide for the four-resonator device, one might reach a point at which the total loss is the same in both the four- and two-resonator devices for a given target behavior. We can illustrate this for the $0$-$100$ beam splitter, taking $\alpha_{\rm L} \to 2\alpha_{\rm L}$ in Eq.~(\ref{eqn:loss_gcc_4reso_1wv}), we immediately recover the same expression as Eq.~(\ref{eqn:loss_gcc_2mode}). We further illustrate this fact in Fig.~\ref{fig:fig_loss}b,c, where we compare the total loss for the two targeted behaviors generated with both the two- and four-resonator devices as a function of $\alpha_{\rm L}$.

\begin{figure}
    \centering
    \includegraphics[width=0.98\linewidth]{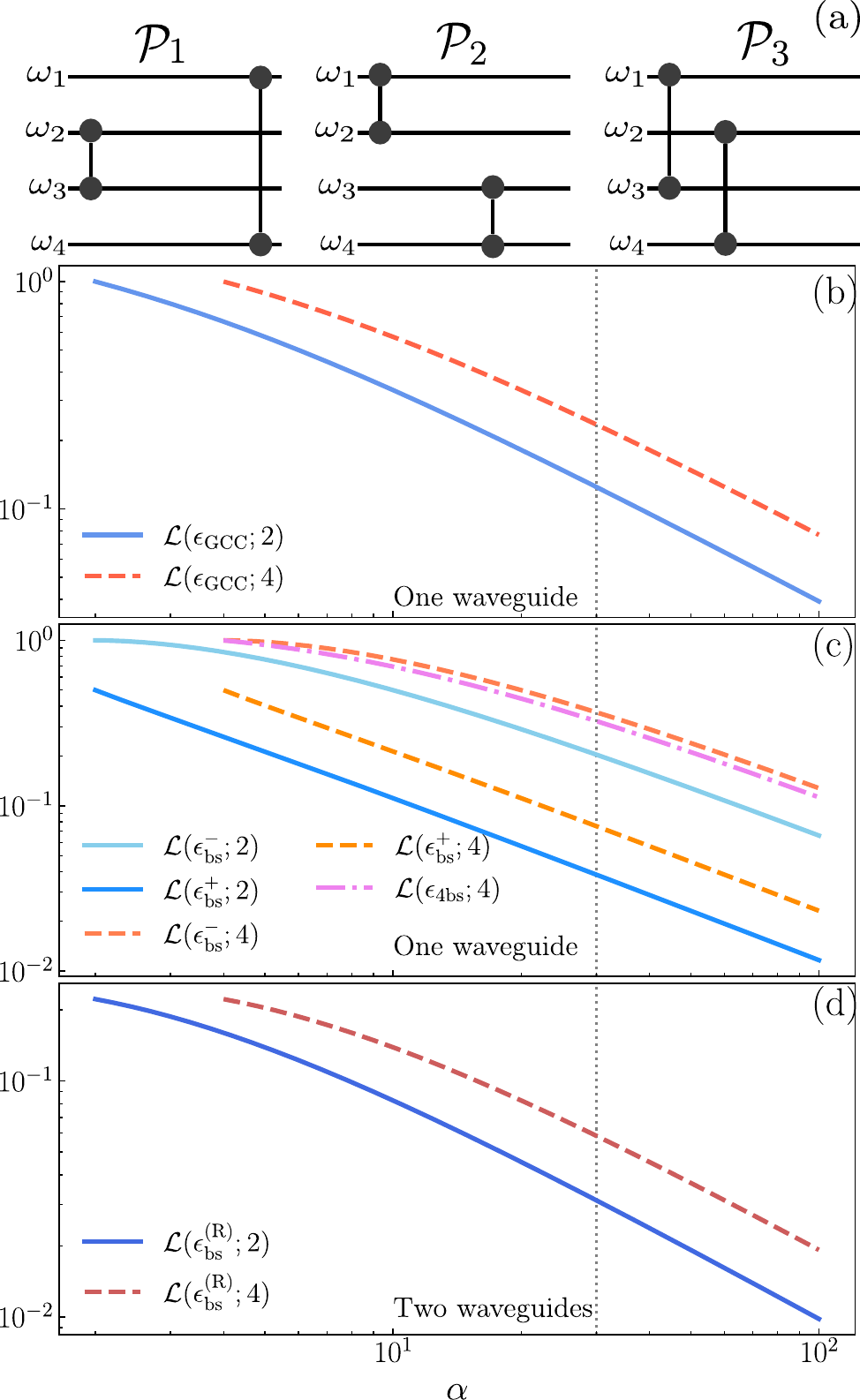}
    \caption{(a) Schematic illustration of the nonoverlapping coupling patterns, $\mathcal{P}_{1,2,3}$, of normal modes generated by driving all four resonators. (b-d) Total loss of the different beam splitters generated with the four-resonator device (broken lines) compared against the total loss of the beam splitters generated with the two-resonator device (solid lines). (b) Total loss of one $0$-$100$ frequency beam splitter. (c) Total loss of one $50$-$50$ frequency beam splitter, at the two working points, and the $4$-way symmetric beam splitter (dashed-dotted line). (d) Total loss of one $50$-$50$ frequency beam splitter for the case of two waveguides, we are taking $\Gamma_{\rm L}=\Gamma_{\rm R} = \Gamma$. The vertical dotted line shows $\alpha = 30$.} 
    \label{fig:fig_loss}
\end{figure}

\subsubsection{Two waveguides}
\label{subsec:4_reso_2wv}
Each $4\times 4$ subblock in the transfer matrix has the form in Eq.~(\ref{eqn:transfer_mat_2reso_2wv}), as such, we know there is no value of the modulation amplitude at which the output on the second waveguide is split with a $0$-$100$ ratio. On the other hand, at the modulation amplitude $\epsilon_{\rm bs}^{(\rm R)} = \kappa = \gamma_{\rm L} + \gamma_{\rm R} + \kappa_{\rm int}$, the output on the second waveguide is split at a $50$-$50$ ratio for each of the two beam splitters. However, one cannot route more than half of the total input, of each beam splitter, into the output of the second waveguide, and this maximum value is only achieved when $\kappa_{\rm int} = 0$ and $\gamma_{\rm L}=\gamma_{\rm R}$. Importantly, since $\gamma_{\rm L,R}=\frac{\Gamma_{\rm L,R}}{4}$, activating the $50$-$50$ behavior of each beam splitter requires a weaker modulation amplitude compared to that necessary to activate the same behavior with the $2$-resonator device.

When internal losses are not zero, we find each of the two $50$-$50$ beam splitters in the pattern suffers from total losses equal to
\begin{equation}
\label{eqn:loss_4resos_2wv}
\mathcal{L}_{\rm R}(\epsilon_{\rm bs}^{\rm (R)}) = \frac{8\alpha_{\rm L}}{(\alpha_{\rm L} + \alpha_{\rm R} + 4)^2}.
\end{equation}
When $\gamma_{\rm L} = \gamma_{\rm R}$ we find $\mathcal{L}_{\rm R}(\epsilon_{\rm bs}^{\rm (R)}) = \frac{2\alpha}{(\alpha + 2)^2} \to \frac{2}{\alpha}$ in the limit of $\alpha\gg1$. In particular when $\alpha=30$ each beam splitter suffers from total losses of about $5.85\%$. 

The previous expression is to be compared with Eq.~(\ref{eqn:loss_2resos_2wv}), which we illustrate, for the case of $\gamma_{\rm L} = \gamma_{\rm R}$, in Fig.~\ref{fig:fig_loss}d. We find their ratio to be $\frac{\mathcal{L}(\epsilon_{\rm bs}^{(R)};4)}{\mathcal{L}(\epsilon_{\rm bs}^{(\rm R)};2)} = \frac{2(\alpha_L + \alpha_R + 2)^2}{(\alpha_L + \alpha_R + 4)^2}$, and when the couplings to both waveguides are equal, we have $\frac{\mathcal{L}(\epsilon_{\rm bs}^{(R)};4)}{\mathcal{L}(\epsilon_{\rm bs}^{(\rm R)};2)} =  \frac{2(\alpha+1)^2}{(\alpha + 2)^2}$, which approaches $2$ from below as $\alpha\to\infty$. Hence, each $50$-$50$ beam splitter implemented with the four resonator devices suffers from about twice as much total loss as a beam splitter implemented with the two resonator device. Again, this issue might be bypassed by engineering stronger coupling to the waveguides. In fact, if one can fabricate the $4$-resonator device such that $\alpha_{\rm L,R}\to2\alpha_{\rm L,R}$, then Eq.~(\ref{eqn:loss_4resos_2wv}) becomes identical to Eq.~(\ref{eqn:loss_bs_2reso_1wv}), that is, by coupling the $4$-resonator device twice as strongly to the waveguides as the $2$-resonator device, both devices will experience the same amount of total loss at the $50$-$50$ beam splitting point.

\subsection{\texorpdfstring{$4$}{\textit{4}}-way symmetric beam splitter}
\label{subsec:4way_symmetric_bs}
One advantage offered by the frequency domain, is that one can implement multimode beamsplitters natively, without needing a sequence of two-mode beamsplitters. We demonstrate this by implementing a $4$-way symmetric beam splitter with the four-resonator RBS device. This requires modulating only two of the resonators. In principle, modulating any pair of resonators would lead to the desired behavior. Here, without loss of generality, we chose to modulate resonators $1$ and $2$. This is achieved with the choice of signs $(f_1,f_2,f_3,f_4)=(1,-1,0,0)$, which leads to the coupling pattern for the normal modes $\mathcal{P}_4 = \{(1,2),(2,3),(3,4),(1,4)\}$ (see Fig.~\ref{fig:fig_tm_4way}a). 

$\mathcal{P}_4$ contains two distinct level splittings $\Delta_{12} = \Delta_{23} = \Delta_{34} = 2u$ and $\Delta_{14} = 6u$. Thus, $\epsilon(t)$ needs to be a two-tone modulation in order to activate the desired behavior. We pick these tones to be of the form $\epsilon(t) = \epsilon(\cos(\omega_{\rm d}^{(1)} t + \phi_1) + \cos(\omega_{\rm d}^{(2)} + \phi_2))$, with each modulation frequency resonant with one of the splittings, $\omega_{\rm d}^{(1)} = 2u$ and $\omega_{\rm d}^{(2)} = 6u$.
 
With the aim of keeping the discussion short, we will only consider the case of a single waveguide coupled to the array via $a_1$ at a rate $\Gamma_{\rm L}$. Further, we take an input field with a single frequency, $\omega_{\rm I} = \omega_1$, which we pick to be resonant with normal mode $c_1$. At this point we follow the steps described at the beginning of this section to arrive at the effective transfer matrix for the optical element. Importantly, our choice of $\phi_{1,2}$ (the modulation phases), can yield a transfer matrix with either only real entries or real and imaginary entries. For the former one sets $\phi_1 = -\phi_2$ with either of the two phases equal to $\pi/2$ and for the latter one sets $\phi_1 = l\pi$ and $\phi_2 = m\pi$ with $l,m = 0,\pm1,\pm2,..$ constrained to their sum, $l+m$, being even. 

We take $\phi_1=\pi/2$ and $\phi_2 = -\pi/2$, then the transfer matrix is given by
\begin{widetext}
\begin{equation}
\label{eqn:transfer_mat_4way_bs_general}
\Xi(\epsilon;\gamma_{\rm L}, \kappa_{\rm int}) = \begin{pmatrix}
1 - \frac{2\gamma_{\rm L}(\kappa^2 + 2\tilde{\epsilon}^2)}{\kappa(\kappa^2 + 4\tilde{\epsilon}^2)} && -\frac{2\gamma_{\rm L}\tilde{\epsilon}}{\kappa^2 + 4\tilde{\epsilon}^2} && -\frac{4\gamma_{\rm L}\tilde{\epsilon}^2}{\kappa(\kappa^2 + 4\tilde{\epsilon}^2)} && \frac{2\gamma_{\rm L}\tilde{\epsilon}}{\kappa^2 + 4\tilde{\epsilon}^2} \\
\frac{2\gamma_{\rm L}\tilde{\epsilon}}{\kappa^2 + 4\tilde{\epsilon}^2} && 1 - \frac{2\gamma_{\rm L}(\kappa^2 + 2\tilde{\epsilon}^2)}{\kappa(\kappa^2 + 4\tilde{\epsilon}^2)} && -\frac{2\gamma_{\rm L}\tilde{\epsilon}}{\kappa^2 + 4\tilde{\epsilon}^2} && -\frac{4\gamma_{\rm L}\tilde{\epsilon}^2}{\kappa(\kappa^2 + 4\tilde{\epsilon}^2)} \\
-\frac{4\gamma_{\rm L}\tilde{\epsilon}^2}{\kappa(\kappa^2 + 4\tilde{\epsilon}^2)} && \frac{2\gamma_{\rm L}\tilde{\epsilon}}{\kappa^2 + 4\tilde{\epsilon}^2} && 1 - \frac{2\gamma_{\rm L}(\kappa^2 + 2\tilde{\epsilon}^2)}{\kappa(\kappa^2 + 4\tilde{\epsilon}^2)} && -\frac{2\gamma_{\rm L}\tilde{\epsilon}}{\kappa^2 + 4\tilde{\epsilon}^2} \\
-\frac{2\gamma_{\rm L}\tilde{\epsilon}}{\kappa^2 + 4\tilde{\epsilon}^2} && -\frac{4\gamma_{\rm L}\tilde{\epsilon}^2}{\kappa(\kappa^2 + 4\tilde{\epsilon}^2)} && \frac{2\gamma_{\rm L}\tilde{\epsilon}}{\kappa^2 + 4\tilde{\epsilon}^2} && 1 - \frac{2\gamma_{\rm L}(\kappa^2 + 2\tilde{\epsilon}^2)}{\kappa(\kappa^2 + 4\tilde{\epsilon}^2)}
\end{pmatrix},
\end{equation}
\end{widetext}
where we defined $\tilde{\epsilon}=\epsilon/2$, and $\kappa = \gamma_{\rm L} + \kappa_{\rm int}$ is the total normal mode linewidth.
For a $4$-way symmetric beam splitter, in the absence of internal losses, the norm squared of all the entries in Eq.~\eqref{eqn:transfer_mat_4way_bs_general} must be $1/4$. Letting $\tilde{\epsilon} = \frac{\sqrt{\gamma_{\rm L}^2 - \kappa_{\rm int}^2}}{2}$ satisfies this condition. Thus the modulation amplitude defining a $4$-way symmetric beam splitter is 
\begin{equation}
\label{eqn:epsilon_4way_bs}
\epsilon_{4 \rm bs} = \sqrt{\gamma_{\rm L}^2 - \kappa_{\rm int}^2}.
\end{equation}
That is, if we modulate two of the resonators with opposite signs, at GCC, the $4$-resonator device implements a $4$-way symmetric beam splitter. Although we found $\epsilon_{4\rm bs}$ under a specific choice of values of modulation phases, Eq.~(\ref{eqn:epsilon_4way_bs}) also yields a $4$-way symmetric splitter had we chosen modulation phases leading to a $\Xi(\epsilon)$ with mixed real and imaginary entries. The behavior of $\Xi$ as function of the modulation amplitude is shown in Fig.~\ref{fig:fig_tm_4way}b,c.

At this modulation amplitude the transfer matrix simplifies to 
\begin{equation}
\label{eqn:transfer_mat_4way_bs}
\Xi(\epsilon_{4\rm bs}) = \frac{1}{2}\begin{pmatrix}
-\mathcal{K}_{\rm 4bs} & -\sqrt{\mathcal{K}_{\rm 4bs}} & -\mathcal{K}_{\rm 4bs} & \sqrt{\mathcal{K}_{\rm 4bs}}\\
\sqrt{\mathcal{K}_{\rm 4bs}} & -\mathcal{K}_{\rm 4bs} & -\sqrt{\mathcal{K}_{\rm 4bs}} & -\mathcal{K}_{\rm 4bs} \\
-\mathcal{K}_{\rm 4bs} & \sqrt{\mathcal{K}_{\rm 4bs}} & -\mathcal{K}_{\rm 4bs} & -\sqrt{\mathcal{K}_{\rm 4bs}} \\
-\sqrt{\mathcal{K}_{\rm 4bs}} & -\mathcal{K}_{\rm 4bs} & \sqrt{\mathcal{K}_{\rm 4bs}} & -\mathcal{K}_{\rm 4bs}
\end{pmatrix},
\end{equation}
where the outputs have transmission amplitude parameters of $\sqrt{\mathcal{K}_{\rm 4bs}}$ or $\mathcal{K}_{\rm 4bs}$, depending on which normal mode was input, and $\mathcal{K}_{\rm 4bs} = \frac{\gamma_{\rm L} - \kappa_{\rm int}}{\gamma_{\rm L} + \kappa_{\rm int}} = \frac{\alpha_{\rm L} - 4}{\alpha_{\rm L} + 4}$. Clearly, when $\kappa_{\rm int} = 0$, we have $\mathcal{K}_{\rm 4bs}=1$, and the system implements a perfect $4$-way symmetric splitter.

Nonzero internal loses affect our $4$-way splitter in two different ways. First, it depletes the output intensity by a factor of $\frac{\mathcal{K}_{\rm 4bs}^2+\mathcal{K}_{\rm 4bs}}{2}$, that is, at this working point the devices suffers from total losses equal to
\begin{equation}
\label{eqn:loss_4way_bs}
\mathcal{L}(\epsilon_{4\rm bs}) = \frac{3\gamma_{\rm L} + \kappa_{\rm int ^2}}{(\gamma_{\rm L} + \kappa_{\rm int})^2} = \frac{4(3\alpha_{\rm L} + 4)}{(\alpha_{\rm L} + 4)^2},
\end{equation}
which goes as $\mathcal{L}(\epsilon_{4\rm bs}) \sim 12 / \alpha_{\rm L}$ when $\alpha_{\rm L}\gg 1$. In Fig.~\ref{fig:fig_loss}c (dashed-dotted line) we compare this loss with the loss of the $50$-$50$ beam splitters generated with the four-resonator device. We find it to be slightly less than the loss experienced by one $2$-mode beam splitter at $\epsilon_{\rm bs}^-$ of the four-resonator device. 
\begin{figure}[ht!]
    \centering
    \includegraphics[width=0.95\linewidth]{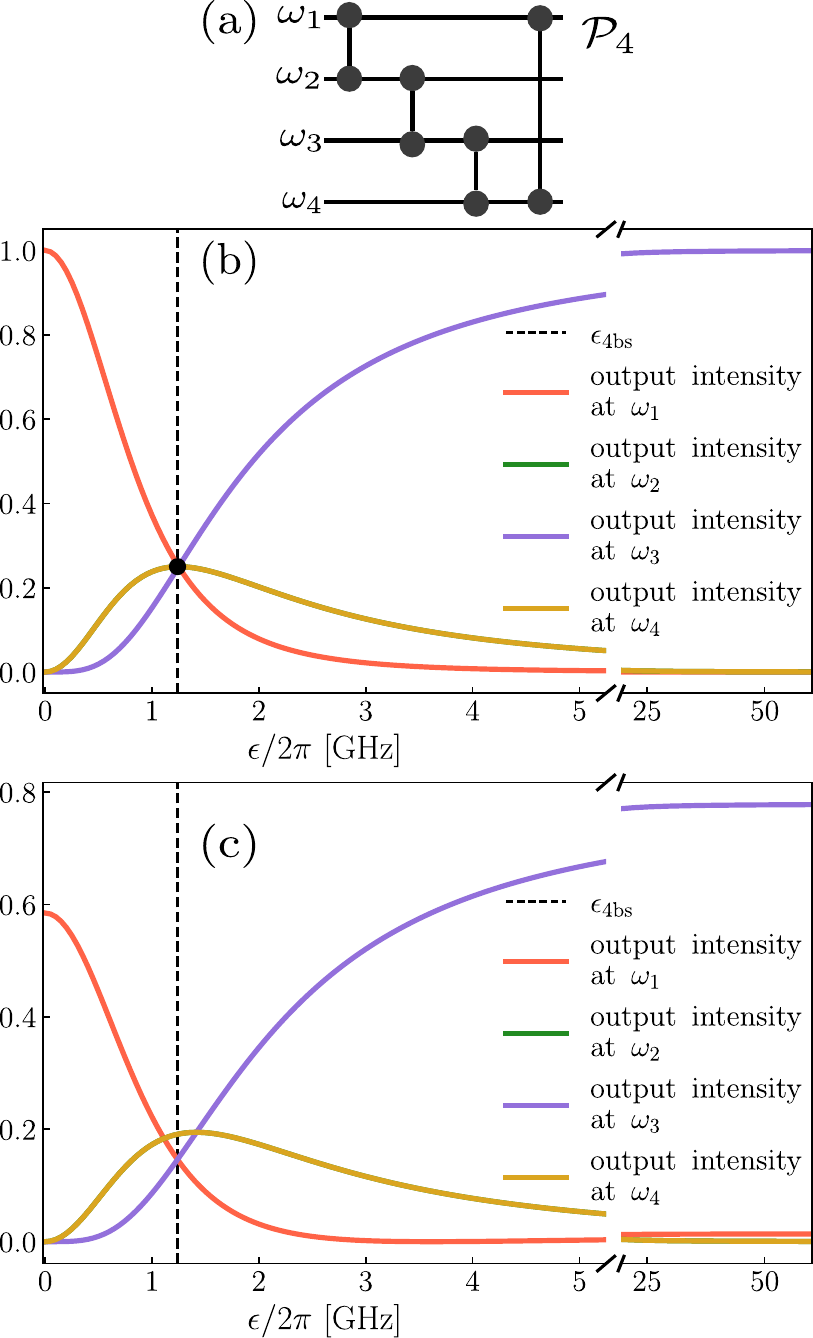}
    \caption{ (a) Illustration of the coupling pattern $\mathcal{P}_4$ used to generate the $4$-way symmetric beam splitter. (b,c) Squared norm of the entries of Eq.~(\ref{eqn:transfer_mat_4way_bs}) in the absence (b) and presence (c) of internal losses $\kappa_{\rm int}$. The input is taken to be resonant with $\omega_1$. The extended axis shows the behavior as $\epsilon\to\infty$ which is governed by Eq.~(\ref{eqn:transfer_mat_4way_limit}). The vertical dashed line indicates $\epsilon_{\rm 4bs}$. Both, in (b) and (c), the output at $\omega_2$ is identical to the output at $\omega_4$, \ie green and orange lines overlap. Parameters: $\Gamma_{\rm L}/2\pi = 5~{\rm GHz}$, $\kappa_{\rm int}/2\pi = 0.17~{\rm GHz}$, and $\gamma_{\rm L} = \Gamma_{\rm L}/4$.} 
    \label{fig:fig_tm_4way}
\end{figure}
Second, nonzero internal losses break the symmetric character of the linear optical element, an effect which is evident from the form of $\Xi(\epsilon_{4\rm bs})$ in Eq.~(\ref{eqn:transfer_mat_4way_bs}). A similar effect takes place regardless of which of the four modes is sent as input, and regardless of which pair of resonators are modulated.

Finally, let us look at the asymptotic behavior of the transfer matrix. In the limit $\epsilon\to0$, all off-diagonal entries in Eq.~(\ref{eqn:transfer_mat_4way_bs_general}) go to zero, and the diagonal entries become $\Xi_{ii}(\epsilon\to0)\to\frac{-(\gamma_{\rm L} - \kappa_{\rm int})}{\gamma_{\rm L} + \kappa_{\rm int}}$, thus the system behaves as a lossy identity element. This behavior, for the case of input resonant with $\omega_1$, is illustrated in Fig.~\ref{fig:fig_tm_4way}b,c. In the limit of large modulation amplitudes, $\epsilon\to\infty$, the transfer matrix takes the form
\begin{equation}
\label{eqn:transfer_mat_4way_limit}
\Xi(\epsilon\to\infty) = \begin{pmatrix}
\kappa_{\rm int}/\kappa & 0 & -\gamma_{\rm L}/\kappa & 0 \\
0 & \kappa_{\rm int}/\kappa & 0 & -\gamma_{\rm L}/\kappa \\
-\gamma_{\rm L}/\kappa & 0 & \kappa_{\rm int}/\kappa & 0 \\
0 & -\gamma_{\rm L}/\kappa & 0 & \kappa_{\rm int}/\kappa
\end{pmatrix},
\end{equation}
which in absence of internal loss implements a perfect swap between $c_1$ and $c_3$, and $c_2$ and $c_4$. Interestingly, in this limit, the way the modes are connected depends on the pair of resonators being modulated. This behavior, for the case of an input resonant with $\omega_1$, is illustrated in the extended axis of Fig.~\ref{fig:fig_tm_4way}c, where the output is mostly mode $c_3$. Notice that in absence of internal loss the mode swap is perfect (see extended axis in Fig.~\ref{fig:fig_tm_4way}b).

%
%
\section{Frequency beamsplitters with \texorpdfstring{$N>4$}{\textit{N>4}} resonators}
\label{sec:no_go}
One interesting and appealing aspect of PQC with frequency-domain encoded qubits is the possibility of implementing $N$-mode transformations using a compact device -- \ie instead of building an $N$-mode beamsplitter from a series of two-mode beamsplitters in the case of spatial modes, one could implement it directly on frequency modes all propagating in a single waveguide. The 4-way symmetric beam splitter constructed in the previous section is an example of such device. In this section we study the implementation of such transformations for $N>4$ modes using coupled ring resonators. 

\subsection{Physical system and model}
\label{subsec:model}
Consider an array of $N$ identical ring resonators of central frequency $\omega_0$. They are all strongly coupled to their immediate neighbors, as dictated by the array connectivity, with strengths $u_{j, k}\ge 0$. This system is described by the time-independent Hamiltonian
\begin{equation}
\label{eqn:hamil_as_Nreso}
H_0 = \omega_0\sum_{j=1}^{N}a^\dagger_{j}a_{j} + \sum_{i=1}^N\sum_{j\in N(i)} \left( u_{i, j} a^\dagger_{i}a_{j} + {\rm h.c.} \right),
\end{equation}
where $N(i)$ is the set of all the neighbors of the $i$-th resonator.

The Hamiltonian in Eq.~(\ref{eqn:hamil_as_Nreso}) is quadratic in the creation and annihilation operators, so we can write it as $H_0 = \mathbf{a}^\dagger h_0 \mathbf{a}$, with $h_0$ a real, symmetric matrix. The system is then guaranteed to have a set of $N$ orthonormal normal modes of oscillation corresponding to the eigenvectors and eigenvalues of $h_0$. After diagonalizing $h_0$, we can write $H_0$ in Eq.~(\ref{eqn:hamil_as_Nreso}) as
\begin{equation}
\label{eqn:hamil_cs_Nreso}
H_0 = \sum_{k=1}^N \omega_k c_k^\dagger c_k,
\end{equation}
with $c_k$ ($c_k^\dagger$) the annihilation (creation) operator for the $k$-th normal mode, and $\omega_k$ its corresponding frequency. Importantly, the structure of the normal modes and their frequencies is completely characterized by the set of coupling strengths $\{u_{i,j}\}$. 

We activate couplings between the normal modes via a modulation of the form
\begin{equation}
\label{eqn:hamil_drive_Nreso}
H_{\rm d}(t) = \epsilon(t)\sum_{j=1}^N f_{j} a_{j}^\dagger a_{j},
\end{equation}
where the coefficients $f_{j} \in \{0,\pm1\}$ can be chosen with complete freedom, and we leave the explicit form of $\epsilon(t)$ unspecified for the moment. Our choice of $f_j$'s dictates how the normal modes will be coupled, such that, in their basis, the total Hamiltonian $H(t) = H_0 + H_{\rm d}(t)$ is
\begin{equation}
H(t) = \sum_{j=1}^N \omega_{j} c_j^\dagger c_j + \epsilon(t)\sum_{(i,j)\in\mathcal{P}} \left( c_i^\dagger c_j + {\rm h.c.} \right),
\end{equation}
where $\mathcal{P}$ is a ``pattern'' of couplings. The specific relation between our choice of $f_j$'s and $\mathcal{P}$ depends on how the normal modes are supported on the physical resonators, and thus, will vary depending on the array structure. In App.~\ref{app:N_reso_rectangle} we present explicit expressions for this relationship in the case of a rectangular array. 

Further the array is coupled to a waveguide via $a_1$ at rate $\Gamma_{\rm L}$, and to a second waveguide, via $a_N$, at rate $\Gamma_{\rm R}$. All resonators suffer internal losses at a finite rate $\kappa_{\rm int}$. From the structure of $\mathcal{P}$ the number of distinct splittings between coupled normal modes and the number of independent beam splitters can be obtained. Thus, we can specify both the form of the modulation $\epsilon(t)$ and its frequency content, as well as the form of the input field $\mathbf{b}_{\rm in}(t)$ and its frequency content. 

After these two quantities have been specified, we obtain the effective transfer matrix of our native $N$-mode linear optical element following a similar approach to the one employed for the two and four resonator devices: from the SLH triple of the system in the normal mode basis (after promoting $\mathbf{L}$ to vector form), the $A(t)BCD$ system is constructed. Then, by going to the slow components of the normal modes, neglecting the fast oscillating terms under the RWA, and discarding the off-resonant inputs, a time-independent $ABCD$ system is recovered, which then is solved via Fourier transform. The detailed derivation of these transfer matrices in the general case can be found in App.~\ref{app:general_transfer_mat}.

While transfer matrices for $N$-mode frequency beam splitters implemented with arbitrary devices of coupled ring resonators can be obtained with the methodology summarized above, the resulting beam splitters might lack a useful structure and tuning them to some specific desired behavior might not always be possible. To illustrate this last point, in App.~\ref{app:N_reso_rectangle} we completely characterize the case of rectangular arrays of sizes $L\times M$.

To constrain this large degree of generality and motivated by our discussion of the $2$-resonator and $4$-resonator devices, we impose two desired properties of any $N$-mode ring resonator device:
\begin{property}
\label{property:1}
Equal spacing between consecutive normal mode frequencies.
\end{property}
\begin{property}
\label{property:2}
Normal modes with uniform support on all the resonators.
\end{property}
We saw that arrays with $N=2$ and $N=4$ exhibit these two properties with suitable parameter choices. While neither is necessary for PQC applications, both are desirable. Together, these two properties guarantee that all the beam splitters in a given pattern can be tuned to the same ratio at the \emph{same} modulation amplitude. 

Are there arrays of $N>4$ resonators with these two desirable properties? And can they be realized with photonic microring resonators? Notably, property~\ref{property:1} depends only on the values of the $u_{i,j}$'s, and thus, given $N>4$ resonators with a fixed array connectivity, one can find, via a simple optimization, an assignment of coupling strengths leading to property~\ref{property:1}. Property~\ref{property:2}, depends both on the array connectivity and the properties (symmetries) of our assignment of the $u_{i,j}$'s. A few scenarios are possible: (a) the array connectivity does not allow property~\ref{property:2}, (b) the array connectivity allows for property~\ref{property:2} but the $u_{i,j}$'s forbid it, or (c) both the array connectivity and the $u_{i,j}$'s allow for property~\ref{property:2}. We, thus, seek to establish conditions on the array connectivity guaranteeing property~\ref{property:2}. We devote the rest of the section to formally establishing these conditions and give a negative answer to the previous question in the form of a no-go theorem.

\subsection{Absence of property~\ref{property:2} in arrays with more than four resonators}
We begin by abstracting the concept of resonator array and array connectivity. A system of ring resonators with identical frequencies (taken to be circular with identical radii) coupled with strengths $u_{i,j}$ can be represented by a (weighted) graph $\mathcal{G}$ whose nodes represent rings and the edge set, $E = \{(i,j, u_{i,j})\}$ for $i,j=1,...,N$, specifies the coupled rings and their coupling strengths. The array connectivity is summarized in the graph $\mathcal{G}_{\rm UW}$, the unweighted version of $\mathcal{G}$, whose edge set, $E_{\rm UW} = \{(i,j)\enspace {\rm iff}\enspace u_{i,j}>0\}$, specifies the pairs of coupled rings.

The coupling between two ring resonators, mediated by their evanescent fields, only exists for rings in close proximity. For any two rings to have a physically meaningful coupling, \ie $u_{i,j}>0$, the separation between their edges must be much smaller than their radii. As such $\mathcal{G}_{\rm UW}$ can be interpreted as specifying the tangency relations of identical rings on the plane, \ie the nodes represent rings and the edges represent two rings ``touching'' but not crossing. In geometric graph theory this type of graph is known as a penny graph~\cite{eppstein2017triangle,sagdeev2025general,alfakih2025colin}~\footnote{A graph that specifies the tangency relations of geometric objects on the plane is known in geometric graph theory as a contact graph. Its nodes represent geometric objects and its edges correspond to two objects touching but not crossing. If the objects are circles, it is known as a coin graph, and if the circles have identical radius, it is known as a penny graph}. Importantly, any penny graph is simple, connected, and planar. 

Naturally any physical on-chip array of coupled optical microring resonators, such as those shown in Fig.~\ref{fig:fig_1}, must have a connectivity represented by a graph $\mathcal{G}_{\rm UW}$ satisfying the above key properties of penny graphs, explicitly
\begin{remark}
\label{remark:planarity}
For an array of optical microring resonators, $\mathcal{G}_{\rm UW}$ (consequently $\mathcal{G}$) must be simple, connected, and planar.
\end{remark}
By construction, the adjacency matrix of $\mathcal{G}$, $A_\mathcal{G}$, is $A_\mathcal{G} = \mathbf{u}$, where the matrix $\mathbf{u}$ has entries equal to the coupling strengths $u_{i,j}$ if resonators $i$ and $j$ are coupled with that strength, and zero otherwise.
The Hamiltonian matrix $h_0$ of the array of coupled resonators, defined via $H_0 = \mathbf{a}^\dagger.h_0.\mathbf{a}$, is $h_0 = \omega_0 \mathbf{I}_N + \mathbf{u}$. Thus the eigenspaces of $h_0$ and $A_{\mathcal{G}}$ are the same, allowing us to reason about the properties of the normal modes of $h_0$ based on the structure of $\mathcal{G}$. 

Property~\ref{property:2} asks for array normal modes with uniform support on all the resonators, which is equivalent to asking for eigenvectors of $h_0$ with uniform support on the canonical basis vectors of $\mathbb{R}^N$. We can formalize this requirement as follows. If a matrix $\mathcal{V}$ diagonalizes $h_0$, that is, $\frac{1}{N}\mathcal{V}^T h_0 \mathcal{V} = D_{\omega}$, with $D_{\omega}$ a diagonal matrix with nonzero entries equal to the normal mode frequencies $\{\omega_j\}$, the normal modes will satisfy property~\ref{property:2} only if $\mathcal{V}$ is a Hadamard matrix. We have that
\begin{remark}
\label{remark:h0_diago}
A resonator array exhibits property~\ref{property:2} only when its Hamiltonian matrix $h_0$, is diagonalized by a Hadamard matrix.
\end{remark}

A real $N\times N$ Hadamard matrix $\mathcal{H}$ has entries equal to either $\pm 1$ with the property $\mathcal{H}^T \mathcal{H} = N\mathbf{I}_N$, \ie the columns are mutually orthogonal~\footnote{Some other properties of these family of matrices are: two Hadamard matrices are equivalent if one can be produced from the other by permuting rows or columns, negating rows or columns, or some combination of these operations. A Hadamard matrix is said to be normalized if each entry in the first row and first column is equal to $+1$. Consequently, any Hadamard matrix is equivalent to a normalized Hadamard matrix.}. It is known that for a Hadamard matrix to exist, a necessary condition is that $N=1,2$ or, a multiple of $4$, whether this condition is sufficient still remains an open problem~\cite{Hedayat1978hadamard}. Consequently, 
\begin{remark}
\label{remark:array_size}
Property~\ref{property:2} can only be exhibited by resonator arrays of sizes $N=2$ or a multiple of $4$.
\end{remark}
Remarks~\ref{remark:planarity} and~\ref{remark:array_size} establish necessary conditions that $\mathcal{G}$ and $\mathcal{G}_{\rm UW}$ must satisfy for the array to exhibit property~\ref{property:2}. Obtaining conditions on $\mathcal{G}$ guaranteeing the diagonabizability of $A_\mathcal{G}$ by a Hadamard matrix is the natural next step. However, constructing these for a general weighted graph is extremely challenging. We will take an alternative path and construct those conditions for $\mathcal{G}_{\rm UW}$. Since, the eigenspaces of $A_{\mathcal{G}_{\rm UW}}$ and $A_{\mathcal{G}}$ are not necessarily the same, these will be necessary conditions, at best, and we will delineate device parameter regimes under which such equivalence can be expected. 

We begin with the following definition
\begin{definition}
An unweighted graph $\tilde{\mathcal{G}}$ is Hadamard diagonalizable if its Laplacian matrix, $L_{\tilde{\mathcal{G}}} = D_{\tilde{\mathcal{G}}} - A_{\tilde{\mathcal{G}}}$, is diagonalized by a Hadamard matrix,
\end{definition}
where $D_{\tilde{\mathcal{G}}}$, is diagonal with entries equal to the degree of each node, and $A_{\tilde{\mathcal{G}}}$ is the adjacency matrix of the graph. Two necessary conditions for a graph $\tilde{\mathcal{G}}$ to be Hadamard diagonalizable are~\cite{Barik2011,johnston2017perfect}: 
\begin{enumerate}
\item $\tilde{\mathcal{G}}$ is regular, \ie has constant degree,
\item the eigenvalues of $L_{\tilde{\mathcal{G}}}$ are all even integers.
\end{enumerate} 
Importantly, for a regular graph the eigenspace of $L_{\tilde{\mathcal{G}}}$ and $A_{\tilde{\mathcal{G}}}$ are the same, thus the Hadamard diagonabizability of the graph can be decided from the adjacency matrix directly.
Hence,
\begin{remark}
\label{remark:regularity}
In order to be Hadamard diagonalizable, $\mathcal{G}_{\rm UW}$ must be regular.
\end{remark}
We can now establish that $\mathcal{G}_{\rm UW}$ must be a Hadamard diagonalizable penny graph. This implies that $\mathcal{G}_{\rm UW}$ is simple, connected, planar, regular, and the number of nodes is a multiple of $4$. We now look at the types of graphs which are Hadamard diagonalizable and identify the subset which are penny graphs.

The characterization of Hadamard diagonalizable graphs has seen big developments in recent years, see Refs.~\cite{Barik2011,johnston2017perfect,breen2020hadamard}. Initially Ref.~\cite{Barik2011} classified all the Hadamard diagonalizable graphs up to $N=12$, later Ref.~\cite{johnston2017perfect} characterized all the graph diagonalizable by a Hadamard matrix of the Sylvester type~\footnote{These are Hadamard matrices constructed as $\mathcal{H}_{k+1} = \begin{pmatrix} \mathcal{H}_{k} & \mathcal{H}_{k} \\ \mathcal{H}_{k} & -\mathcal{H}_{k} \end{pmatrix}$ with $\mathcal{H}_{0} = 1$ and $\mathcal{H}_{1} = \begin{pmatrix} 1 & 1 \\ 1 & -1 \end{pmatrix}$.}. More recently Ref.~\cite{breen2020hadamard} classified \emph{all} the Hadamard diagonalizable graphs of order $N = 8l+4$ with $l=0,1,2,...$, they are
\begin{theorem}[theorem 2.2 of Ref.~\cite{breen2020hadamard}]
\label{theorem:hadamard_k}
For $N = 8l+4$ all the Hadamard diagonalizable graphs are: 
\begin{itemize}
\item The complete graph on $N$ nodes, $K_N$
\item The balanced complete bipartite graph on $N$ nodes, $K_{N/2,N/2}$
\item Two copies of the complete graph on $N/2$ nodes, $2K_{N/2}$
\item $N$ isolated nodes, $N K_1$
\end{itemize}
and their graph complements.
\end{theorem} 
All the graphs in the theorem are simple, but the last two graphs are not connected for any $N$, and thus do not specify valid resonator array configurations. The first two graphs are connected for any $N$, but are not planar for any $N>4$ as can be shown by Kuratowski's theorem~\cite{kuratowski1930probleme,frink1930irreducible}. For $N=4$, $K_4$ is planar but is not a penny graph~\cite{alfakih2025colin}. This can be shown by edge counting. $K_4$ has $6$ edges, penny graphs have at most $|E|\le \lfloor 3N - \sqrt{12N} - 3 \rfloor$~\cite{Harborth1974} edges, which gives $|E|<5$ for a graph on $4$ nodes. Therefore, the only valid resonator array configuration is given by $K_{2,2}$, which is the cycle graph on four nodes. Hence,
\begin{remark}
\label{remark:hadamard_diag_8k4}
Among the Hadamard diagonalizable graphs of sizes $N = 8l+4$ with $l=0,1,...$, only when $N=4$ ($l=0$), there is a valid resonator array configuration.
\end{remark} 
We consider $G_{\rm UW} = K_{2,2}$ and ask: is the corresponding $\mathcal{G}$ also Hadamard diagonalizable? Interestingly arbitrary values of $u_{i,j}$'s (arbitrary edge weights) lead to a negative answer. However, if the values $u_{i,j}$'s respect enough of the symmetries of the unweighted graph, the Hadamard diagonalizability is preserved. An example is our choice of coupling strengths, $u$ (resonators $1$ and $2$, and $3$ and $4$), and $v$ (resonators $2$ and $3$, and $1$ and $4$), in Sec.~\ref{sec:4_reso}, which lead to a weighted regular $\mathcal{G}$ with every node of degree $u+v$.

We thus arrive at our no-go theorem
\begin{theorem}
\label{theorem:no_go}
For sizes $N = 8l+4$ with $l>0$, there are no resonator arrays which exhibit property~\ref{property:2}.
\end{theorem}
\begin{proof}
This result follows from theorem~\ref{theorem:hadamard_k} and remark~\ref{remark:hadamard_diag_8k4}.
\end{proof}
As a consequence of this theorem when we natively implement a network of beam splitters with a resonator array of size $N = 8l+4$ with $l>0$, the individual beam splitters cannot be tuned to the same ratio at the same modulation amplitude. 

Since $N=8l+4=4(2l+1)$, theorem~\ref{theorem:hadamard_k} only covers half of all the multiples of $4$. While no general classification of all the Hadamard diagonalizable graphs exists at present~\footnote{This is in part a consequence of the fact that the existence of Hadamard matrices of arbitrary order $4l$ is still an open problem}, in Ref.~\cite{breen2020hadamard}, all the Hadamard diagonalizable graphs up to $N=36$ were found via direct computation. The first graph not covered by theorem~\ref{theorem:hadamard_k} appears at $N=8$. It corresponds to the cube graph (see Fig.~\ref{fig:cube_graph}a). This graph is simple, connected, planar, $3$-regular, and Hadamard diagonalizable, hence a good candidate for a valid resonator array featuring property~\ref{property:2}. Unfortunately, this is not the case. 

\begin{figure}
    \centering
    \includegraphics[width=\linewidth]{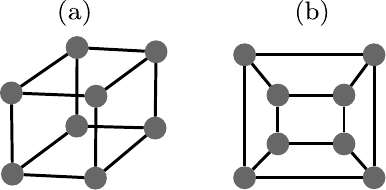}
    \caption{(a) The cube graph with $N=8$, which is the first Hadamard diagonalizable graph not covered by theorem~\ref{theorem:hadamard_k}. (b) A typical planar drawing of the cube graph. In this drawing, it can be easily appreciated that the shortest cycle in this graph is of length $4$.} 
    \label{fig:cube_graph}
\end{figure}

In fact, the planar drawing of the cube is not a penny graph. We show this by edge counting. This graph is triangle-free, \ie it does not contain any length-$3$ cycles, the shortest cycle has length $4$. Any triangle free penny graph has at most $|E| \le \lfloor 2N - 1.65\sqrt{N}\rfloor$ edges (see theorem 4 in Ref.~\cite{eppstein2017triangle}), which gives $|E| = 11$ for a triangle-free penny graph with $N=8$ nodes, and the cube graph has $12$ edges, completing the proof. This graph is typically drawn in the plane as two concentric squares with corresponding vertices connected (see Fig.~\ref{fig:cube_graph}b), if we place a (circular) ring resonator at each node, the edges between resonators on the outer square will be too far apart to have a meaningful coupling. 

For $8< N \le 36$ and $N\ne8l+1$, none of the Hadamard diagonalizable graphs listed in Ref.~\cite{breen2020hadamard} is planar. We can then safely assert that up to $N=36$, the cycle graph on four nodes is the only  valid resonator array configuration. Thus,
\begin{remark}
\label{remark:lack_of_configs}
No Hadamard diagonalizable graph with $4<N\leq 36$ constitutes a valid resonator array configuration.
\end{remark} 
Given the lack of Hadamard diagonalizable penny graphs with $N>4$ and up to $N=36$, we conjecture that remark~\ref{remark:lack_of_configs} extends to all $N$ multiples of $4$ and $N>4$.

%
%

\section{Summary, discussion, and outlook}
\label{sec:discussion}
Frequency-encoded PQC~\cite{Lukens2017} enables parallel execution of logical operations, largely reducing the hardware footprint of a quantum processing unit. Controlling the frequency degree of freedom, however, requires active control of the photon's energy, rendering linear optics costly. The success of this architecture relies on carefully balancing this trade-off in a favorable way. 

Our modeling framework represents a step forward in this direction. We used it to obtain transfer matrices for a variety of multimode frequency-domain linear optical components, based on arrays of modulated coupled microring resonators, showcasing the generality, flexibility, and composability of the framework. We anticipate our modeling toolkit and the obtained transfer matrices will aid in the modeling and rapid prototyping of the primitives required for fault tolerant linear optical quantum computing with frequency-encoded qubits. Namely, resource state generation~\cite{Bartolucci2021,Forbes2025} and entangling measurements~\cite{Browne2005,Bartolucci2021,Lingaraju2022bell,Melkozerov2024} (type-I and type-II fusions). A judicious study of these optical circuits in the frequency domain is an important future direction.

In all the studied devices the total loss of the optical element is governed by the modulation amplitude (equivalently the cooperativity $\alpha$). Thus, in a realistic implementation the ultimate limiting factor of these beam splitters will be the the RF generation capabilities which are experimentally available. Once this modulation amplitude ceiling is known, the minimum loss achievable is fixed, it can be decided what is the best way of implementing certain beam splitter. Notice that for a fixed level of loss, $\epsilon_{\rm bs}^+$ and $\epsilon_{\rm bs}^-$ require different $\alpha$, and engineering one or the other might be simpler.

Any array of modulated coupled ring resonators generates some type of multi-mode frequency beam splitter, making the space of possible optical elements too large for an exhaustive characterization. Property~\ref{property:1} and \ref{property:2}, which are desirable in the context of quantum computation, effectively narrowed this space. As shown in Sec.~\ref{sec:no_go}, these two properties are guaranteed only when the array geometry is represented by a Hadamard diagonalizable penny graph. The existence of such graphs finds a hard limit (up to a conjecture) at $N=4$. Consequently: 
\begin{enumerate}
	\item if we desire to natively implement a linear optical element composed of independent $k$-mode beam splitters, these cannot be tuned to the same splitting ratios at the same modulation amplitude,
	\item we might not find a modulation amplitude at which the array natively implements a $N$-way symmetric beam splitter.
\end{enumerate}
As such, the four-ring device from Sec.~\ref{sec:4_reso} turns out to be quite special. However, we do not rule out the existence of ring resonator arrays without either property~\ref{property:1} or property~\ref{property:2}, that implement interesting and potentially useful multi-mode frequency-domain transformations. Motivated by this possibility, in App.~\ref{app:N_reso_rectangle}, we explore the properties of $N$-mode frequency beam splitters generated with modulated rectangular arrays of $N=L\times M$ resonators.

We envision two possible paths to implement $N$-mode frequency beam splitters (with $N>4$) going beyond the scope of the no-go theorem~\footnote{Crucially the no-go theorem relies on the resonators being rings of equal radii. A third avenue to by-pass this theorem is given by arrays of resonators with different geometries. For instance, rings of mixed sizes (accounting for the fact that they must have at least one resonance at the same frequency), or other shapes. Interestingly, by mixing rings and moon-shaped resonators, the case of $N=8$ could potentially be satisfied.}. First, by introducing tailored edge weights we might promote a graph to be Hadamard diagonalizable. Second, by endowing each resonator with an internal multi-mode structure, we might access coupling graphs between internal modes which are Hadamard diagonalizable. We comment on these two scenarios below.

A scenario not covered in Sec.~\ref{sec:no_go} are weighted graphs which are Hadamard diagonalizable even when their unweighted counter part is not (see for instance example~1 in Ref.~\cite{johnston2017perfect}). Nontrivial edge weights might regularize the graph and imbue it with enough structure and symmetry to allow for Hadamard diagonalizability. To our knowledge neither a systematic classification of such graphs nor general rules for their construction exists. One is then forced to heuristically attempt their construction, we do not discard the possibility of finding graphs exhibiting property~\ref{property:2} via this method, and leave a study of this possibility for future research.

As a consequence of their geometry, microring resonators naturally host, at every resonant frequency, a degenerate pair of clockwise (CW) and counterclockwise (CCW) rotating modes. This pair of modes defines a pseudospin degree of freedom~\cite{Xu2022} for each resonator. In any coupled rings system, at every coupling point, the pseudospin necessarily flips, coupling CW and CCW modes on neighboring resonators (and vice versa). In absence of coupling between CW and CCW modes on the same resonator, any $N$-resonator array host two families of $N$-normal modes, decoupled from each other, each of which couples to a specific propagation direction of the light on the bus waveguide. 

Hence, when the CW and CCW modes on the each resonator are controllably coupled, one can promote an $N$-mode frequency beam splitter implemented with an $N$-resonator device, to a $2N$-mode frequency beam splitter, and potentially by-pass theorem~\ref{theorem:no_go}. It is know that surface roughness and imperfections cause some of the light inside the ring to backscatter, leading to an unwanted weak coupling between CW and CCW modes~\cite{Little:97,Kippenberg:02,Morichetti2010}. One can exploit this scattering process in a controllable way~\cite{Mazzei2007,Zhang:08,Lu2014selective}, to induce a desired amount of coupling between these modes. However, the output light on these larger beam splitters will, necessarily, have components on counter propagating directions. Whether such a configuration can be useful for PQC is a question for future research.     

Finally, an arbitrary $N$-mode frequency domain unitary could be generated in two  ways. On the one hand, one could compose as many $2$-mode beam splitters of variable ratios, see Eq.~(\ref{eqn:transfer_mat_R}), as necessary. However, a large hardware footprint might deem this approach not viable. On the other hand, one could generate this transformation natively as the effective transfer matrix of a modulated array of coupled resonators. One might frame this question as an optimal control problem and search for the appropriate modulation functions maximizing the fidelity with the desired unitary. We leave a complete exploration of this possibility for future research.

\acknowledgements
We acknowledge support from the Sandia National Laboratories Laboratory Directed Research and Development program (EPIQ project). This work was also supported by the U.S. Department of Energy, Office of Science, Office of Advanced Scientific Computing Research through the Accelerated Research in Quantum Computing Program MACH-Q project. Sandia National Laboratories is a multimission laboratory managed and operated by National Technology and Engineering Solutions of Sandia LLC, a wholly owned subsidiary of Honeywell International Inc. for the U.S. Department of Energy’s National Nuclear Security Administration under contract DE-NA0003525.

\bibliography{references.bib}

\newpage
\clearpage
\onecolumngrid

\appendix

%
%

\section{Asymptotic behavior of the transfer function}
\label{app:2_reso}
Now we consider the behavior of the transfer matrix in Eq.~(\ref{eqn:transfer_mat_2reso}) in the limits of vanishing drive amplitude, $\epsilon\to0$, and large drive amplitude, $\epsilon\gg\kappa$.
In the limit of a vanishing drive amplitude we can write $\Xi_{11}(\epsilon\to0)\approx 1 - \frac{2\gamma_{\rm L}}{\kappa} = \frac{\kappa_{\rm int} - \gamma_{\rm L}}{\kappa_{\rm int} + \gamma_{\rm L}} = \frac{2-\alpha}{2+\alpha}$, and $\Xi_{12}(\epsilon\to0) = \Xi_{21}(\epsilon\to0)\approx0$. We write the transfer matrix in this limit as 
\begin{equation}
\Xi(\epsilon\to0) = \left(\frac{\kappa_{\rm int} - \gamma_{\rm L}}{\kappa_{\rm int} + \gamma_{\rm L}}\right)\begin{pmatrix}
    1 && 0 \\
    0 && 1
\end{pmatrix} = \left(\frac{2-\alpha}{2+\alpha}\right)\begin{pmatrix}
    1 && 0 \\
    0 && 1
\end{pmatrix},
\end{equation}
leading to a total output intensity of $\mathcal{I}(\epsilon\to0) = \left(\frac{\kappa_{\rm int} - \gamma_{\rm L}}{\kappa_{\rm int} + \gamma_{\rm L}} \right)^2 = \left( \frac{2-\alpha}{2+\alpha}\right)^2$. 
In the limit of a large drive amplitude, $\epsilon\gg\kappa$, we can write the transfer function as 
\begin{equation}
\Xi(\epsilon\gg\kappa) \approx \begin{pmatrix}
    1 - \frac{2\gamma_{\rm L}\kappa}{\epsilon^2} && i\frac{\gamma_{\rm L}}{\epsilon}e^{i\phi} \\
    i\frac{\gamma_{\rm L}}{\epsilon}e^{-i\phi} && 1 - \frac{2\gamma_{\rm L}\kappa}{\epsilon^2}
\end{pmatrix},
\end{equation}
and thus the total output intensity is $\mathcal{I}(\epsilon\gg\kappa) \approx 1 - \frac{4\gamma_{\rm L}\kappa}{\epsilon^2} + \mathcal{O}(1/\epsilon^4)$. As such, in the limit of $\epsilon\to\infty$ the device acts as an identity on the subspace of the AS and S modes, with $\Xi(\epsilon\to\infty) \to \mathbf{I}_2$ and $\mathcal{I}(\epsilon\to\infty)\to1$.

As demonstrated in Ref.~\cite{zhang2019electronically}, in presence of the modulation, the frequency modes of the coupled ring system undergo a second order splitting, an effect similar to the Autler--Townes splitting. As a consequence, at a large enough modulation amplitude the input light will no longer be resonant with the system frequency modes, and no frequency-domain operation will take place.

To see this we analyze the RWA Hamiltonian of the two-ring device in Eq.~(\ref{eqn:rwa_hamil_2ring}), which can be written as 
\begin{equation}
\tilde{H} = \begin{pmatrix}
\tilde{c}_1^\dagger  & \tilde{c}_2^\dagger    
\end{pmatrix} 
\begin{pmatrix}
-\frac{\delta}{2} && \frac{\epsilon}{2}e^{i\phi} \\ 
\frac{\epsilon}{2}e^{-i\phi} && \frac{\delta}{2}
\end{pmatrix} 
\begin{pmatrix}
\tilde{c}_1 \\
\tilde{c}_2
\end{pmatrix},
\end{equation}
where $\delta$ is a detuning between modulation frequency and level splitting. The eigenvalues of the Hamiltonian matrix in the previous equation are $\pm\frac{1}{2}\sqrt{\epsilon^2 + \delta^2}$, which lead to normal mode frequencies $\omega_{1,2, \pm} = \omega_{1,2} \pm \frac{1}{2}\sqrt{\epsilon^2 + \delta^2}$. Thus, each frequency level splits into a pair of levels with spacing $\sqrt{\epsilon^2 + \delta^2}$. Consequently, if the input is resonant with $\omega_1$ (similarly $\omega_2$), as $\epsilon$ increases, the systems no longer features a resonance at that frequency, and the input will no longer couple to the ring system, with the latter acting as an identity element.

The previous analysis allows to derive a quantitative limit to the validity of our effective transfer matrix model. In the case of a resonant modulation ($\delta=0$), considered in the main text, the frequency levels are given by $\omega_{1,2, \pm} = \omega_{1,2} \pm \frac{\epsilon}{2}$. As $\epsilon$ increases $\omega_{1,+}$ and $\omega_{2,-}$ move towards each other, with $\omega_{2,-}-\omega_{1,+} = 2u-\epsilon$, simultaneously $\omega_{1,-}$ and $\omega_{2,+}$ move away from each other, with $\omega_{2,+}-\omega_{1,-} = 2u+\epsilon$. Incidentally, when $\epsilon = 2u$, the fast oscillating terms at $2\omega_{\rm d}=4u$ (discarded under the RWA), will couple $\omega_{2,+}$ with $\omega_{1,-}$, and will make $\omega_{2,-}$ collide with $\omega_{1,+}$, making our model invalid. We thus impose the limit $\epsilon\ll2u$ to assure model validity. Interestingly, this limit extends to systems with more than two rings, with $2u$ replaced by the level splitting between consecutive frequency levels.

\section{Robustness of some frequency beam splitters to imperfections}
\label{app:robustness}
The resonator arrays considered in the main text had two idealized properties: (a) the central frequencies are all equal to $\omega_0$, and (b) for $N>2$, the coupling strengths are equal by pairs, further, they satisfy a specific relation leading to equally spaced consecutive normal mode frequencies. In this appendix we show that some of the beam splitters discussed in the main text are robust, to first order in perturbation theory, to small deviations from these idealized configurations.

We will consider perturbations taking the system out of these idealized configurations. The perturbation Hamiltonian, $H_{\rm p} = \sum_i H_{\rm p}^{(i)}$, will be a sum of terms quadratic in the $a_j$'s, with individual terms acting on each resonator or on each coupling. As such, we can write $H_{\rm p} = \sum_j \mathbf{a}^\dagger . V_j . \mathbf{a}$, where $\mathbf{a} = (a_1,...,a_N)^T$, and $V_j$ a real symmetric matrix. We thus represent each perturbation term with its associated $V$-matrix. Similarly, we can represent the normal modes of the unperturbed system as vectors in the basis of the $a_j$'s, that is, $c_j = W_j.\mathbf{a}$. The $V_j$ matrices and $W_j$ vectors allow us to apply perturbation theory in the same way one would do it for quantum states. 

We seek to find the effect the drive terms have on the normal modes of the perturbed system. To this end we will use first order perturbation theory to rediagonalize the Hamiltonian $H_0 + H_{\rm p}$, and then write the drive Hamiltonian $H(t)$ in terms of the new normal modes. To first order in perturbation theory the corrected normal mode frequencies are
\begin{equation}
\omega_{j}^{(1)} = \omega_j + \sum_l W_l . V_l. W_l^T,
\end{equation}
where $\omega_j$ are the unperturbed normal mode frequencies. The first order corrected normal modes are 
\begin{equation}
c_{j}^{(1)} = c_j + \sum_{l\ne j} \frac{W_l . V_j . W_j^T}{\omega_l - \omega_j} c_l.
\end{equation}
We now use these expression to look at the cases of unequal central frequencies and imperfect symmetric couplings.

\subsection{Robustness of the two resonator device to unequal frequencies}
Let us now look at the expected behavior of our two-resonator device when the frequencies of the resonators are not identical. We consider a weak diagonal perturbation taking the system out of this idealized configuration, given by 
\begin{equation}
H_{\rm p} = \delta_{a_1} a_1^\dagger a_1 + \delta_{a_2}  a_2^\dagger a_2,
\end{equation}
where $|\delta_{a_{1,2}}|\ll u$ are the perturbation strengths.
We use first order perturbation theory to find the new normal modes which diagonalize the perturbed Hamiltonian $H = H_0 + H_{\rm p}$, where $H_0$ is the time independent part of Eq.~(\ref{eqn:hamil_2mode_as}). We find these normal modes to be 
\begin{equation}
c_1^{(1)} = c_1 + \frac{\delta_-}{2u}c_2, \quad
c_2^{(1)} = c_2 - \frac{\delta_-}{2u}c_1,
\end{equation}
where $\delta_- = \delta_a - \delta_b$. The corresponding first order corrected normal mode frequencies are $\omega_{1,2}^{(1)} = \omega_{1,2} + \frac{\delta_+}{2}$, with $\delta_+ = \delta_{a_1} + \delta_{a_2}$.
With these first-order corrected normal modes, we can evaluate the form the drive Hamiltonian takes when written in the new normal modes. We find
\begin{equation}
\begin{split}
&H(t) = \epsilon(t)(a_1^\dagger a_1 - a_2^\dagger a_2) \\
&= \epsilon(t)(\frac{\delta_-}{u} c_1^\dagger c_1 - \frac{\delta_-}{u} c_2^\dagger c_2 + c_1^\dagger c_2 + c_1 c_2^\dagger),
\end{split}
\end{equation}
to first order in $\frac{\delta_-}{u}$. Hence, the drive still activates the desired couplings, however the perturbations introduce an explicit time-dependence on the normal mode frequencies. These will introduce an additional time dependence on the $\Omega$ and, consequently, the $A$ matrices. However, in the picture of the slow components of $c_1$ and $c_2$, this explicit time dependence, at frequency $\omega_{\rm d}$, will be discarded, as fast oscillating terms, under the RWA. With the net effect being to decrease the scope of validity of the RWA from frequency $2\omega_{\rm d}$ to $\omega_{\rm d}$.

We have also explicitly computed the corrected normal mode frequencies and normal modes in the case of a diagonal perturbation acting on the four resonator system. Unfortunately, we find this system not to be robust against this type of imperfection. However, this does not signify a complete failure of the device, as these frequencies can be tuned by external bias voltages, and thus, if one is only slightly off of the ideal configuration, through this external tuning, the frequencies can be made to coincide.

\subsection{Robustness of the four resonator device to imperfect symmetric couplings}
In Sec.~\ref{sec:4_reso} we assume the vertical couplings (see Fig.~\ref{fig:fig_1}c) are equal to $v$ and the horizontal couplings are equal to $u$, further we took $v=2u$ to obtain equally spaced consecutive frequencies. 
Due to imperfections in the device fabrication, we do not expect to have perfectly symmetric couplings, but rather being away from it by a \emph{not too large} amount. We thus consider independent perturbations on each coupling described by the Hamiltonian
\begin{equation}
H_{\rm p} = \sum_{j=1}^4 \varepsilon_j H_j, 
\end{equation}
where $|\varepsilon_j|\ll u,v$ are the perturbation strengths and $H_j$ are Hamiltonians coupling adjacent resonators.

We employ first-order time-independent perturbation theory to find the corrected normal modes which diagonalize the perturbed Hamiltonian $H = H_0 + H_{\rm p}$, with $H_{0}$ in Eq.~(\ref{eqn:hamil_a4_usigned}), and look at the effect that the drive has on these perturbed normal modes.
The first order corrected normal mode frequencies are obtained as $\omega_{j}^{(1)} = \omega_j + \sum_l \varepsilon_l W_l . V_l. W_l^T$, and we find
\begin{subequations}
\begin{align}
\omega_1^{(1)} &= \omega_1 - \frac{1}{2}(\varepsilon_1 + \varepsilon_2 +\varepsilon_3 + \varepsilon_4), \\
\omega_2^{(1)} &= \omega_2 + \frac{1}{2}(\varepsilon_1 - \varepsilon_2 + \varepsilon_3 - \varepsilon_4), \\
\omega_3^{(1)} &= \omega_3 + \frac{1}{2}(-\varepsilon_1 + \varepsilon_2 - \varepsilon_3 + \varepsilon_4), \\
\omega_4^{(1)} &= \omega_4 + \frac{1}{2}(\varepsilon_1 + \varepsilon_2 + \varepsilon_3 + \varepsilon_4),
\end{align}
\end{subequations}
where the $\omega_j$'s are given below Eq.~(\ref{eqn:normal_modes_4reso}). Clearly, the perturbation immediately forbids us from satisfying property~\ref{property:1}.

Similarly, the first order corrected normal modes are obtained as $c_{j}^{(1)} = c_j + \sum_{l\ne j} \varepsilon_j \frac{W_l . V_j . W_j^T}{\omega_l - \omega_j} c_l$, and are given by
\begin{subequations}
\label{eqn:corrected_normal_modes_4reso}
\begin{align}
c_1^{(1)} &= c_1 + \frac{\varepsilon_3 - \varepsilon_1}{u}c_3 + \frac{\varepsilon_2 - \varepsilon_4}{v}c_2, \\
c_2^{(1)} &= c_2 + \frac{\varepsilon_4 - \varepsilon_2}{v}c_1 + \frac{\varepsilon_1-\varepsilon_3}{u}c_4, \\
c_3^{(1)} &= c_3 + \frac{\varepsilon_1 - \varepsilon_3}{u}c_1 + \frac{\varepsilon_4 - \varepsilon_2}{v}c_4, \\
c_4^{(1)} &= c_4 + \frac{\varepsilon_3 - \varepsilon_1}{u}c_2 + \frac{\varepsilon_2 - \varepsilon_4}{v}c_3
\end{align}
\end{subequations}
with the $c_j$'s, the unperturbed normal modes, given in Eq.~(\ref{eqn:normal_modes_4reso}). These corrected normal modes diagonalize $H = H_0 + H_{\rm p}$, and we write 
$ H^{(1)} = \sum_{j} \omega_j^{(1)} c_j^{(1)\dagger} c_j^{(1)}$,
where we have neglected terms $O(\varepsilon_j^2)$. We ask now: what types of beam splitters are activated once we introduce the time-dependent drive? To evaluate this, we write the $a_j$'s operators in terms of the new normal modes and re-sum the drive terms with the sign choices leading to the three patterns $\mathcal{P}_{1,2,3}$.
Importantly, we find pattern $\mathcal{P}_1 = \{(1,4),(2,3)\}$ to be robust to first order in $\frac{\varepsilon_1 - \varepsilon_3}{u}$ and $\frac{\varepsilon_2 - \varepsilon_4}{v}$. This can be intuitively understood from the form of Eq.~(\ref{eqn:corrected_normal_modes_4reso}). If we send $c_1\leftrightarrow c_4$ and $c_2\leftrightarrow c_3$, this in turn makes the corrected normal modes transform as $c_1^{(1)}\leftrightarrow c_4^{(1)}$ and $c_2^{(1)}\leftrightarrow c_3^{(1)}$, and hence the first order robustness. However, for patterns $\mathcal{P}_2$ and $\mathcal{P}_3$, we do not find this same first order robustness.

Explicit computation of the drive terms for $\mathcal{P}_2$ yields
\begin{equation}
\begin{split}
H^{(1)}(t) &= \epsilon(t)(c_1^{(1)\dagger} c_2^{(1)} + c_3^{(1)\dagger} c_4^{(1)} + {\rm h.c.}) \\
&+ 2\epsilon(t) d_2\left(\sum_{j=1,4} c_j^{(1)\dagger} c_j^{(1)} - \sum_{j=2,3} c_j^{(1)\dagger} c_j^{(1)}\right) \\
&+ 2\epsilon(t)d_1(c_1^{(1)\dagger} c_4^{(1)} - c_2^{(1)\dagger} c_3^{(1)} + {\rm h.c.}),
\end{split}
\end{equation}
where $d_1 = \frac{\varepsilon_3 - \varepsilon_1}{u}$ and $d_2 = \frac{\varepsilon_2 - \varepsilon_4}{v}$. Thus, we obtain the desired couplings according to $\mathcal{P}_2$, but also unwanted diagonal terms and activate the couplings corresponding to $\mathcal{P}_1$ with strength proportional to $d_1$. 
Similarly, by explicit computation of the drive terms for $\mathcal{P}_3$ we find
\begin{equation}
\begin{split}
H^{(1)}(t) &= \epsilon(t)(c_1^{(1)\dagger} c_3^{(1)} + c_2^{(1)\dagger} c_4^{(1)} + {\rm h.c.}) \\
&+ 2\epsilon(t) d_1\left(\sum_{j=1,4} c_j^{(1)\dagger} c_j^{(1)} - \sum_{j=2,3} c_j^{(1)\dagger} c_j^{(1)}\right) \\
&+ 2\epsilon(t)d_2(c_1^{(1)\dagger} c_4^{(1)} - c_2^{(1)\dagger} c_3^{(1)} + {\rm h.c.}).
\end{split}
\end{equation}
Again, we obtain the desired couplings together with unwanted diagonal terms and the couplings corresponding to $\mathcal{P}_1$ at strength $d_2$. Together, these two expressions inform us that small local errors will translate into highly nonlocal errors in the mixing of the normal modes, and render the beam splitter networks defined by patterns $\mathcal{P}_2$ and $\mathcal{P}_3$ highly fragile to fabrication imperfections.

%
%

\section{Construction of effective transfer matrices for arbit      rary arrays of resonators}
\label{app:general_transfer_mat}
In this appendix we present the detailed derivation of the effective transfer matrices for the $N$-mode beam splitters generated with arrays of $N$-resonators.
Our starting point is the free Hamiltonian of the normal modes of the system (see  Sec.~\ref{sec:no_go} for details), given by
\begin{equation}
H_0 = \sum_{k=1}^N \omega_k c_k^\dagger c_k
\end{equation}
with $c_k$ ($c_k^\dagger$) the annihilation (creation) operator for the $k$-th normal mode, and $\omega_k$ its corresponding frequency.
To implement frequency beam splitters we need to establish couplings between the discrete frequency levels, \ie the normal modes. This is done with a time-dependent modulation acting on some, or all, of the resonators, which we choose to have the general form
\begin{equation}
\label{eqn:hamil_drive}
H_{\rm d}(t) = \epsilon(t)\sum_{j} f_{j} a_{j}^\dagger a_{j},
\end{equation}
where the coefficients $f_{j} \in \{0,\pm1\}$ and we leave the explicit form of $\epsilon(t)$ unspecified for the moment. As stated in Sec.~\ref{sec:no_go}, our choice of $f_j$'s dictates the way in which the normal modes couple. In the basis of the normal modes, the total time-dependent Hamiltonian now reads
\begin{equation}
\label{eqn:hamil_total}
\begin{split}
H(t) &= H_0 + H_{\rm d}(t) \\
&= \sum_{k=1}^N \omega_k c_k^\dagger c_k + \epsilon(t)\sum_{(j,j')\in\mathcal{P}} \left( c^\dagger_jc_{j'} + {\rm h.c.} \right),
\end{split}
\end{equation}
where $\mathcal{P}$ is a ``pattern'' of coupled normal modes. The specific relation between the $f_j$'s and $\mathcal{P}$ depends on the way the normal modes are supported on the resonators, and thus varies from array to array. Explicit expressions for a rectangular array are given in App.~\ref{app:N_reso_rectangle}. The pattern of couplings allows us to activate either an $N$-way symmetric frequency beam splitter, a pair of $N/2$-way symmetric beam splitters, and so on, all the way down to a network of $N/2$ nonoverlapping $2$-mode frequency beam splitters. The four ring device studied in Sec.~\ref{sec:4_reso} is an example of this freedom. We take the first resonator, $a_1$, to be strongly coupled to a waveguide at a rate $\Gamma_{\rm L}$, and all resonators suffer internal losses at a finite rate $\kappa_{\rm int}$. We also consider a setup where a second resonator is strongly coupled to a second waveguide at rate $\Gamma_{\rm R}$, and we will require the output on this second waveguide to be split by a desired ratio. These waveguides are the input-outputs ports to the ring resonator system and each of the waveguides will carry multiple frequency modes.

Since Eq.~(\ref{eqn:hamil_total}) is linear in the normal mode operators, we can write the equations of motion for the $c_j$'s and the input-output relation in the $ABCD$ representation, albeit with a time-dependent $A$ matrix (consequence of the time dependent Hamiltonian), which precludes an immediate solution via Fourier transform (see Sec.~\ref{sec:background_linear_optics}). Nevertheless, assuming well defined and separated frequency levels and under the RWA, we can recover a time-independent system and obtain an effective transfer matrix.

As a first step in this process we need to specify both the forms of the waveguide input $b_{\rm in}(t)$ and the modulation $\epsilon(t)$. The frequency content of the input is fixed by the type of $N$-mode linear optical element we wish to implement with our $N$-resonator device, which is encoded in the structure of $\mathcal{P}$. We ask: how many independent (nonoverlapping) subsets of coupled normal modes can we split $\mathcal{P}$ into? This number corresponds exactly with the number of distinct frequencies the input must carry. It can be as small as $1$, which corresponds to the case of a $N$-way symmetric beam splitter, for which the input frequency should be nearly resonant with any of the $\omega_j$'s. And as large as $N/2$ (for $N$ even), corresponding to a network of $N/2$ nonoverlapping $2$-mode frequency beam splitters. With each input frequency component nearly resonant with one of the two levels being mixed on each $2$-mode beam splitter. We write this input field as
\begin{equation}
\label{eqn:physical_input}
b_{\rm in}(t) = b_{\rm in}\sum_{j=1}^{N_\omega} e^{-i \omega_{\rm I}^{(j)} t},
\end{equation}
where $1\le N_{\omega}\le N/2$ is the number of distinct frequency components in the input, $b_{\rm in}$ the amplitude of the input light, which we assume to be equal for all frequencies, and  $\omega_{\rm I}^{(k)}$ the $k$-th frequency component of the input.

Our choice of equally spaced consecutive normal mode frequencies, i.e, property~\ref{property:1}, fixes the total number of different splittings among the $N$ normal modes to be $N-1$. Thus, we might need a drive with at most $N-1$ distinct tones, each of them nearly resonant with one of the splittings, and at least $1$ tone. In fact, a pattern of coupled levels $\mathcal{P}$ defines a set of splittings, and the number of required tones is given by the number of different splittings in that set. We define the multi-tone drive
\begin{equation}
\label{eqn:general_drive}
\epsilon(t) = \epsilon \sum_{j=1}^{J}\cos(\omega_{\rm d}^{(j)}t + \phi_j),
\end{equation}
where $1 \le J \le N-1$ is the number of distinct tones $\{\omega_{\rm d}^{(j)}\}_{j=1,..,J}$, where for each tone the drive is a pure harmonic, $\epsilon$ is the drive amplitude, and  $\omega_{\rm d}^{(j)}$ and $\phi_j$ are the frequency and phase of the $j$-th tone, respectively.
With these definitions, we now move to the construction of effective transfer matrices for our frequency-domain optical elements. In the following we will treat the case of one and two waveguides separately.

\subsubsection{The case of a single waveguide}
A single waveguide means a single \emph{physical} input-output port. The SLH triple of the system, in terms of the $a_j$ operators is $(\mathbf{I}, [\sqrt{\Gamma_{\rm R}} a_1], H(t))$ with $H(t) = H_0 + H_{\rm d}(t)$ written in terms of the $a_j$'s. Since our frequency beam splitters mix the normal modes, the starting point for the construction of the effective transfer matrix is to write the system's SLH triple in terms of the $c_j$'s.
In general, for an array of coupled resonators, $a_1 = \sum_{j} \mu_j c_j$, with $0<\mu_j<1$ for all $j$, and $\sum_j \mu_j^2 = 1$. If we introduce $\tau_j = \mu_j^2$, we find $\mathbf{L} = \sum_j \sqrt{\Gamma_{\rm L}\tau_j} c_j$. Hence, the normal modes exchange energy with the propagating fields on the waveguide as a collective, \ie they drive and are driven as a collective, with each normal mode contributing to this collective energy exchange at a rate of $\gamma_{\rm L}^{(j)} = \Gamma_{\rm L}\tau_j$. Motivated by this observation, we will treat the single physical input-output port as $N$ \emph{virtual} input-output ports (although the inputs on different virtual ports would not be completely independent). This leads to an $\mathbf{L}$ of the form 
\begin{equation}
\label{eqn:L_nornal_modes}
\mathbf{L} = \left(\sqrt{\gamma_{\rm L}^{(1)}}c_1, ..., \sqrt{\gamma_{\rm L}^{(N)}}c_{N} \right)^T.
\end{equation}
For a given virtual input port feeding one of the normal modes, the input field is given by Eq.~(\ref{eqn:physical_input}). Out of all the frequency components either none or only one will be nearly resonant with the normal mode frequency and thus will effectively couple to it. Thus, the effective input to the $i$-th virtual input-output port is either $b_{\rm in}^{(i)}(t) = b_{\rm in}e^{-i \omega_{\rm I}^{(k)}}$ or $b_{\rm in}^{(i)}(t) = 0$, depending on whether there is one frequency component, $\omega_{\rm I}^{(k)}$, in the input, nearly resonant with $\omega_i$. As such, the system is then described by an effective SLH triple $(\mathbf{I}_N, \mathbf{L}, H(t))$ with $\mathbf{L}$ given in Eq.~(\ref{eqn:L_nornal_modes}) and $H(t)$ given in Eq.~(\ref{eqn:hamil_total}).

The associated $\Omega$ and $\Phi$ matrices are defined via their entries as
\begin{equation}
\label{eqn:omega_mat}
[\Omega]_{ij} = \begin{cases}
\omega_i \quad \text{if}\quad i=j, \\
\epsilon(t) \enspace \text{if} \enspace (i,j)\in\mathcal{P}\enspace\text{or}\enspace (j,i)\in\mathcal{P}, \\
0 \quad \text{otherwise}
\end{cases}
\end{equation}
and 
\begin{equation}
[\Phi]_{ij} = \begin{cases}
\sqrt{\gamma_{\rm L}^{(i)}} \enspace\text{if}\enspace i=j, \\
0 \enspace\text{otherwise}
\end{cases}
\end{equation}
where $\epsilon(t)$ is of the form of Eq.~(\ref{eqn:general_drive}), and $\gamma_{\rm L}^{(i)} = \Gamma_{\rm L}\tau_{i}$ as before. The $ABCD$ matrices then take the form
\begin{equation}
[A]_{ij} = \begin{cases}
-i\omega_i -\gamma_{\rm L}^{(i)}/2 - \kappa_{\rm int}/2 \quad \text{if}\quad i=j, \\
-i\epsilon(t) \enspace \text{if} \enspace (i,j)\in\mathcal{P}\enspace\text{or}\enspace (j,i)\in\mathcal{P}, \\
0 \quad \text{otherwise}
\end{cases}
\end{equation}
where internal losses for each normal mode were added phenomenologically~\footnote{This can be done by including the non-Hermitian term $-i\frac{\kappa_{\rm int}}{2}\sum_{j=1}^N a^\dagger_j a_j = -i\frac{\kappa_{\rm int}}{2}\sum_{j=1}^N c^\dagger_j c_j$ in the total Hamiltonian. If all the physical resonators experience the same amount of internal losses, so do the normal modes.}, and $B = -\Phi$, $C = \Phi$, $D = \mathbf{I}_{N}$. This results in equations of motion of the form
\begin{subequations}
\label{eqn:abcd_linear_nonauto}
\begin{align}
\dot{\mathbf{c}}(t) &= A(t) \mathbf{c}(t) + B \mathbf{b}_{\rm in}(t), \\
\mathbf{b}_{\rm out}(t) &= C \mathbf{c}(t) + D \mathbf{b}_{\rm in}(t),
\end{align}
\end{subequations}
where $\mathbf{c}(t) = (c_1(t),...,c_{N}(t))^T$ is the vector of annihilation operators for the normal modes.

The structure of $\mathcal{P}$ allows us to extract a time-independent system out of Eq.~(\ref{eqn:abcd_linear_nonauto}). Indeed, $\mathcal{P}$ is composed of clusters of coupled normal modes. The number of clusters goes from $1$ (one $N$-way symmetric frequency beam splitter) to $N/2$ ($N/2$ $2$-mode frequency beam splitters), and it is equal to the number of distinct frequencies in the input field, Eq.~(\ref{eqn:physical_input}). By construction each of these frequencies, $\omega_{\rm I}^{(k)}$, is nearly resonant with a normal mode, say $c_k$, at $\omega_k$, and all the normal modes coupled to $c_k$ (in the same cluster), are almost frequency-matched to $\omega_{\rm I}^{(k)}$ by one of the drive frequencies. We now extract the slowly varying components of the normal modes. For a given cluster, we write $c_k\to c_ke^{-i\omega_{\rm I}^{(k)}t}$ for the normal mode nearly resonant with $\omega_{\rm I}^{(k)}$, and $c_j\to c_j e^{-i(\omega_{\rm I}^{(k)} + \omega_{\rm d}^{(k')})t}$ for all the other normal modes in the cluster, where $\omega_{\rm d}^{(k')}$ is the drive frequency which almost frequency-matches $\omega_{\rm I}^{(k)}$ to $\omega_j$. 

In these rotating frames, invoking the RWA and after neglecting fast oscillating terms at $\pm 2\omega_{\rm d}^{(k')}$ or higher, we recover a time-independent equation of motion for the slow components of the normal modes. Further, the effective inputs on the slow components of the normal modes in the cluster under consideration are $b_{\rm in}^{(i)}(t) \to b_{\rm in}$ and $b_{\rm in}^{(j)}(t) = b_{\rm in}e^{i \omega_{\rm d}^{(k')} t}$, the latter being highly off-resonant from the normal mode frequency and will not couple to it, and thus we neglect its effect.
These approximations lead to equations of motion of the form 
\begin{subequations}
\label{eqn:abcd_linear_slow}
\begin{align}
\dot{\mathbf{c}}(t) &= (A + W)\mathbf{c}(t) + B \mathbf{b}_{\rm in}(t), \\
\mathbf{b}_{\rm out}(t) &= C \mathbf{c}(t) + D \mathbf{b}_{\rm in}(t),
\end{align}
\end{subequations}
with an time-independent $A$ matrix of the form
\begin{equation}
\label{eqn:amat_time_inde}
[A]_{ij} = \begin{cases}
-i\omega_i - \kappa_i/2 \quad \text{if}\quad i=j, \\
-i\frac{\epsilon}{2}e^{i\phi_{k'}} \enspace \text{if} \enspace (i,j)\in\mathcal{P}, \\
-i\frac{\epsilon}{2}e^{-i\phi_{k'}} \enspace \text{if} \enspace (j,i)\in\mathcal{P}, \\
0 \quad \text{otherwise}
\end{cases}
\end{equation}
where we have introduced the total level linewidth of the $i$-th normal mode, $\kappa_i = \gamma_{\rm L}^{(i)} + \kappa_{\rm int}$, and $\phi_{k'}$ is the phase of the $k'$-th frequency component of the drive. The $W$ matrix is given by 
\begin{equation}
\label{eqn:Wmat}
[W]_{lm} = \begin{cases}
i\omega_{\rm I}^{(k)}\enspace \text{if}\enspace l=m \enspace\text{and}\enspace l=i\enspace \text{with}\enspace (i,j)\in\mathcal{P}, \\
i(\omega_{\rm I}^{(k)} + \omega_{\rm d}^{(k')}) \enspace\text{if}\enspace l=m, \enspace l=j\enspace \text{with}\enspace (i,j)\in\mathcal{P}, \\
0 \quad\text{otherwise}.
\end{cases}
\end{equation}
This time invariant system can be solved via direct Fourier transform. The input-output relation is then
\begin{equation}
\label{eqn:transfer_func_linear}
\mathbf{b}_{\rm out}(\omega) = \Xi(\omega) \mathbf{b}_{\rm int}(\omega),
\end{equation}
where the frequency-domain transfer function is given by
\begin{equation}
\label{eqn:freq_domain_transfer}
\Xi(\omega) = \left[ C( i\omega \mathbf{I}_N - (A + W))^{-1} B + D \right].
\end{equation}

\subsubsection{Example: Network of nonoverlapping \texorpdfstring{$2$}{\textit{2}}-mode beam splitters}
Let us consider a pattern $\mathcal{P}$ of nonoverlaping pairs of coupled normal modes. For this type of coupling pattern, the transfer matrix inherits the block structure of the $A$ matrix, which is composed of $(2\times 2)$ blocks. With each block defining a $2$-mode frequency beam splitter mixing normal modes $i$ and $j$ for $(i,j)\in\mathcal{P}$. 

The ideal energy exchange occurs when the input light is resonant with one of the two levels, \ie $\omega_{\rm I}^{(k)}=\omega_i$ or $\omega_{\rm I}^{(k)}=\omega_j$, and one of the frequencies in the drive is resonant with the level splitting, \ie $\omega_{\rm d}^{(k')}=\Delta_{ij}=\omega_{j}-\omega_i$ for $(i,j)\in\mathcal{P}$. In this setting the diagonal components of the matrix $A+W$ no longer depend on the frequencies involved. Finally, notice that Eq.~(\ref{eqn:abcd_linear_slow}) describe the dynamics of the slow components, as such we are interested in the limit $\omega\to 0$. In this limit, with the ideal selection of frequencies, we find each of the $(2\times 2)$ blocks of $\Xi(\omega\to 0) = \Xi(\epsilon, \kappa_i, \kappa_j)$ corresponding to the coupled levels $(i,j)\in\mathcal{P}$ to be 
\begin{equation}
\label{eqn:transfer_mat_block}
\Xi_{ij}(\epsilon, \kappa_i, \kappa_j) = \begin{pmatrix}
1 - \frac{2\kappa_j\gamma_{\rm L}^{(i)}}{\epsilon^2 + \kappa_i\kappa_j} && i \frac{2\epsilon\sqrt{\gamma_{\rm L}^{(i)}\gamma_{\rm L}^{(j)}}}{\epsilon^2 + \kappa_i\kappa_j} e^{i\phi_{k}} \\ 
i \frac{2\epsilon\sqrt{\gamma_{\rm L}^{(i)}\gamma_{\rm L}^{(j)}}}{\epsilon^2 + \kappa_i\kappa_j} e^{-i\phi_{k}}  && 1 - \frac{2\kappa_i\gamma_{\rm L}^{(j)}}{\epsilon^2 + \kappa_i\kappa_j}
\end{pmatrix},
\end{equation} 
where $\kappa_i = \gamma_{\rm L}^{(i)} + \kappa_{\rm int}$ is the total linewidth of the $i$-th normal mode. The total output on the waveguide is 
\begin{equation}
b_{\rm out}(t) = \sum_{(i,j)\in\mathcal{P}} \left[\left( 1 - \frac{2\kappa_j\gamma_{\rm L}^{(i)}}{\epsilon^2 + \kappa_i\kappa_j} \right) e^{-i\omega_i t}  + i \frac{2\epsilon\sqrt{\gamma_{\rm L}^{(i)}\gamma_{\rm L}^{(j)}}}{\epsilon^2 + \kappa_i\kappa_j} e^{i\phi_{k}}e^{-i\omega_j t}\right],
\end{equation}
where $\phi_k$ is the phase of the frequency component of the drive which is resonant with the splitting $\Delta_{ij}$. 

\subsubsection{The case of two waveguides}
Extracting the output through the same waveguide used to inject the input might lead to distortion of the output photonic wavepacket. While this might not be a cause of serious issues in a setting aimed at demonstrating the frequency-domain beam splitting capabilities, when these devices are integrated as part of a larger QION with the aim of executing a complex computation, this distortion will diminish our capability of coupling the resulting output photons to other localized elements down the line. A workaround to this issue is to have a second waveguide dedicated only to extract the output of our optical element.
Consider the array of coupled resonators now coupled to two waveguides. The first waveguide couples to the array via the first resonator at rate $\Gamma_{\rm L}$, and the second waveguide couples to the array via a second resonator, different from $a_1$, at rate $\Gamma_{\rm R}$. We now have two physical input-output ports and the system, in terms of the $a_j$'s, this is described by the SLH triple $(\mathbf{I}_{4}, [\sqrt{\Gamma_{\rm L}} a_1, \sqrt{\Gamma_{\rm R}} a_N]^T, H(t))$, where without loss of generality we have taken the second waveguide to be coupled to the last resonator, $a_N$, and $H(t) = H_0 + H_{\rm d}(t)$ the device Hamiltonian in terms of the $a_j$'s. We now seek the SLH representation in terms of the $c_j$'s.  

The normal modes exchange energy with the first waveguide collectively, via $a_1$, with each normal mode contributing to this energy exchange at a rate $\gamma_{\rm L}^{(i)} = \Gamma_{\rm L} \tau_i$. Similarly, the normal modes exchange energy with the second waveguide as a collective. Namely, one can expand $a_N = \sum_{j}g_j \tilde{\nu}_j c_j$, where $\tilde{\nu}_j = |\nu_j|$ and $g_j = {\rm sign}(\nu_j)$, then after introducing $\chi_j = \tilde{\nu}_j^2$, we write $\sqrt{\Gamma_{\rm R}}a_N = \sum_{j}g_j \sqrt{\gamma_{\rm R}^{(j)}} c_j$. Thus, each normal mode contributes to this collective energy exchange with the second waveguide at a rate $\gamma_{\rm R}^{(j)} = \Gamma_{R}\chi_j$, which is also mediated by the sign $g_j=\pm1$.
We promote each of the two physical input-output ports to $N$ virtual input output ports (whose inputs are not completely independent), and ascribe two of each of them to the each normal mode. We write 
\begin{equation}
\label{eqn:L_nornal_modes_2wv}
\mathbf{L} = \left( \sqrt{\gamma_{\rm L}^{(1)}}c_1, g_1\sqrt{\gamma_{\rm R}^{(1)}}c_1, ..., \sqrt{\gamma_{\rm L}^{(N)}}c_{N}, g_N\sqrt{\gamma_{\rm R}^{(N)}}c_{N} \right)^T,
\end{equation}
which now has size $(2N\times 1)$. 
The second waveguide is only used to extract the output, its input is always vacuum, $b_{\rm in}^{(\rm R)} = 0$, while the input on the first waveguide, $b_{\rm in}^{(\rm L)}$, is given by Eq.~(\ref{eqn:physical_input}). As such, for a given normal mode, there is either one or zero frequency component in the input which is nearly resonant with it. Then, the effective input on the virtual input port for the first waveguide corresponding to the $i$-th normal mode is $b_{\rm in}^{(i, {\rm L})}(t) = b_{\rm in}e^{-i \omega_{\rm I}^{(k)}t}$ or $b_{\rm in}^{(i, {\rm L})}(t) = 0$, depending on whether there is a frequency component, $\omega_{\rm I}^{(k)}$, in the physical input nearly resonant with $\omega_i$. The system is then described by the effective SLH triple $(\mathbf{I}_{2N}, \mathbf{L}, H(t))$ with $\mathbf{L}$ as in Eq.~(\ref{eqn:L_nornal_modes_2wv}) and $H(t)$ as in Eq.~(\ref{eqn:hamil_total}).

The auxiliary matrices $\Omega$, of size $(N\times N)$, and $\Phi$, of size $(2N\times N)$, are defined by their entries and given by Eq.~(\ref{eqn:omega_mat}), and 
\begin{equation}
\label{eqn:phi_mat_2wv}
[\Phi]_{p,q} = \begin{cases}
\sqrt{\gamma_{\rm L}^{(p)}} \enspace\text{if}\enspace p = (q-1)N + 1, \\
g_p \sqrt{\gamma_{\rm R}^{(p)}} \enspace\text{if}\enspace p = (q-1)N + 2, \\
0 \quad \text{otherwise}
\end{cases}
\end{equation} 
respectively. Correspondingly the $A(t),B,C,D$ matrices are given by
\begin{equation}
\label{eqn:amat_2wv}
[A]_{ij} = \begin{cases}
-i\omega_i - \kappa_i/2 \quad \text{if}\quad i=j, \\
-i\epsilon(t) \enspace \text{if} \enspace (i,j)\in\mathcal{P}\enspace\text{or}\enspace (j,i)\in\mathcal{P}, \\
0 \quad \text{otherwise}
\end{cases}
\end{equation}
where we defined the total linewidth $\kappa_i = \gamma_{\rm L}^{(i)} + \gamma_{\rm R}^{(i)} + \kappa_{\rm int}$ for the $i$-th normal mode, and introduced the internal loss of each mode phenomenologically. And $B = -\Phi^T$, $C = \Phi$, $D = \mathbf{I}_{2N}$. Introducing the second waveguide explicitly changed the form of the $B,C,$ and $D$ matrices, but in general terms the $A(t)$ matrix retained the same structure as in the case of a single waveguide. This leads to two important consequences: (1) we can exploit the same set of arguments, invoking the RWA, as in the previous section to arrive at a time-independent form for $A(t)$ acting on the slow components of the normal modes. (2) This time-independent form of the $A$ matrix will have the same structure as that in Eq.~(\ref{eqn:amat_time_inde}), with the appropriate definition of $\kappa_i$ for the case of two waveguides. The solution to the input-output relation for the case of two waveguides can be written as in Eq.~(\ref{eqn:transfer_func_linear}), with the transfer matrix defined as in Eq.~(\ref{eqn:freq_domain_transfer}), with the $B,C$, and $D$ matrices as defined below Eq.~(\ref{eqn:amat_2wv}) and the $W$ matrix as defined in Eq.~(\ref{eqn:Wmat}).

\subsubsection{Example: Network of nonoverlapping \texorpdfstring{$2$}{\textit{2}}-mode beam splitters}
Let us consider a pattern $\mathcal{P}$ of nonoverlaping pairs of coupled normal modes. In this case, the effective transfer matrix has a simple block structure with each block being $(4\times 4)$. 
For a pair of coupled normal modes $(i,j)\in\mathcal{P}$, with the input on the first waveguide resonant with one of the two levels and the drive carrying a component resonant with the splitting $\Delta_{ij}$, we can write the $(4\times 4)$ block representing the mixing between these two levels as 
\begin{equation}
\label{eqn:transfer_mat_block_2wv}
\Xi_{ij}(\epsilon, \kappa_i, \kappa_j) = \\
 \begin{pmatrix}
1 - \frac{2\kappa_j \gamma_{\rm L}^{(i)}}{\kappa_i\kappa_j + \epsilon^2}  & -g_i \frac{2\kappa_j \sqrt{\gamma_{\rm L}^{(i)}\gamma_{\rm R}^{(i)}}}{\kappa_i\kappa_j + \epsilon^2} & i\frac{2\epsilon \sqrt{\gamma_{\rm L}^{(i)} \gamma_{\rm L}^{(j)}}}{\kappa_i\kappa_j + \epsilon^2} e^{i\phi_k} & ig_j\frac{2\epsilon \sqrt{\gamma_{\rm L}^{(i)}\gamma_{\rm R}^{(j)}}}{\kappa_i\kappa_j + \epsilon^2}e^{i\phi_k} \\
-g_i \frac{2\kappa_j \sqrt{\gamma_{\rm L}^{(i)}\gamma_{\rm R}^{(i)}}}{\kappa_i\kappa_j + \epsilon^2} & 1 - \frac{2\kappa_j \gamma_{\rm R}^{(i)}}{\kappa_i\kappa_j + \epsilon^2} & ig_i\frac{2\epsilon \sqrt{\gamma_{\rm L}^{(j)}\gamma_{\rm R}^{(i)}}}{\kappa_i\kappa_j + \epsilon^2} e^{i\phi_k} & ig_ig_j\frac{2\epsilon \sqrt{\gamma_{\rm R}^{(i)}\gamma_{\rm R}^{(j)}}}{\kappa_i\kappa_j + \epsilon^2}e^{i\phi_k} \\
i\frac{2\epsilon \sqrt{\gamma_{\rm L}^{(i)} \gamma_{\rm L}^{(j)}}}{\kappa_i\kappa_j + \epsilon^2} e^{-i\phi_k} & ig_i\frac{2\epsilon \sqrt{\gamma_{\rm L}^{(j)} \gamma_{\rm R}^{(i)}}}{\kappa_i\kappa_j + \epsilon^2} e^{-i\phi_k} & 1 - \frac{2\kappa_i \gamma_{\rm L}^{(j)}}{\kappa_i\kappa_j + \epsilon^2}  & -g_j \frac{2\kappa_i \sqrt{\gamma_{\rm L}^{(j)} \gamma_{\rm R}^{(j)}}}{\kappa_i\kappa_j + \epsilon^2} \\
ig_j\frac{2\epsilon \sqrt{\gamma_{\rm L}^{(i)} \gamma_{\rm R}^{(j)}}}{\kappa_i\kappa_j + \epsilon^2} e^{-i\phi_k} & ig_ig_j\frac{2\epsilon \sqrt{\gamma_{\rm R}^{(i)} \gamma_{\rm R}^{(j)}}}{\kappa_i\kappa_j + \epsilon^2} e^{-i\phi_k} & -g_j \frac{2\kappa_i \sqrt{\gamma_{\rm L}^{(j)} \gamma_{\rm R}^{(j)} }}{\kappa_i\kappa_j + \epsilon^2} & 1 - \frac{2\kappa_i \gamma_{\rm R}^{(j)}}{\kappa_i\kappa_j + \epsilon^2}
\end{pmatrix},
\end{equation}
where $\phi_k$ is the phase of the drive frequency component resonant with the splitting $\Delta_{ij} = \omega_j - \omega_i$ for $(i,j)\in\mathcal{P}$.
The total output of the device is then
\begin{equation}
\label{eqn:physical_outpu_2wv}
b_{\rm out}(t) = b_{\rm out}^{(\rm L)}(t) + b_{\rm out}^{(\rm R)}(t),
\end{equation}
where the outputs on each waveguide are 
\begin{equation}
 b_{\rm out}^{(\rm L)}(t) = \sum_{(i,j)\in\mathcal{P}} \left[\left( 1 - \frac{2\kappa_j \gamma_{\rm L}^{(i)}}{\kappa_i\kappa_j + \epsilon^2} \right) e^{-i\omega_i t} 
 + i\frac{2\epsilon \sqrt{\gamma_{\rm L}^{(i)} \gamma_{\rm L}^{(j)}}}{\kappa_i\kappa_j + \epsilon^2} e^{-i\phi_k} e^{-i\omega_j t}\right],
\end{equation}
and
\begin{equation}
 b_{out}^{(\rm R)}(t) = \sum_{(i,j)\in\mathcal{P}} \left[ -g_i \frac{2\kappa_j \sqrt{\gamma_{\rm L}^{(i)}\gamma_{\rm R}^{(i)}}}{\kappa_i\kappa_j + \epsilon^2}e^{-i\omega_i t} + ig_j\frac{2\epsilon \sqrt{\gamma_{\rm L}^{(i)} \gamma_{\rm R}^{(j)}}}{\kappa_i\kappa_j + \epsilon^2} e^{-i\phi_k} e^{-i\omega_j t} \right].
\end{equation}

%
%

\section{Beam splitters implementable with rectangular arrays of coupled resonators}
\label{app:N_reso_rectangle}
In this section we investigate the properties of the frequency beam splittters that can be generated with a rectangular array of strongly coupled $N = L\times M$ resonators of equal central frequency $\omega_0$, and coupling strengths $u$ and $v$, along the horizontal and vertical directions of the array, respectively. For the rest of the section we will index the position of a resonator in the array with the tuple $(l,m)$, similar to how one indexes the entries of a matrix. 

\begin{figure}[ht!]
    \centering
    \includegraphics[width=0.5\columnwidth]{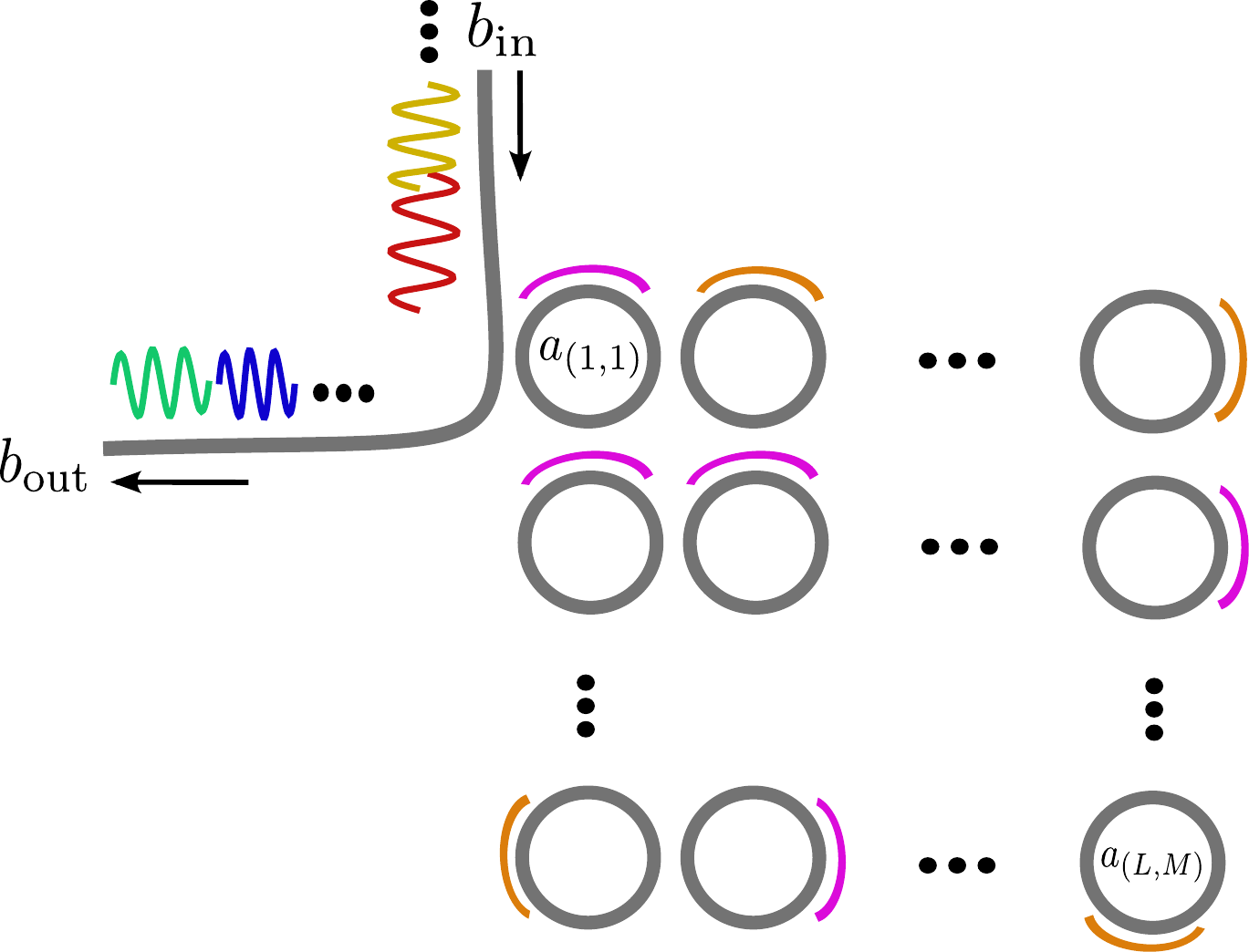}
    \caption{Schematic of the rectangular array of ring resonators investigated in this appendix. We show the case of a single waveguide coupled to the array via $a_{(1,1)}$ at rate $\Gamma_{\rm L}$.}
    \label{fig:fig_rectangle}
\end{figure}

In absence of a time-dependent modulation, the array of resonators is described by the Hamiltonian
\begin{equation}
\label{eqn:hamil_as_nmode}
H_0 = \omega_0\sum_{l=1,m=1}^{L,M}a^\dagger_{(l,m)}a_{(l,m)} 
+ u\sum_{l=1}^L\sum_{m=1}^{M-1}\left( a^\dagger_{(l,m)}a_{(l,m+1)} + {\rm h.c.} \right)
+ v \sum_{m=1}^M \sum_{l=1}^{L-1}\left( a^\dagger_{(l,m)}a_{(l+1,m)} + {\rm h.c.}\right),
\end{equation}
where $a_{(l,m)}$ is the anihilation operator for the resonator at position $(l,m)$ in the array. The Hamiltonian in Eq.~(\ref{eqn:hamil_as_nmode}) is quadratic in the creation/annihilation operators, defining $\mathbf{a} = (a_{(1,1)},...,a_{(L,M)})^T$ the vector of all the annihilation operators, we can write Eq.~(\ref{eqn:hamil_as_nmode}) as
\begin{equation}
H_0 = \mathbf{a}^\dagger . h_0 . \mathbf{a}
\end{equation}
for some real symmetric matrix $h_0$. Explicitly finding the normal modes and their frequencies requires us to diagonalize $h_0$.
Since the rectangular array has ``hard'' boundaries, the normal modes are given by standing waves~\cite{McDonald2018}, which we can define with the help of a discrete sine transform~\cite{Britanak2007}. They are given by
\begin{equation}
\label{eqn:normal_modes_discrete_sines}
c_{(p,q)} = \sum_{l,m=1}^{L,M} \mu_{(p,q)}^{(l,m)} a_{(l,m)}, \quad\text{and} \quad a_{(l,m)} = \sum_{p,q=1}^{L,M} \mu_{(p,q)}^{(l,m)} c_{(p,q)}
\end{equation}
with the coefficients given by products of discrete sines
\begin{equation}
\mu_{(p,q)}^{(l,m)} = \frac{2 \sin\left(\frac{\pi l}{L+1} p \right) \sin\left(\frac{\pi m}{M+1} q \right) }{\sqrt{(L+1)(M+1)}}.
\end{equation}
The orthonormality of the normal modes $c_{(p,q)}$ follow from the orthogonality of the discrete sines~\cite{Britanak2007}. One can verify by direct computation that the above set of normal modes diagonalize $H_0$ in Eq.~(\ref{eqn:hamil_as_nmode}). Indeed, we find
\begin{equation}
\label{eqn:hamil_cs_nmode}
H_0 = \sum_{p=1,q=1}^{L,M} \omega_{(p,q)} c^\dagger_{(p,q)} c_{(p,q)},
\end{equation}
with $c_{(p,q)}$ as in Eq.~(\ref{eqn:normal_modes_discrete_sines}) and their frequencies given by
\begin{equation}
\label{eqn:normal_mode_freqs_general}
\omega_{(p,q)} = \omega_0 + 2v\cos\left(\frac{\pi}{L+1}p \right) + 2u\cos\left( \frac{\pi}{M+1}q \right),
\end{equation}
which for $p = 1,...,L$ and $q=1,...,M$ lists the normal mode frequencies in descending order.

\subsection{Equally spaced normal mode frequencies}
In this subsection we investigate the conditions on the coupling strengths of the rectangular array to guarantee property~\ref{property:1}, \ie equally spaced consecutive normal mode frequencies.
Consider the $\omega_{(p,q)}$'s of Eq.~(\ref{eqn:normal_mode_freqs_general}) and sort them in ascending order. We construct the set of level spacings, $\{\Delta_{(p,q),(p',q')}\}$, for $(p,q)$ and $(p',q')$ such that $\omega_{(p,q)}$ and $\omega_{(p',q')}$ are consecutive and $\omega_{(p,q)}<\omega_{(p',q')}$. We obtain conditions on the coupling strengths $u$ and $v$ guaranteeing property~\ref{property:1}. Importantly, if the coupling strengths are parametrized by two values only, $u$ and $v$, then, for a general array of resonators one cannot make all the $N-1$ spacings equal. In fact, as we increase the size of the rectangular array, one encounters a hard limit at $N=9=3\times 3$.

This can be readily seen by examining the condition stemming from requiring equality between the first two frequency spacings. That is, $\Delta_{(L,M),(L,M-1)} = \Delta_{(L,M-1),(L,M-2)}$. Using Eq.~(\ref{eqn:normal_mode_freqs_general}), we find
\begin{equation}
\cos\left( \frac{\pi M}{M+1} \right) - \cos\left( \frac{\pi (M-1)}{M+1}  \right) 
= \cos\left( \frac{\pi (M-1)}{M+1} \right) - \cos\left( \frac{\pi (M-2)}{M+1} \right),
\end{equation}
which is independent of the coupling strengths, valid for $M\ge3$, and holds true when $M=3$. If we fix $M=3$, then by requiring $\Delta_{(L,1),(L-1,3)} = \Delta_{(L-1,1),(L-2,3)}$, we obtain a similar condition on $L$ which is satisfied when $L=3$. We summarize these observations in the following remarks,
\begin{remark}
\label{remark:equal_spacings}
If we desire normal mode frequencies which are equally spaced, any of the two sides of the array cannot have more than three resonators. 
\end{remark}
\begin{remark}
\label{remark:equal_spacings_2}
The largest rectangular array for which we can find values of $u$ and $v$ leading to equally spaced consecutive frequency levels is $N=9$.
\end{remark}
Incidentally, remark~\ref{remark:equal_spacings_2} includes the case of $N=8$, but remark~\ref{remark:equal_spacings} excludes it, resulting in only four array sizes for which we can tune the coupling strengths to obtain equally spaced consecutive energy levels. We collect these conditions in Table~\ref{tab:conditions_uv}. Finally, if we allow all the coupling strengths to be different, $u_{(l,m),(l',m')}$ where $(l,m)$ and $(l',m')$ are neighbors on the lattice, a simple numerical optimization of these values allows us to recover property~\ref{property:1} for larger array sizes.
\begin{table}
\begin{tabular}{|c|c|c|c|}
\hline
$N$ & $L$ & $M$ & Condition \\
\hline 
$4$ & $2$ & $2$ & $v=2u$ \\
\hline
$6$ & $2$ & $3$ & $v = 3u/\sqrt{2}$ \\
\hline
$6$ & $3$ & $2$ & $u = 3u/\sqrt{2}$ \\
\hline
$9$ & $3$ & $3$ & $v=3u$ \\
\hline
\end{tabular}
\caption{Conditions enabling an $N=L\times M$ array satisfying property~\ref{property:1}.}
\label{tab:conditions_uv}
\end{table}

\subsection{The structure of the coupling rates to the waveguides}
The array couples to a first waveguide via $a_{(1,1)}$ at rate $\Gamma_{\rm L}$, and to a second waveguide via $a_{(L,M)}$ at rate $\Gamma_{\rm R}$. As discussed in App.~\ref{app:general_transfer_mat}, the coupling rates to the waveguides, $\gamma_{\rm L,R}^{(p,q)}$, of each normal mode depend on the support that $a_{(1,1)}$ and $a_{(L,M)}$ have on the $c_{(p,q)}$'s. We write, $\sqrt{\Gamma_{\rm L}} a_{(1,1)} = \sum_{p,q} \sqrt{\gamma_{\rm L}^{(p,q)}} c_{(p,q)}$, with $\gamma_{\rm L}^{(p,q)} = \Gamma_{\rm L} \left(\mu_{(p,q)}^{(1,1)}\right)^2$, and the square root is always well define since $\mu_{(p,q)}^{(1,1)}\ge0$ for all $p,q\le L,M$. And $\sqrt{\Gamma_{\rm R}} a_{(L,M)} = \sum_{p,q} g_{(p,q)} \sqrt{\gamma_{\rm R}^{(p,q)}} c_{(p,q)}$, with $g_{(p,q)} = {\rm sign}(\mu_{(p,q)}^{(L,M)})$ and $\gamma_{\rm R}^{(p,q)} = \Gamma_{\rm R} \left(\tilde{\mu}_{(p,q)}^{(L,M)}\right)^2$, where $\tilde{\mu}_{(p,q)}^{(L,M)} = |\mu_{(p,q)}^{(L,M)}|$.

The structure of the coupling rates, $\gamma_{\rm L,R}^{(p,q)}$'s, as defined above, is established in the following remark
\begin{remark}
\label{remark:coupling_rates}
The coupling rates $\gamma_{\rm L,R}^{(p,q)}$ come in groups of four. In fact, we have 
\begin{equation}
\begin{split}
\gamma_{\rm L}^{(p,q)} = \gamma_{\rm L}^{(L-(p-1),q)} = \gamma_{\rm L}^{(p,M-(q-1))} = \gamma_{\rm L}^{(L-(p-1),M-(q-1))},\nonumber \\
\gamma_{\rm R}^{(p,q)} = \gamma_{\rm R}^{(L-(p-1),q)} = \gamma_{\rm R}^{(p,M-(q-1))} = \gamma_{\rm R}^{(L-(p-1),M-(q-1))}.
\end{split}
\end{equation}
\end{remark}
This is direct consequence of identical chains of equalities for the $\mu_{(p,q)}^{(1,1)}$ and $\tilde{\mu}_{(p,q)}^{(L,M)}$, which follow from the identity: $\sin\left(\frac{\pi l}{L+1} (L - (p-1)) \right) = (-1)^{l+1}\sin\left(\frac{\pi l}{L+1} p \right)$, and similarly for $m,q$ and $M$. Finally, it follows from remark~\ref{remark:coupling_rates}, that the largest array of resonators for which the normal modes couple with equal rates to the waveguides is $N=4$. 
This structure also limits the patterns of nonoverlapping couplings which can be implemented using drives of the type in Eq.~(\ref{eqn:hamil_drive}), which we discuss next.

\subsection{N-mode nonoverlapping coupling patterns}
Consider a time-dependent drive of the type in Eq.~(\ref{eqn:hamil_drive}), with the indexing used in this section, we write
\begin{equation}
\label{eqn:hamil_drive_nmode}
H(t) = \epsilon(t)\sum_{l=1,m=1}^{L,M} f_{(l,m)} a_{(l,m)}^\dagger a_{(l,m)},
\end{equation}
where the coefficients $f_{(l,m)} \in \{0,\pm1\}$. Introducing this drive to the system Hamiltonian $H_0$ activates couplings between the normal modes. We are interested in patterns of coupled normal modes which do not overlap, generalizing the three patterns, $\mathcal{P}_1$, $\mathcal{P}_2$, and $\mathcal{P}_3$, defined for $N=4$ in Sec.~\ref{sec:4_reso}. 
The structure of the $\mu_{(p,q)}^{(l,m)}$'s implied by remark~\ref{remark:coupling_rates} informs us of the type of nonoverlapping $N$-mode patterns which can be accessed. 
They are constructed by subdividing the $N$ modes in groups of four, with the four modes within a group having the same value of $\mu_{(p,q)}^{(1,1)}$. Given a group of four, the normal modes are coupled as dictated by $\mathcal{P}_1$, $\mathcal{P}_2$, or $\mathcal{P}_3$, and the pattern is completed repeating this same pairing for all the groups. Since $N$ is even, this covers all the normal modes or leaves two of them not belonging to any group, which are consequently paired together. Finally, the chain of equalities in remark~\ref{remark:coupling_rates} can be split into three equalities, one relating $(p,q)$ to $(p-(L-1), q-(M-1))$, one relating it to $(p, q-(M-1))$, and one relating it to $(p-(L-1),q)$. This thus defines the normal modes each of the patterns, constructed before, will be mixing, and we provide the choice of drive signs achieving each of them below.

The normal mode frequencies are listed in ascending order with $\omega_0$ at the center. The pattern $\mathcal{P}_1$ for $N$ modes will couple pairs of normal modes which are the same distance away from $\omega_0$. This is achieved with the choice of signs $f_{(l,m)} = (-1)^{l+m}$, we find
\begin{equation}
H(t) = \epsilon(t)\sum_{l=1,m=1}^{L,M} (-1)^{l+m} a_{(l,m)}^\dagger a_{(l,m)} 
= \epsilon(t)\sum_{p,q=1}^{L,M} c^\dagger_{(L-(p-1), M-(q-1))} c_{(p,q)} + {\rm h.c.},
\end{equation} 
critically, this pattern requires us to choose coupling strengths $u$ and $v$ that do not lead to degenerate normal modes.
The generalized pattern $\mathcal{P}_2$ is achieved with the choice of signs $f_{(l,m)} = (-1)^{m+1}$, we find
\begin{equation}
H(t) = \epsilon(t)\sum_{l=1,m=1}^{L,M} (-1)^{m+1} a_{(l,m)}^\dagger a_{(l,m)} 
= \epsilon(t)\sum_{p,q=1}^{L,M} c^\dagger_{(p,M-(q-1))} c_{(p,q)} + {\rm h.c.}.
\end{equation} 
The generalized pattern $\mathcal{P}_3$ is achieved with the choice of signs $f_{(l,m)} = (-1)^{l+1}$, we find 
\begin{equation}
H(t) = \epsilon(t)\sum_{l=1,m=1}^{L,M} (-1)^{m+1} a_{(l,m)}^\dagger a_{(l,m)} 
= \epsilon(t)\sum_{p,q=1}^{L,M} c^\dagger_{(L-(p-1),q)} c_{(p,q)} + {\rm h.c.}.
\end{equation} 
Before closing this section we note a shortcoming of the generalizations of patterns $\mathcal{P}_2$ and $\mathcal{P}_3$.
\begin{remark}
\label{remark:rectangular_array_patterns}
Arrays where one side contains an odd number of resonators lead to drive Hamiltonians that activate unwanted diagonal terms. Thus, patterns $\mathcal{P}_2$ and $\mathcal{P}_3$ might not be generalized.
\end{remark}
In pattern $\mathcal{P}_2$, mode $(p,M-(q-1))$ is connected to mode $(p,q)$, thus one will obtain a diagonal term whenever $q = \frac{M+1}{2}$ is an integer in $[0,M]$. When $M$ is even, $\frac{M+1}{2}$ is never an integer, however, when $M$ is odd, the above situation is unavoidable. Thus, for any rectangular array with $M$ odd, pattern $\mathcal{P}_2$ cannot be implemented in general. In pattern $\mathcal{P}_3$, mode $(L-(p-1), q)$ is connected to mode $(p,q)$. One obtains diagonal terms whenever $p = \frac{L+1}{2}$ is an integer in $[1,L]$. When $L$ is even, $\frac{L+1}{2}$ is never an integer, however, when $L$ is odd, this situation is unavoidable. Thus, for any rectangular array with $L$ odd, pattern $\mathcal{P}_3$ cannot be implemented in general. Finally, we point out that the generalization of pattern $\mathcal{P}_1$ can \emph{always} be implemented regardless of the array dimensions.  

\subsection{Transfer matrices for the rectangular array of ring resonators}
In this section we apply our methodology for the construction of effective transfer matrices to the beam splitters generated with a rectangular array of resonators. As discussed in Sec.~\ref{sec:no_go}, even though these transfer matrices can be obtained with our methodology, the lack of property~\ref{property:2}, implies that tuning all the beam splitters to the same ratio at the same modulation amplitude cannot be achieved. We will illustrate this with an example.

We will focus on a rectangular array of $N=2\times 3$ ring resonators, fixing the coupling strengths as to have equally spaced consecutive levels (see Table~\ref{tab:conditions_uv}). We will consider only the case of a single waveguide, and look at the case of generating a network of three nonoverlapping beam splitters by activating any of the three patterns discussed in the previous section.
To better understand how the different beam splitters in the network will behave at a given modulation amplitude, we need the effective coupling rates of the normal modes. To construct them, we first find the system normal modes. They are 
\begin{equation}
\begin{split}
c_{(1,1)} &= \frac{1}{\sqrt{8}}(1,-\sqrt{2},	1,-1, \sqrt{2}, -1).\mathbf{a}, \quad
c_{(1,2)} = \frac{1}{\sqrt{8}}(\sqrt{2},0,-\sqrt{2},-\sqrt{2}, 0, \sqrt{2}).\mathbf{a}, \\ 
c_{(1,3)} &= \frac{1}{\sqrt{8}}(1,\sqrt{2},	1,-1, -\sqrt{2}, -1).\mathbf{a}, \quad
c_{(2,1)} = \frac{1}{\sqrt{8}}(1,-\sqrt{2},	1,1, -\sqrt{2}, 1).\mathbf{a},   \\
c_{(2,2)} &= \frac{1}{\sqrt{8}}(\sqrt{2},0,-\sqrt{2},\sqrt{2}, 0, -\sqrt{2}).\mathbf{a},  \quad
c_{(2,3)} = \frac{1}{\sqrt{8}}(1,\sqrt{2},	1,1, \sqrt{2}, 1).\mathbf{a},   
\end{split}
\end{equation}
where $\mathbf{a} = (a_{(1,1)}, a_{(1,2)}, a_{(1,3)}, a_{(2,1)}, a_{(2,2)}, a_{(2,3)})^T$ is a vector of the resonator's annihilation operators, and the associated frequencies are $\omega_{(1,1)} = \omega_0 - \frac{5}{\sqrt{2}}u$, $\omega_{(1,2)} = \omega_0 - \frac{3}{\sqrt{2}}u$, $\omega_{(1,3)} = \omega_0 - \frac{1}{\sqrt{2}}u$, $\omega_{(2,1)} = \omega_0 + \frac{1}{\sqrt{2}}u$, $\omega_{(2,2)} = \omega_0 + \frac{3}{\sqrt{2}}u$, and $\omega_{(2,3)} = \omega_0 + \frac{5}{\sqrt{2}}u$. 
From the structure of the normal modes we have $a_{(1,1)} = \frac{1}{\sqrt{8}}(1,\sqrt{2},1,1,\sqrt{2},1).\mathbf{c}$, where $\mathbf{c}$ is a vector of the normal mode annihilation operators. 
Since the array couples to the waveguide via $a_{(1,1)}$, there are only two different effective coupling rates: $\gamma_1 = \frac{\Gamma}{8}$ corresponding to normal modes in $O_1 = \{(1,1), (1,3), (2,1), (2,3)\}$ and $\gamma_2 = \frac{\Gamma}{4}$ corresponding to normal modes in $O_2 = \{(1,2),(2,2)\}$, where $\Gamma$ is the coupling rate to the waveguide, and $\gamma_2 = 2\gamma_1$.  

Consider the case of connecting the four normal modes $\{(1,1), (1,3), (2,1), (2,3)\}$ as dictated by $\mathcal{P}_3$, that is, mode $(1,1)$ with $(2,1)$, and mode $(1,3)$ with $(2,3)$, naturally the pattern for the six modes is completed by connecting $(1,2)$ with $(2,2)$. With the modes coupled in this way there is only one level splitting, $\Delta = \frac{4}{\sqrt{2}}u$, and thus we only need a monochromatic modulation of the form $\epsilon(t) = \epsilon\cos(\omega_{\rm d} t + \phi)$ with $\omega_{\rm d}$ resonant with $\Delta$, and an input field which carries three distinct frequencies, each resonant with one of the normal modes being mixed on each of the three nonoverlapping $2$-mode beam splitters.  
The effective transfer matrix for this $6$-mode beam splitter composed of three nonoverlapping $2$-mode beam splitters can readily be constructed as a block-diagonal matrix (see App.~\ref{app:general_transfer_mat} for details), with each $(2\times 2)$ block having the form in Eq.~(\ref{eqn:transfer_mat_R}), where the transmission amplitude parameter is modified to 
\begin{equation}
\label{eqn:transmission_amp_appendix}
\mathcal{K}_\pm(\alpha, R, \mu^{(1,1)}) = \frac{\sqrt{1-R} \pm \sqrt{(\mu^{(1,1)})^4 \alpha^2 - R}}{\left( (\mu^{(1,1)})^2 \alpha + 1 \right)},
\end{equation}
where $\alpha = \frac{\Gamma}{\kappa_{\rm int}}$ is the cooperativity parameter of the array, and $\mu^{(1,1)}$ is the support that $a_{(1,1)}$ has on the two normal modes being mixed in the beam splitter. For the present example, we have $\mu^{(1,1)} = \frac{1}{\sqrt{8}}$ for those normal modes in $O_1$, and $\mu^{(1,1)} = \frac{1}{2}$ for those normal modes in $O_2$.
The beam splitters connecting normal modes with $\gamma_1$ have ratio $R_1$, and the beam splitter connecting the two normal modes with $\gamma_2$ has ratio $R_2$, and $R_1\ne R_2$. 

Let us look at the particular example of using the device to implement frequency swaps on the four coupled normal modes, that is $R_1 = 0$, which happens at the drive amplitude $\epsilon = \sqrt{\gamma_1^2 - \kappa_{\rm int}}$. At this modulation amplitude, we find the third beam splitter to have ratio $R_2 = \frac{9}{25}$, and the full transfer matrix is
\begin{equation}
\Xi(\alpha) = \begin{pmatrix}
0 & 0 & 0 & ie^{i\phi}\mathcal{K}(\alpha, 0, \frac{1}{\sqrt{8}}) & 0 & 0 \\
0 & -\frac{3}{5}\mathcal{K}_-(\alpha, \frac{9}{25}, \frac{1}{4}) & 0 & 0 & i\frac{4}{5}e^{i\phi}\mathcal{K}_-(\alpha, \frac{9}{25}, \frac{1}{4}) & 0 \\
0 & 0 & 0 & 0 & 0 & ie^{i\phi}\mathcal{K}(\alpha, 0, \frac{1}{\sqrt{8}})  \\
ie^{-i\phi}\mathcal{K}(\alpha, 0, \frac{1}{\sqrt{8}})  & 0 & 0 & 0 & 0 & 0 \\
0 & i\frac{4}{5}e^{i\phi}\mathcal{K}_-(\alpha, \frac{9}{25}, \frac{1}{4}) & 0 & 0 & -\frac{3}{5}\mathcal{K}_-(\alpha, \frac{9}{25}, \frac{1}{4}) & 0 \\
0 & 0 & ie^{-i\phi}\mathcal{K}(\alpha, 0, \frac{1}{\sqrt{8}}) & 0 & 0 & 0 \\
\end{pmatrix}
\end{equation}
with the transmission amplitude parameter $\mathcal{K}_\pm$ defined in Eq.~(\ref{eqn:transmission_amp_appendix}). We see then, that the methodology developed in this work is fully general and flexible, and allows us to obtain effective transfer matrices for arbitrary resonator arrays, even when they do not satisfy property~\ref{property:2}.

\subsection{Robustness to imperfect coupling strengths}
Following App.~\ref{app:robustness}, in this section we investigate the robustness of the generalizations of patterns $\mathcal{P}_{1,2,3}$ to imperfections in the symmetric coupling configuration. 
We consider independent perturbations to each of the horizontal and vertical couplings in the Hamiltonian of Eq.~(\ref{eqn:hamil_as_nmode}), with the total perturbation Hamiltonian given by
\begin{equation}
\label{eqn:hamil_pertu}
H_{\rm pt} = \sum_{l,m} \varepsilon^{(l,m)}_{\rm h}H_{\rm h}^{(l,m)} + \sum_{l,m}\varepsilon^{(l,m)}_{\rm v}H_{\rm v}^{(l,m)},
\end{equation}
where $\varepsilon^{(l,m)}_{\rm h,v}$, with $|\varepsilon^{(l,m)}_{\rm h,v}|\ll u,v$, are the independent perturbation strengths for the horizontal and vertical coupling strengths, respectively. The Hamiltonians $H_{\rm h,v}^{(l,m)}$ are given by
\begin{equation}
\begin{split}
H_{\rm h}^{(l,m)} &= a_{(l,m)}^\dagger a_{(l,m+1)} + {\rm h.c.},\\
H_{\rm v}^{(l,m)} &= a_{(l,m)}^\dagger a_{(l+1,m)} + {\rm h.c.}.
\end{split}
\end{equation}

Recall the way we proceeded in App.~\ref{app:robustness}. Using first order perturbation theory we find the corrected normal modes $\tilde{c}_{(p,q)}$ and their corresponding frequencies. With these at hand, we rediagonalize the Hamiltonian $H_0 + H_{\rm pt}$. Then, writing the $a_{(l,m)}$ in terms of the corrected normal modes we re-sum the drive terms leading to each of the three coupling patterns $\mathcal{P}_{1,2,3}$, then conclude about their robustness. 
Although one could explicitly construct the first order corrected normal modes and find the action of the drive Hamiltonian on these corrected modes---given the structure of the coupling rates discussed previously, and the structure of the generalized patterns---they inherit the robustness or lack of it from the coupling patterns of the array with $N=4$, which we discussed in App.~\ref{app:robustness}. This immediately allows us to establish the following two remarks

\begin{remark}
The set of beam splitters described by the Hamiltonian $H_{\rm bs} = \sum_{p,q} c^\dagger_{(L-(p-1), M-(q-1))} c_{(p,q)} + {\rm h.c.}$, corresponding to pattern $\mathcal{P}_1$, whose action, is to connect the normal modes as  
\begin{subequations}
\begin{align}
(p,q)\leftrightarrow (L-(p-1), M-(q-1)) \nonumber
\end{align}
\end{subequations}
acts in the same manner on the first-order corrected normal modes $\tilde{c}_{(p,q)}$.
\end{remark}
Intuitively, this robustness is a consequence of the way we constructed $\mathcal{P}_1$ for $N$-modes. We know from Sec.~\ref{sec:4_reso} that for $N=4$ the beam splitters in the pattern $\mathcal{P}_1$ are robust to imperfections in the coupling strengths. In the case of $N$-modes, there is a natural subdivision of the modes in groups of four, with each group defined by the strength with which they couple to the waveguide. And we generalize the pattern by coupling the modes within a group as dictated by $\mathcal{P}_1$, thus, the two beam splitters within a group are robust, and this is true for all groups, leading to the robustness of the whole pattern for the $N$-modes.

\begin{remark}
Patterns $\mathcal{P}_2$ and $\mathcal{P}_3$, and the beam splitter Hamiltonians they define, are not robust to imperfections in the coupling strengths.
\end{remark}
In fact, the perturbations activate unwanted diagonal terms and additional couplings which follow the structure of pattern $\mathcal{P}_1$.

\end{document}